\documentclass[aps,prx,floatfix,nofootinbib,superscriptaddress,twocolumn,longbibliography]{revtex4-2}
\usepackage{hyperref}
\usepackage{graphics}

\makeatother

\usepackage{amsmath,mathtools,amsthm,amssymb}
\usepackage{tabularx}
\usepackage{units}
\usepackage{complexity}
\usepackage{graphicx}
\usepackage{color}
\usepackage{xcolor}
\usepackage{bbold}
\usepackage{lipsum} 
\usepackage{titlesec}
\usepackage{enumitem}
\usepackage{pifont}
\usepackage{times}
\usepackage{comment}
\usepackage[linesnumbered,ruled,vlined]{algorithm2e}
\usepackage{listings}
\def\squiggly{\bgroup \markoverwith{\textcolor{red}{\lower3.5\p@\hbox{\sixly \char58}}}\ULon}

\newcommand{\geqt}[1]{\stackrel{\mathclap{\scriptsize \mbox{#1}}}{\geq}}

 \newcommand{\be}{\begin{equation}}
\newcommand{\ee}{\end{equation}}
\newcommand{\ba}{\begin{eqnarray}}
\newcommand{\epsstat}{\varepsilon_{\mathrm{stat}}}
\newcommand{\epsT}{\varepsilon_{\mathrm{T}}}
\newcommand{\ea}{\end{eqnarray}}
\usepackage{array}
\newcolumntype{H}{>{\setbox0=\hbox\bgroup}c<{\egroup}@{}}
\usepackage{graphics}
\usepackage[normalem]{ulem}

\usepackage{complexity}

\PassOptionsToPackage{pdftex,hyperfootnotes=false,pdfpagelabels}{hyperref}

\hypersetup{
    colorlinks = true,
    linkcolor = [rgb]{0.15, 0.35, 0.55},  % Dark blue for links
    citecolor = [rgb]{0.2, 0.5, 0.3},     % Muted green for citations
    urlcolor = [rgb]{0.7, 0.2, 0.2}       % Soft red for URLs
}
\usepackage{libertine}
\usepackage[toc,page]{appendix}

\usepackage{braket}
\usepackage{physics}
\usepackage{csquotes}

\newtheorem{theorem}{Theorem}
\newtheorem{proposition}{Proposition}
\newtheorem{definition}{Definition}
\newtheorem{lemma}{Lemma}

\newtheorem{remark}{Remark}
\newtheorem{problem}{Problem}

\newcommand{\hs}[2]{\langle #1, #2\rangle_{HS}}
\newcommand{\Pf}{\operatorname{Pf}}

\newcommand{\fu}{Dahlem Center for Complex Quantum Systems, Freie Universit\"{a}t Berlin, 14195 Berlin, Germany}

\newcommand{\nocontentsline}[3]{}
\let\origcontentsline\addcontentsline
\newcommand\stoptoc{\let\addcontentsline\nocontentsline}
\newcommand\resumetoc{\let\addcontentsline\origcontentsline}

\begin{document}

%\title{Testing free-fermionic quantum states and improved tomography}
\title{Optimal trace-distance bounds for free-fermionic states: \\ Testing and improved tomography}

%\title{Estimates of trace distance between free-fermionic states:\\ Property testing and improved tomography}

%\date{}

\author{Lennart Bittel}
\thanks{\{\href{mailto:l.bittel@fu-berlin.de}{l.bittel}, \href{mailto:a.mele@fu-berlin.de}{a.mele}, \href{mailto:jense@fu-berlin.de}{jense}, \href{mailto:lorenzo.leone@fu-berlin.de}{lorenzo.leone}\}@fu-berlin.de}
\affiliation{\fu}

\author{Antonio Anna Mele}
\thanks{\{\href{mailto:l.bittel@fu-berlin.de}{l.bittel}, \href{mailto:a.mele@fu-berlin.de}{a.mele}, \href{mailto:jense@fu-berlin.de}{jense}, \href{mailto:lorenzo.leone@fu-berlin.de}{lorenzo.leone}\}@fu-berlin.de}
\affiliation{\fu}
\affiliation{Theoretical Division, Los Alamos National Laboratory, Los Alamos, New Mexico 87545, USA}

\author{Jens Eisert}
\thanks{\{\href{mailto:l.bittel@fu-berlin.de}{l.bittel}, \href{mailto:a.mele@fu-berlin.de}{a.mele}, \href{mailto:jense@fu-berlin.de}{jense}, \href{mailto:lorenzo.leone@fu-berlin.de}{lorenzo.leone}\}@fu-berlin.de}
\affiliation{\fu}

\author{Lorenzo Leone}
\thanks{\{\href{mailto:l.bittel@fu-berlin.de}{l.bittel}, \href{mailto:a.mele@fu-berlin.de}{a.mele}, \href{mailto:jense@fu-berlin.de}{jense}, \href{mailto:lorenzo.leone@fu-berlin.de}{lorenzo.leone}\}@fu-berlin.de}
\affiliation{\fu}

\begin{abstract}
Free-fermionic states, also known as fermionic Gaussian states, represent an important class of quantum states ubiquitous in physics. They are uniquely and efficiently described by their correlation matrix. However, in practical experiments, the correlation matrix can only be estimated with finite accuracy. This raises the question: how does the error in estimating the correlation matrix affect the trace-distance error of the state? We show that if the correlation matrix is known with an error \(\varepsilon\), the trace-distance error also scales as \(\varepsilon\) (and vice versa). 
Specifically, we provide distance bounds between (both pure and mixed) free-fermionic states in relation to their correlation matrix distance. Our analysis also extends to cases where one state may not be free-fermionic. Importantly, we leverage our preceding results to derive significant advancements in property testing and tomography of free-fermionic states. Property testing involves determining whether an unknown state is close to or far from being a free-fermionic state. We first demonstrate that any algorithm capable of testing arbitrary (possibly mixed) free-fermionic states would inevitably be inefficient, implying that there is no efficient strategy to estimate the non-Gaussianity of a state. Then, we present an efficient algorithm for testing low-rank free-fermionic states. For free-fermionic state tomography, we provide improved bounds on sample complexity in the pure-state scenario, substantially improving over previous literature, and we generalize the efficient algorithm to mixed states, discussing its noise-robustness.
\end{abstract}

\stoptoc
 \maketitle
\vspace{-0.5cm}
\section{Introduction}
\vspace{-0.2cm}
As the construction of quantum devices progresses, there is a growing emphasis on developing efficient learning schemes to extract key diagnostic information from quantum systems~\cite{PRXQuantum.2.010201,Eisert_2020}. The tasks of quantum certification and benchmarking~\cite{Eisert_2020} are crucial in any effort that aims to manipulating or preparing quantum states to utmost precision -- and hence to achieve predictive power of a sort.
Extracting information from quantum systems is generally a challenging task,
marred by obstructions in sample and computational complexity~\cite{Haah_2017,anshu2023survey},
but real-world scenarios often defy general no-go results. %States prepared on current quantum devices commonly feature a lot of structure; they possess specific properties and symmetries, diverging from general states which oftentimes can be seen largely as abstractions. 

Within the realm of fermionic quantum computation, free-fermionic states~\cite{Surace_2022} (also known as fermionic Gaussian states, or states prepared by one-dimensional matchgate circuits~\cite{knill2001fermionic,Terhal_2002}) play a crucial role. These states are ubiquitous in various physics domains, ranging from condensed matter~\cite{CMfree}, where they typically arise in ``non-interacting'' settings, to the study of analog quantum simulators with ultra-cold fermionic atoms~\cite{jordens_mott_2008,GrossFermions}, and quantum chemistry~\cite{chemistry}. What makes free-fermionic states important for quantum computation is that they belong to a non-trivial physically interesting class of efficiently classically simulable states~\cite{knill2001fermionic,Terhal_2002,Jozsa_2008}.

While these states may be considered ineffective for advantageous quantum computation, they serve as an essential milestone in the ongoing construction of fault-tolerant quantum devices: Their efficiency in classical simulation~\cite{knill2001fermionic,Terhal_2002,Jozsa_2008,cudby2023Gaussian,dias2023classical,reardonsmith2024improved} and learning~\cite{Gluza_2018,ogorman2022fermionic,aaronson2023efficient,mele2024efficient} provides a powerful tool for benchmarking quantum computation to ensure the correct functionality of quantum chips. This situation is not dissimilar from that of stabilizers and matrix product states~\cite{gottesman1998heisenberg,montanaro2017learning,grewal2023efficient,leone2023learning,hangleiter2024bell,cramer_efficient_2010,fanizza2023learning,huang2024learning}.

 \begin{figure}
    \centering
\includegraphics[width=.49\textwidth]{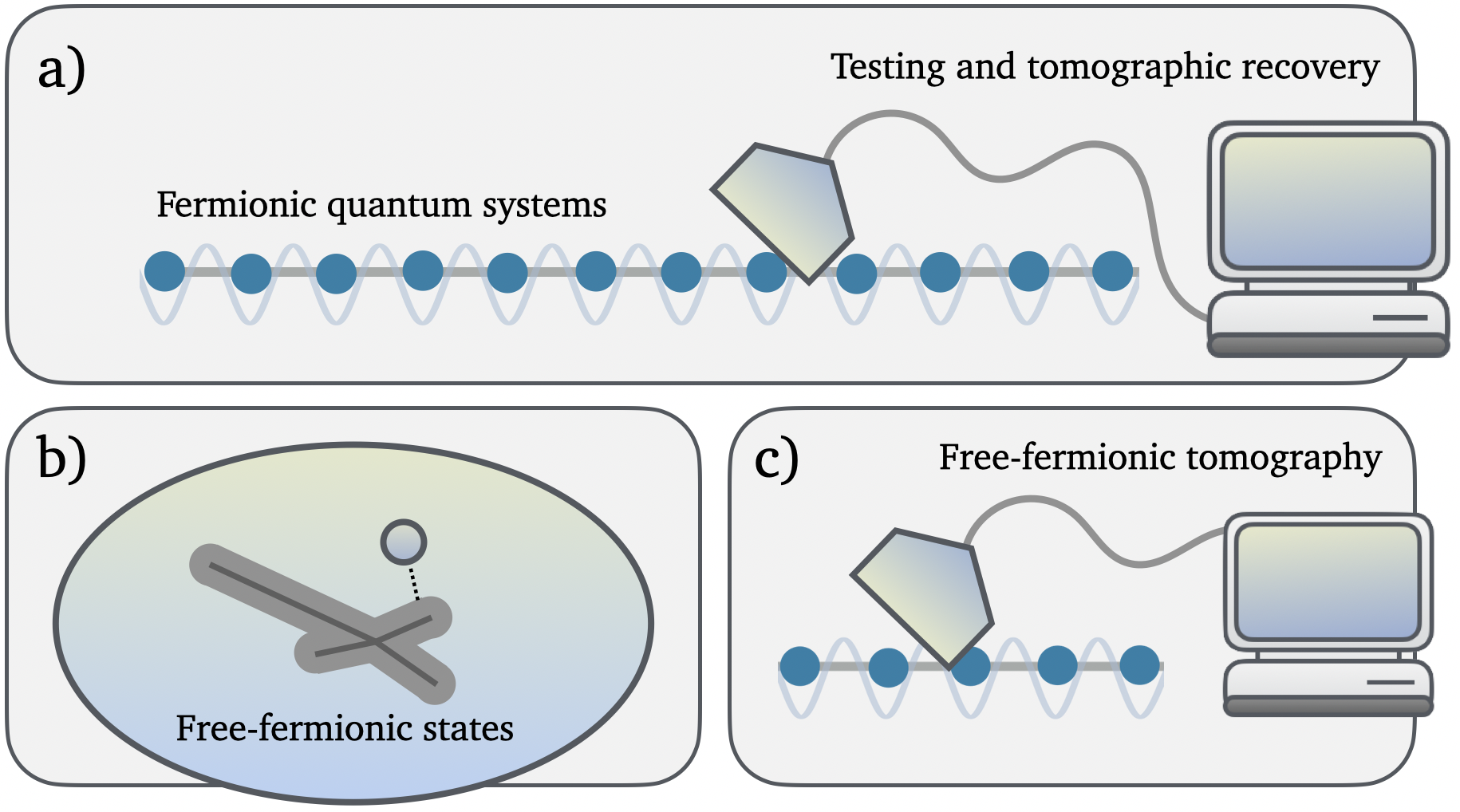}   

\caption{(a) Testing free-fermionic quantum states and improved tomography. (b) Any procedure aiming to solve the property testing problem for possibly mixed free-fermionic states must have a sample complexity scaling exponentially with the system size.
Low-rank free-fermionic states can be
efficiently and robustly tested, however.
(c) We then show sample optimal bounds for quantum state
tomography for free-fermionic states. 
Our results strongly rely on a new set of inequalities that we derive concerning the distance of free-fermionic states (which constitute a non-convex set) and the distance between their respective, efficiently tractable, correlation matrices (which form a convex set).}
\end{figure} 
However, before implementing any benchmarking protocol based on free-fermionic states, experimentalists must verify that the state prepared on a quantum device is indeed close to a free-fermionic state. This verification process can be formalized as a \textit{property testing problem}, which aims to determine whether a given unknown (possibly mixed) state is close to or far from the set of free-fermionic states through measurements on the unknown state. Similarly, given access to an unknown free-fermionic state, the task of quantum state tomography is to derive a full classical possibly efficient description of the state.

Free-fermionic quantum states are uniquely characterized by a polynomial number of correlation functions, encoded in the \textit{correlation matrix}. Experimentally feasible strategies for the learning tasks discussed above may thus rely on directly estimating these correlations to infer the nature of the state prepared on quantum processors. However, in real experiments, it is impossible to exactly determine correlation functions, resulting in only a \textit{perturbed} version of the correlation matrix. This raises a natural question: \textit{Given a correlation matrix estimated with some fixed accuracy, limited by finite sample constraints, what is the resulting error in describing the corresponding free-fermionic state?} In this work, we answer this question in different flavors, providing optimal perturbation bounds for free-fermionic states. These results, described below, not only provide fundamental theoretical tools and insights but also enable practical experimental applications.

The centerpiece of our work is the derivation of optimal perturbation bounds that relate the trace distance between two free-fermionic states to the distance between their correlation matrices. While it is widely understood that specifying a free-fermionic state can be achieved through its correlation matrix, in practice, a fundamental problem is to determine the error in trace distance when estimating the correlation matrix of an unknown free-fermionic state up to a given precision. Our Theorem~\ref{lem:4} thoroughly addresses this problem by providing two optimal perturbation bounds, applicable to pure and mixed states respectively. To our knowledge, such bounds constitute a novelty in the literature. However, when restricted to the pure-state scenario, similar, yet not optimal, bounds were already known~\cite{Gluza_2018,aaronson2023efficient}, which we consistently improve upon.
Additionally, we derive bounds relating the distance between a pure free-fermionic state and an arbitrary (possibly non-free-fermionic) state in terms of the difference between their correlation matrices. We also provide efficiently computable and experimentally measurable lower bounds for the minimum distance between a given state and the set of free-fermionic states, effectively quantifying the non-Gaussianity of the state.

We leverage our preceding results to derive significant advancements in applications such as property testing and tomography of free-fermionic states.
Concerning property testing, we first show that any procedure aiming to solve the property testing problem for possibly mixed free-fermionic states must have a sample complexity scaling exponentially with the size of 
the system. This, in particular, implies that there exists no efficient quantum algorithm for estimating the non-Gaussianity (i.e., distance to the free-fermionic set) of a general state. However, we then provide a sample and time-efficient algorithm to test:
1) whether a general state is close to or far from the set of low-rank free-fermionic states, and
2) whether a low-rank state is close to or far from the full set of free-fermionic states.
Although seemingly similar, the two approaches differ significantly: ultimately, the choice between them rests entirely with the user and the amount of prior knowledge they have about the state.

Concerning quantum state tomography, we first improve the sample complexity for learning pure free-fermionic states: we show that \( O\!\left({n^3}/{\varepsilon^2}\right) \) copies are sufficient, whereas the best bound previously known~\cite{Gluza_2018,ogorman2022fermionic,aaronson2023efficient,mele2024efficient} required \( O\!\left({n^5}/{\varepsilon^4}\right) \) copies, where \( n \) is the number of fermionic modes (or qubits) and \( \varepsilon \) is the desired accuracy in trace distance. We then generalize the algorithm to encompass the case where the unknown state is possibly mixed and show that \( O\!\left({n^4}/{\varepsilon^2}\right) \) samples are sufficient. Both tomography algorithms, for pure and mixed states, involve measuring the correlation matrix up to precision \( \varepsilon \), which can be efficiently done using experimentally feasible measurements, and exploit Theorem~\ref{lem:4} to transfer the error in trace distance. We remark that the scaling \( O(\varepsilon^{-2}) \) that we get is optimal for state tomography.
We also analyze the noise robustness of our tomography algorithm, exploring whether it can still be reliably applied to an unknown state that is not exactly free-fermionic but is sufficiently close to the set of free-fermionic states. This situation is particularly meaningful in experimental settings where one aims to prepare and subsequently learn a free-fermionic state for benchmarking/verification purposes, but the noise in the device makes the state preparation inexact. Our analysis reveals the noise-robustness of our algorithm, making it not only a theoretical speculation but also experimental-friendly and suitable for practical scenarios. 

On a practical note, it is worth emphasizing that our testing and learning algorithms may be considered experimentally feasible to implement both in a near-term fermionic analog quantum simulator setting, such as cold atoms in optical lattices~\cite{jordens_mott_2008,GrossFermions}, and in digital quantum computer settings, as they rely on easy-to-implement fermionic time evolutions and local measurements.

This work is organized as follows. We first introduce preliminaries. In subsection~\ref{sub:fund}, we explain our more fundamental results concerning free-fermionic states. Then, relying on such results, in subsections~\ref{sub:testmain} and~\ref{sub:tommain}, we respectively analyze property testing and tomography of free-fermionic states. Subsequently, in subsections~\ref{sub:relworkmain} and~\ref{sub:discmain}, we respectively discuss previous works in more detail and draw our conclusions.
For readability, the proofs of our results are extensively reported in the appendix. For a more detailed treatment, the interested reader is then referred to the more technical part of this work in Sections,~\ref{sec:preliminaries},~\ref{Sec:matrixtracedistanceinequalities},~\ref{Sec:propertytesting}, and~\ref{Sec:tomography}.

\section{Preliminaries}
In this section, we briefly provide the basic definitions necessary for the main results of this work, which are presented later on. 
We consider systems of $n$ fermionic modes (or qubits, by virtue of the 
Jordan-Wigner representation). Majorana operators are defined, through standard Pauli single qubit operators $\{I, X_j,Z_j,Y_j\}^{n}_{j=1}$, as 
\be
\gamma_{2k-1} := (\prod_{j=1}^{k-1} Z_j) X_k\,, \quad\quad\gamma_{2k} := (\prod_{j=1}^{k-1} Z_j) Y_k
\ee
for $k\in [n]$, where $[n] := \{1, \dots, n\}$.
Given a quantum state $\rho$, its \emph{correlation matrix}
(or covariance matrix) 
$\Gamma(\rho)$ is a $2n \times 2n$ matrix with elements 
\be
[\Gamma(\rho)]_{j,k} = -\frac{i}{2}\Tr\left(\left[\gamma_j, \gamma_k\right]\rho\right), 
\ee
where $j,k\in[2n]$. Correlation matrices are real and 
anti-symmetric, with (purely imaginary) eigenvalues in absolute value contained in
the interval $[0, 1]$. 
A well-known result in linear algebra~\cite{BookLinAlg} asserts that any real anti-symmetric matrix $C$ (such as correlation matrices) can be decomposed in the so-called `normal form': 
\begin{align}
    C=Q\bigoplus_{j = 1}^{n} \begin{pmatrix} 0 &  \lambda_j \\ -\lambda_j & 0 \end{pmatrix}Q^T,
    \label{eq:decomAntisym}
\end{align}
where $Q$ is an orthogonal matrix in $\mathrm{O}(2n)$ and $\lambda_j\ge 0 $, for any $j\in[n]$, are dubbed as `normal' eigenvalues, ordered in increasing order.

Let $\mathrm{O}(2n)$ denote the orthogonal group on a $2n$-dimensional vector space. There is bijection between $\mathrm{O}(2n)$ and \textit{free-fermionic unitaries} (or Gaussian) acting a on $n$ qubits system: For any orthogonal matrix $Q \in \mathrm{O}(2n)$, a free-fermionic unitary $U_Q$ is a unitary satisfying 
\be
U_Q^\dagger \gamma_\mu U_Q = \sum^{2n}_{\nu=1} Q_{\mu,\nu} \gamma_\nu 
\ee
for any $\mu \in [2n]$. 
This is a mild generalization of the $\mathrm{SO}(2n)$ case, which is associated with physical parity preserving Gaussian unitaries.
For the sake of generality, we always present our results in the context of $\mathrm{O}(2n)$. Another subset of physical Gaussian unitaries often considered across physics is the one of \emph{particle-number preserving} Gaussian unitary, which we explicitly define in the appendix and mention how our results would reduce in such a case.

Under free fermionic unitaries $U_Q$ with associated orthogonal matrix $Q\in \mathrm{O}(2n)$, correlation matrices transform simply as  
\be
\Gamma(U_Q\rho U_Q^{\dagger}) = Q \Gamma(\rho) Q^{T}.
\ee
A free fermionic state $\rho$ can be defined as a state which can be expressed as
\begin{align}
    \rho = U_Q \bigotimes^n_{j=1}\left(\frac{I + \lambda_j Z_{j} }{2}\right) U^{\dagger}_Q\,,
    \label{eq:defFREE}
\end{align} where $U_Q$ is the free-fermionic unitary associated with $Q \in \mathrm{O}(2n)$ and $\{\lambda_j\}^{n}_{j=1}$ are real numbers with $|\lambda_j| \leq 1$.

Any free-fermionic state can be uniquely specified by its correlation matrix. For more details, refer to Section~\ref{sec:preliminaries}.  

We frequently employ the Schatten $p$-norms as matrix norms, defined as $\|A\|_{p}\coloneqq\tr(|A|^p)^{1/p}$, for $|A|\coloneqq\sqrt{A^{\dag}A}$. From this, the \emph{trace distance} between two quantum states $\rho$ and $\sigma$ is (half of) their one-norm difference
\begin{equation}
    \|\rho-\sigma\|_1\coloneqq
    \tr(|\rho-\sigma|). 
\end{equation}
The trace distance has a nice operational meaning since it is related to the maximum success probability of distinguishing two states via arbitrary quantum measurements~\cite{MARK}. 

Before concluding the section, let us introduce the set of free-fermionic states. In particular, in the rest of this work, two sets of free-fermionic states are considered: $\mathcal{G}_{\mathrm{mixed}}$, comprising all free-fermionic quantum states, and $\mathcal{G}_R$, encompassing free-fermionic states with a rank at most $R=2^r$, where $r\in [n]$. We denote $\mathcal{G}_{\mathrm{pure}}$ the set $\mathcal{G}_R$ for $R=1$. As a matter of fact, we have $\mathcal{G}_{\mathrm{pure}}\subset \mathcal{G}_{R}\subset \mathcal{G}_{\mathrm{mixed}}$ for $1<r<n$.

\section{Perturbation bounds}
\label{sub:fund}
The primary motivation behind this work was to provide robust and experimental-friendly methods for testing and learning free-fermionic states. In achieving this goal, we have derived novel bounds between the distance of fermionic states with respect to their correlation matrices. These bounds are of independent interest and are instrumental for our subsequent testing and learning analysis. Below, we first present these bounds and briefly comment on their potential usage in other contexts beyond the scope of this work. The interested reader is referred to Section~\ref{Sec:matrixtracedistanceinequalities} for additional inequalities relating distance and correlation matrices.
\smallskip

{\bf Optimal trace-distance bounds for free-fermionic states.} It is a well-known part of folklore that `to know a free-fermionic state is sufficient to its correlation matrix.' However, in practice, one can only derive an estimation of the correlation matrix and, consequently, one can only know the free-fermionic state approximately. As explained above, thanks to its information-theoretic operational meaning, the trace distance is the most meaningful notion of distance to characterize the inaccuracy in the approximation. It is, therefore, 
a fundamental problem in quantum information to determine what is the error in trace distance when estimating the correlation matrix associated with a free-fermionic state up to precision $\varepsilon$. We now introduce two new perturbation bounds for the trace distance of two free-fermionic states in relation to the distance between their correlation matrices.
\begin{theorem}[Perturbation bounds]\label{lem:4}
Let $\psi$ and $\phi$ be two pure free-fermionic states, then it holds that
\begin{align}
\|\psi-\phi\|_1&\leq \frac{1}{2} \|\Gamma(\psi) - \Gamma(\phi)\|_2,
\label{eq:GAUSSpurestatedistance}
\end{align}
while for $\rho$ and $\sigma$ being two possibly mixed free-fermionic states, it holds
that
\begin{equation}
    \|\rho-\sigma\|_1\leq \frac{1}{2}\|\Gamma(\rho)-\Gamma(\sigma)\|_1\,\,.
\label{eq:upperboundmixed}
\end{equation}
\end{theorem}

The proof is detailed in Section~\ref{Sec:matrixtracedistanceinequalities}. The proof idea for the pure-case formula is to use the exact fidelity formula between two free-fermionic states~\cite{Bravyi_2016} and derive tight bounds from it. For the mixed-case scenario, we create a path in the space of correlation matrices and tightly bound each small displacement using the Schatten $1$-norm between correlation matrices. Both inequalities \eqref{eq:GAUSSpurestatedistance} and \eqref{eq:upperboundmixed} are essentially optimal, as they are obtained from `exact derivative' calculations and are saturated for some families of states. In particular the first bound is saturated for all pure $3$-mode states, while the second bound is saturated for all single-mode states. 
To showcase the relevance of our proof method, we mention that we have applied the same technique in a follow-up work~\cite{bittel2025optimalestimatestracedistance} on free-bosonic states, demonstrating that this method yields an optimal trace distance bound in these different systems as well, improving upon the state-of-the-art bounds available in the literature~\cite{Mele2024bosonic,holevo2024estimatestracenormdistancequantum}.

We now discuss the previous state-of-the-art bounds in the literature for free-fermionic states, which are improved upon by Theorem~\ref{lem:4}. To the best of our knowledge, no bound of this type was previously known for mixed free-fermionic states in the literature. For pure states, a bound was established in~\cite{aaronson2023efficient}, but it holds only in the more restrictive setting of particle-number conserving states and is strictly weaker than ours. In particular, the trace distance between two free-fermionic states was upper bounded by the \emph{square root} of the distance between their correlation matrices. In contrast, our bound in Eq.~\eqref{eq:GAUSSpurestatedistance} eliminates the square root, yielding a strictly tighter estimate. This improvement translates directly into better performance for tomography, as we elaborate in the following section.
In section~\ref{Sec:matrixtracedistanceinequalities} we also give the equivalent perturbation bounds for the state fidelity between fermionic states which may be of independent interest.

Similarly to Theorem~\ref{lem:4}, the following Theorem provides lower bound the trace distance between two pure Gaussian states in terms of their covariance matrices. 

\begin{theorem}[Lower bound on trace distance]\label{th2new} Let $\psi$ and $\phi$ be two non-orthogonal pure free-fermionic states, then it holds that 
\begin{align}
\|\psi-\phi\|_1\ge 2\sqrt{1-\exp\left(-\frac{\|\Gamma(\psi)-\Gamma(\phi)\|_2^2}{16}\right)}
\end{align}
\end{theorem}
The proof is detailed in Section~\ref{Sec:matrixtracedistanceinequalities}, and it is closely related to the proof of Theorem~\ref{lem:4}.

Besides being fundamental for the efficiency of subsequent testing and tomography algorithms, we anticipate that Theorem~\ref{lem:4}'s inequalities will have extensive applicability in the context of free fermions, both theoretically and experimentally. This is particularly relevant when one aims to bound the distance between two free-fermionic states (or overlap if one of them is pure) in terms of the norm difference of their correlation matrices (which, we recall, are the $2n \times 2n$ matrices that uniquely characterize a free-fermionic state).

{\bf Perturbation bound beyond free-fermionic states.} Imagine one aims to prepare a pure free-fermionic state $\psi$ on a quantum processor. In practice, an imperfect, possibly non-free-fermionic state $\rho$, yet close version of the state is actually prepared. 
The task of quantum state certification~\cite{PRXQuantum.2.010201} is to verify the almost correct preparation of the target state with respect to the trace distance. 
Here, we aim to understand how the trace distance between between the target pure state $\psi$ and the unknown state $\rho$ is controlled by their respective correlation matrices, which can be efficiently measured in practical scenarios. Below, we upper bound the trace distance between a (target) perfect pure free-fermionic state $\psi$ and a possibly non-free-fermionic mixed imperfect realization $\rho$ through their respective correlation matrices. 

\begin{lemma}[Closeness of quantum states in terms of correlation matrices]\label{lem:3} Let $\psi$ be a free fermionic pure state and $\rho$ be an arbitrary (possibly non-free-fermionic) quantum state. 
Then
\be
\|\psi-\rho\|_{1}\le\sqrt{\|\Gamma(\psi)-\Gamma(\rho)\|_{1}}.\label{eq:inpm}
\ee
\end{lemma}
Notice that Eq.~\eqref{eq:inpm} constitute a generalization of Theorem~\ref{lem:4} to the broader scenario where one of the states may be non-free-fermionic.

Assuming that the $2n \times 2n$ correlation matrix $\Gamma(\psi)$ of the target free-fermionic state  $\psi$ is known (which can be efficiently computed given any efficient classical description of $\psi$), Eq.~\eqref{eq:inpm} offers an efficient and direct method to verify the accurate preparation of $\psi$, as the right-hand side can be efficiently estimated experimentally.
This is further discussed and proved in Section~\ref{Sec:matrixtracedistanceinequalities}.
\smallskip

{\bf Distance from the set of free-fermionic states: computable measures of non-Gaussianity.} Given a quantum state $\rho$, the minimal distance from the set of free-fermionic states provides an operational measure of non-Gaussianity for fermionic systems. We recall that the minimal trace distance from a given set of free states constitutes a universally valid measure that respects all the desired properties of resource monotones~\cite{RevModPhys.91.025001}. However, the problem with such a measure is that it is neither computable, involving a minimization procedure, nor experimentally measurable. Here, we present lower bounds on the trace distance between a quantum state $\rho$ and the set of free-fermionic states, that allow for the efficient estimation of non-Gaussianity in practical and experimental settings. 

\begin{lemma}[Lower bounds to trace distances]\label{lem:1}
Let $R=2^r$ with $r\in [0,n-1]$ being an integer. Let $\rho$ be an arbitrary quantum state and let $\lambda_{r + 1}$ be the $(r+1)$-th smallest normal eigenvalue of its correlation matrix $\Gamma(\rho)$, then
\begin{align}
\min_{\sigma \in \mathcal{G}_{R}} \norm{\rho - \sigma}_1 &\ge 1-\lambda_{r+1},
\label{eq:non-GaussianitymeasureMAIN}
\end{align}
where $\mathcal{G}_{R}$ is the set of free-fermionic state with at most $R$-rank. Conversely, if $\rho$ is a arbitrary quantum state with rank $\le 2^r$, then 
\begin{align}
\min_{\sigma\in\mathcal{G}_{\mathrm{mixed}}}\norm{\rho-\sigma}_{1}\ge\frac{(1-\lambda_{r+1})^{r+1}}{1+(r+1)(1-\lambda_{r+1})^r}.
\label{eq:inCH}
\end{align}
\end{lemma}

Note that Eq.~\eqref{eq:non-GaussianitymeasureMAIN} requires no assumption on the state $\rho$, and it lower bounds the distance from $\mathcal{G}_{R}$, the set of free-fermionic states with rank $\le 2^r$. 
Conversely, with the promise that $\rho$ has a rank at most $2^r$, Eq.~\eqref{eq:inCH} establishes a lower bound on the distance from the set of all free-fermionic states $\mathcal{G}_{\mathrm{mixed}}$. The proof for the above results can be found in Section~\ref{sec:non-Gaussianity}. 
For the pure-case scenario, we also find the upper bound
\begin{align}
\min_{\sigma \in \mathcal{G}_{\mathrm{pure}}} \norm{\rho - \sigma}_1 &\le \sqrt{2\sum^n_{j=1}(1-\lambda_j)}.
\end{align}
As discussed in Section~\ref{sec:preliminaries}, estimating the normal eigenvalues of the correlation matrix of a quantum state can be done efficiently by employing single-qubit Pauli or free-fermionic measurements. Thus, the above inequalities provide an efficient  procedure to quantify how far a state is from the set of free-fermionic states in practical scenarios.

\section{Property testing}
\label{sub:testmain}
A central question that stands behind our work is:
Given an unknown quantum state $\rho$, is it close to or far from the set of free fermionic states?
Our goal is to identify scenarios where this question can be addressed resource-efficiently, provide algorithms with provable efficiency guarantees, and delineate situations in which answering the question is challenging. Let us formally define the problem.
\begin{problem}[Property testing of free-fermionic states]
\label{prob:testingMAIN}
Let $\varepsilon_B > \varepsilon_A \ge 0$.
Given $N$ copies of an unknown quantum state $\rho$ with the promise that it falls into one of two distinct situations:
\begin{itemize}
    \item \textbf{Case A:} There exists a free-fermionic state $\sigma \in \mathcal{G}$ such that $\| \rho - \sigma \|_{1} \le \varepsilon_A$.
    \item \textbf{Case B:} The state $\rho$ is $\varepsilon_B$-far from all free-fermionic states $\sigma$, indicating $\min_{\sigma \in \mathcal{G}} \| \rho - \sigma \|_{1} > \varepsilon_B$.
\end{itemize}
Determine whether the state is in Case A or Case B with high probability through measurements performed on the provided $N$ state copies. Further specifications regarding the rank of the state $\rho$ and the set of free-fermionic states $\mathcal{G} \equiv \mathcal{G}_{\mathrm{pure}}, \mathcal{G}_{\mathrm{mixed}}, \mathcal{G}_{R}$ must be provided.
\end{problem}

While we show that solving Problem~\ref{prob:testingMAIN} in its full generality requires an exponential amount of resources, rendering it unfeasible for practical purposes, we present algorithms capable of efficiently solving it under certain assumptions about the state under examination or by restricting the set of considered free-fermionic states.

When no prior assumptions on the state $\rho$ and no restrictions on the set of free-fermionic states $\mathcal{G}$ are provided, we establish the general hardness for Problem~\ref{prob:testingMAIN}, demonstrating that $N = \Omega(2^n)$ copies of the state $\rho$ are necessary. 
This result also directly implies that any strategy aiming to estimate non-Gaussianity—i.e., to provide a meaningful probe of the trace distance between a given state and the set of (possibly mixed) free-fermionic states—must necessarily require an exponential number of copies in the system size. Indeed, if such an efficient estimation strategy existed, it could be used to solve the property testing problem, contradicting our hardness result.
Below, we present a more refined version of the mentioned no-go result for testing free-fermionic states.

\begin{theorem}[Hardness of testing bounded rank free-fermionic states] \label{th:hardnessMAIN}
Let $\varepsilon_B=O(1)$ and $r \in [0,n]$. Let $\rho$ denote the unknown state and $\mathcal{G}$ the set of free-fermionic states considered. To solve Problem~\ref{prob:testingMAIN}, with at least a $2/3$ probability of success, $N = \Omega(2^r)$ copies are necessary if either of the following hypotheses is assumed:
\begin{itemize}
    \item $\rho$ is such that $\mathrm{rank}(\rho) \le 2^r$,
    \item $\mathcal{G} \equiv \mathcal{G}_{R}$ with $R \le 2^r$,
\end{itemize}
with $r \in [n]$. In particular, for $r = \Omega(n)$, the sample complexity grows exponentially in the number of modes $n$.
\end{theorem}

While the formal proof is found in Section~\ref{subsec:hardtesting}, the core idea of this complexity arises from recognizing that the maximally mixed state is free-fermionic. This fact allows us to leverage the hardness of a closely related problem of the identity testing, i.e., distinguishing whether the underlying state is the maximally mixed state or far from it in trace distance~\cite{odonnell2015quantum}.

Another intuition behind the hardness of this problem arises from the fact that determining whether a state is free-fermionic (see Definition~\eqref{eq:defFREE}) is inherently harder than determining if it is a (diagonal) tensor product state, and it is well-established that testing for a tensor product structure is an hard problem~\cite{productstatehardness}.

The above theorem establishes necessary conditions on the samples $N$ to solve Problem~\ref{prob:testingMAIN}. Given this result, Theorem~\ref{th:hardnessMAIN} prompts the natural question of whether an algorithm exists that scales exponentially with $r$ (polynomially with the rank), which solves the property testing problem under the assumption that $\rho$ has a rank at most $R = 2^r$ or when restricting to the set $\mathcal{G}_R$ of free-fermionic state with rank $\le R$. In response, we propose two learning algorithms that scale as $O(\mathrm{poly}(n,2^r))$, and summarize our findings in the following informal theorem.

\begin{theorem}[Efficient free-fermionic testing - Informal version of Theorems~\ref{th:mix1} and~\ref{th:mix2}]
\label{th:efficientMAIN}
Problem~\ref{prob:testingMAIN} can be solved with $N = \mathrm{poly}(n, 2^r)$ copies of $\rho$ in the following scenarios:
\begin{enumerate}[label=(\roman*)]
    \item The given state $\rho$ is such that $\mathrm{rank}(\rho) \le 2^r$.
    \item The set of Gaussian states is restricted to $\mathcal{G}_R$ with $R \le 2^r$.
\end{enumerate}
We provide algorithms for both cases that require
\begin{align}
    \tilde{O}\left[n^3 \varepsilon_{\mathrm{stat}}^{-2} + n^2 4^r \varepsilon_{\mathrm{tom}}^{-2}\right]
\end{align}
samples (where we neglect $\log$ factors in $n$ and $1/\delta$). Specifically, $\varepsilon_{\mathrm{tom}}=\frac{\varepsilon_B}{2}-(n+1)\varepsilon_A$ represents the desired error in the `subsystem tomography step' performed in our testing algorithm, and $\varepsilon_{\mathrm{stat}}$ is the error in estimating the correlation matrix, which differs between the two cases:

\begin{enumerate}[label=(\roman*)]
    \item $\epsstat=\frac{\varepsilon_B^2}{32(n-r)}-(2\varepsilon_A)^{\frac{1}{r+1}}$\,,
    \item $\epsstat=\frac{\varepsilon_B^2}{32(n-r)}-\varepsilon_A$\,.
\end{enumerate}
The testing protocol requires that both $\epsstat$ and $\varepsilon_{\mathrm{tom}}$ are non-negative, thereby imposing nontrivial conditions on $\varepsilon_A,\varepsilon_B$ for the success of the algorithm. In each case, the computational resources scale as $\mathrm{poly}(n, 2^r)$.
\end{theorem}

The core idea behind the efficiency of this task lies in the fact that any $2^r$-rank free-fermionic state, up to a specific free-fermionic unitary transformation, can be reduced to a diagonal product state which is mixed over at most $r$ modes/qubits. Consequently, the algorithm identifies this free-fermionic unitary, allowing the state to be factorized into a product of an $r$-modes diagonal tensor product state and a computational basis state on $n-r$ modes. This reduction simplifies the problem to testing whether the $r$-modes state has the desired free-fermionic structure, which can be efficiently solved with a complexity scaling as $2^r$. For example, this can be achieved by performing full state tomography on $r$ qubits only. The correctness of such algorithms relies on the inequalities we have derived and presented in Lemma~\ref{lem:3}, \ref{lem:1}, and Theorem~\ref{lem:4}. Additional details, including assumptions on $\varepsilon_A$ and $\varepsilon_B$, the algorithms presented in pseudocode, as well as the sample and computational analyses, can be found in Section~\ref{Sec:propertytesting}.

Let us summarize our findings. In its full generality, that is, without rank assumptions on the state $\rho$ and considering $\mathcal{G} \equiv \mathcal{G}_{\mathrm{mixed}}$, Problem~\ref{prob:testingMAIN} requires $N = \Omega(2^n)$ samples of the state to be solved. However, as claimed in Theorem~\ref{th:efficientMAIN}, we have established that Problem~\ref{prob:testingMAIN} can be efficiently addressed both sample-wise and computationally under two specific scenarios: $(i)$ when the given state $\rho$ has a rank that scales polynomially with the number of modes $n$, or $(ii)$ when the focus is solely on establishing the closeness to the set of Gaussian states with polynomially bounded rank.
Although these two approaches share similarities in principle, they differ significantly both conceptually and operationally. Ultimately, the choice between them rests entirely with the user. Specifically, if they possess guarantees that the prepared state predominantly should have low-rank, then we provide provable guarantees in probing the distance concerning all Gaussian states $\mathcal{G}_{\mathrm{mixed}}$. Conversely, if an experimentalist lacks a clear understanding of the nature of the state due to imprecisions or noise, they only have reliable assurances in determining whether the state is close to or far from the set of Gaussian states with polynomially bounded rank. Although seemingly less satisfactory, the latter approach is more general and applicable in completely agnostic scenarios.

\bigskip

\section{Efficient and robust tomography}
\label{sub:tommain}
We conclude by presenting a simple algorithm for efficiently learning an unknown $n$-mode/qubit free-fermionic state, pure or mixed. More rigorously, we address the following problem: given an unknown free-fermionic state $\rho$, design a computationally efficient quantum learning algorithm that consumes $N$ copies of $\rho$ and outputs a classical description $\hat{\rho}$ such that, with probability $\geq 1-\delta$, it is guaranteed that $\|\hat{\rho} - \rho\|_1 \leq \varepsilon$.

Our algorithm relies on accurately estimating the correlation matrix of the state, which can be done efficiently via simple free-fermionic measurements. Since such states are completely characterized by their correlation matrices, this provides a straightforward approach to tomograph free-fermionic states. However, the challenging part lies in accurately propagating the error incurred in the estimation of the correlation matrix with respect to the trace distance.

Some earlier works~\cite{Gluza_2018,aaronson2023efficient,ogorman2022fermionic} have already tackled this problem for the specific case of $\rho$ being a pure free-fermionic state. However, we provide a better sample complexity compared to previous works by capitalizing on our new inequality in Eq.~\eqref{eq:GAUSSpurestatedistance} of Theorem~\ref{lem:4}.

\begin{proposition}[Tomography of pure free-fermionic states]\label{prop:purestatetominmain}
Let $\psi$ be a pure free-fermionic quantum state. For $\varepsilon, \delta \in (0,1)$, there exists a learning algorithm that utilizes $N=\lceil 32(n^3/\varepsilon^2) \log(4n^2/\delta)\rceil$ copies of the state and only single-copy measurements to learn 
an efficient classical representation of a pure state $\hat{\psi}$ obeying $\|\hat{\psi} - \psi\|_1 \leq \varepsilon$ with probability at least $1 - \delta$.
\end{proposition}
\begin{figure}
    \centering
    \includegraphics[width=1\linewidth]{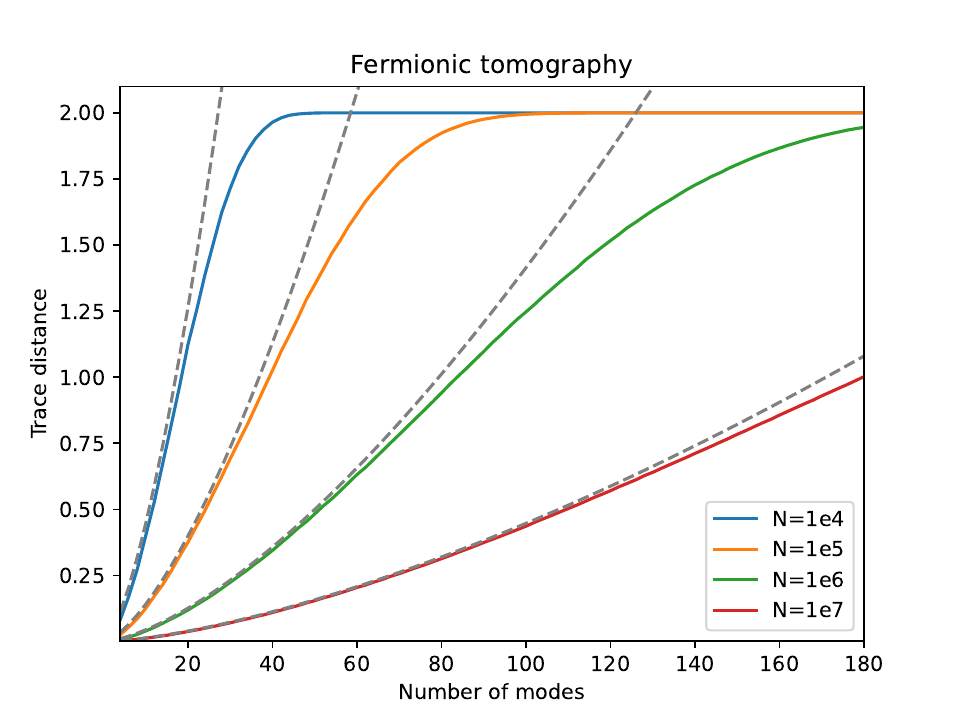}
    \caption{Numerical simulation of the performance of  the tomography algorithm for pure Gaussian states. We consider different numbers of total measurement rounds $N$ as a function of the number of modes. The fit (gray dotted line) uses $\varepsilon= 2n^{3/2} N^{-1/2}$, where $\varepsilon$ denotes the trace distance error. The data is averaged over 50 randomly selected Gaussian states each.}
    \label{fig:tom_pff}
\end{figure}
Specifically, the best previously known sample-complexity bound~\cite{Gluza_2018, aaronson2023efficient, ogorman2022fermionic} scaled as $O(n^5/\varepsilon^4)$, while our bound scales as $O(n^3/\varepsilon^2)$, without any assumptions on the nature of the free-fermionic state (e.g., a fixed number of particles).
In Figure~\ref{fig:tom_pff}, we simulate the tomography algorithm. As expected, the numerical predictions tightly match the analytical predictions, albeit with a better prefactor than the one obtained from the proof. 

Moreover, our analysis extends to the more realistic mixed-state scenario. This extension also heavily relies on the toolkit developed and discussed in the first part of this work, namely Eq.~\eqref{eq:upperboundmixed} of Theorem~\ref{lem:4}. 
\begin{theorem}[Tomography of mixed free-fermionic states]
\label{th:mixedtommain}
Let $\varepsilon, \delta \in (0,1)$ and $\rho$ be a mixed free-fermionic state. There exists a computationally efficient quantum algorithm that, employing $N=O\!\left((n^4/\varepsilon^2)\log({n^2}/\delta)\right)$ copies of the state $\rho$ and single-copy measurements, learns an efficient classical representation of a state $\hat{\rho}$ obeying $\|\hat{\rho} - \rho\|_1 \leq \varepsilon$, with $\ge 1 - \delta$ probability.
\end{theorem}
The above theorems achieve the task of learning pure and mixed free fermionic quantum states, respectively. We remark that the scaling in $\varepsilon$ for both tomography algorithms is $O(\varepsilon^{-2})$ and thus optimal, given the finite shot regime. This optimal scaling is achieved only due to the optimal trace-distance bounds presented in Theorem~\ref{lem:4}. 

Both algorithms assume access to a perfect copy of a free fermionic state. Next, we address the concern about the robustness of these algorithms with respect to noisy preparations. We strengthen these results by showing that Theorem~\ref{th:mixedtommain} still holds if the state $\rho$ is sufficiently close to the set of all free-fermionic states $\mathcal{G} = \mathcal{G}_{\mathrm{mixed}}$. In particular, in Proposition~\ref{thm_robustness_gaussian_states2}, we show that if the unknown state $\rho$ satisfies $\min_{\sigma \in \mathcal{G}} \|\rho - \sigma\|_1 \leq O(\varepsilon / n)$, then Theorem~\ref{th:mixedtommain} can still be applied. Similarly, in Proposition~\ref{thm_robustness_gaussian_states}, we show that Theorem~\ref{th:mixedtommain} can still be applied also in the case in which the state $\rho$ is such that its \emph{relative entropy of non-Gaussianity} was $\leq \varepsilon^2$, i.e., if $\rho$ satisfies $\min_{\sigma \in \mathcal{G}} S(\rho \| \sigma)\le \varepsilon^2$, where $S(\rho \| \sigma)$ is the quantum relative entropy between $\rho$ and $\sigma$. The relative entropy of non-Gaussianity possesses desirable properties, making it a meaningful measure of non-Gaussianity~\cite{Marian2013,Lumia:2023ofv}.
A more detailed analysis of what we have examined here can be found in Section~\ref{Sec:tomography} of the appendix.

 \section{Related works}
\label{sub:relworkmain}
It should be noted that an exact formula for the fidelity between two possibly mixed free-fermionic states is known~\cite{Banchi_2014,Zhang_2023}, exhibiting a notably non-trivial dependence on their correlation matrices. While one might attempt to derive a perturbation bound from this formula in terms of the norm difference between the correlation matrices, our exploration of this strategy revealed it to be more complex and ultimately less effective in yielding an upper bound as good as that obtained with our current approach.
Beside providing optimal trace-distance bounds for free-fermionic states, the main focus of this work are practical and efficient quantum algorithms for free-fermionic state testing and tomography. In this section, we discuss previous works in the literature.
Addressing the testing of specific properties in quantum states spans various contexts. A property tester for a quantum state class $\mathcal{C}$ involves taking copies of a state $\rho$ as input and determining either whether $\rho$ belongs to $\mathcal{C}$ (or is close enough to it) or if $\rho$ is $\varepsilon$-far in trace distance from all states in $\mathcal{C}$, provided that one of these scenarios is true.

In a similar spirit to our work, Ref.\ \cite{gross2021schur} analyzed property testing of pure stabilizer states~\cite{gottesman1998heisenberg}, which are another class of non-trivial classical simulable states. It has been shown that $O(1)$ copies of the unknown state suffice to solve the property testing problem. This remarkable efficiency is ultimately attributed to the use of entangling (Bell) measurements and techniques known as \textit{Bell difference sampling}. In a slightly more general and practical scenario, this has been extended in 
Ref.\ \cite{grewal2023improved,arunachalam2024toleranttestingstabilizerstates} to the class $\mathcal{C}$ consisting of pure states \emph{sufficiently} close in trace distance to pure stabilizer states, i.e., to a so-called \emph{tolerant} testing. We remark that also in our free-fermionic case we achieve tolerant testing (albeit with tolerant accuracy parameter scaling inversely with the system size).
Property testing of matrix product states has also been analyzed in Ref.~\cite{soleimanifar2022testing}. Similar settings in the free-bosonic realm have also been considered \cite{Leandro}. Recently, works have appeared in the context of testing the locality of Hamiltonians through time-evolutions~\cite{bluhm2024hamiltonian,gutiérrez2024simple}. In Ref.\ \cite{odonnell2015quantum}, the identity testing problem (i.e., in which the set $\mathcal{C}$ is composed only of the maximally mixed state) has been shown to require a sample complexity $\Theta(2^n)$, where $n$ is the number of qubits. Similarly, product state testing has been examined in Ref.\ \cite{productstatehardness}, focusing on the class $\mathcal{C}$ being the set of product states.

Quantum state tomography~\cite{anshu2023survey} is the second prominent learning problem addressed in this work. Given copies of an unknown quantum state $\rho$, the aim of quantum state tomography is, through arbitrary measurements, to output a classical description $\hat{\rho}$ of a state that is $\varepsilon$-close in trace distance to $\rho$ with high probability. In Refs.~\cite{Haah_2017,odonnell2015efficient}, a sample-optimal quantum state tomography algorithm has been introduced. Specifically, it has been shown that $\tilde{\Theta}(4^n/\varepsilon^2)$ copies of the state are both necessary and sufficient for quantum state tomography of arbitrary mixed states. Given the intrinsic inefficiency of quantum state tomography algorithms for general states, a natural question arises: if one restricts the class of input quantum states, can quantum state tomography be conducted both sample- and computationally efficiently? In this regard, a number of insightful works have been produced that demonstrate sample and computational efficiency for specific classes of states, which include pure stabilizer states~\cite{montanaro2017learning}, $t$-doped stabilizer states~\cite{grewal2023efficient,leone2023learning,leone2023learning22PUBL1,PhysRevA.109.022429,hangleiter2023bell} (states obtained by at most $t$ non-Clifford local gates), matrix product states~\cite{cramer_efficient_2010,Wick_MPS,lanyonEfficientTomographyQuantum2017,Hangleiter,baumgratzScalableReconstructionDensity2013,huang2024learning}, finitely-correlated states~\cite{fanizza2023learning}, high-temperature Gibbs states~\cite{rouzé2023learning}, and more. 

As already mentioned, prior to our work, previous works have tackled the problem of free-fermionic pure-state tomography~\cite{Gluza_2018,aaronson2023efficient,ogorman2022fermionic}, over which we improve upon their sample complexity and generalize to the mixed scenario. Recently, a related work appeared in which it has been shown how to perform efficient tomography of $t$-doped free-fermionic states~\cite{mele2024efficient}, i.e., states prepared by free-fermionic circuits doped with at most $t$ local non-free-fermionic gates. Moreover, the analogous problem of learning mixed Gaussian states has previously been unresolved also in the bosonic context. However, in a parallel work~\cite{Mele2024bosonic}, together with other coauthors, we fill this gap in the bosonic literature.
Furthermore, process tomography of fermionic Gaussian unitaries has been addressed in Refs.~\cite{Oszmaniec_2022,cudby2024learninggaussianoperationsmatchgate}, leveraging distance error bounds between Gaussian unitaries~\cite{Oszmaniec_2022,vanLuijk_2024}.

Recently, a manuscript addressing pure-state fermionic Gaussian testing was posted on the ArXiv~\cite{lyu2024fermionicgaussiantestingnongaussian}. While it introduces promising alternative methods to achieve this task through the concept of fermionic quantum convolution, it focuses solely on pure-state testing and lacks guarantees on sample complexity and on trace-distance accuracy, which we fully address in this manuscript. Similarly, we became aware of a work in preparation~\cite{Duttfermionictesting} that aims to address the same problem of improved fermionic states tomography, though restricted to the pure-state scenario and employing different methods.

\section{Discussion and open questions}
\label{sub:discmain}
In this work, we have shown important norm bounds concerning free-fermionic states, which we leveraged to obtain multiple notable applications:
First, we have demonstrated methods for testing whether a quantum state is free-fermionic or not, utilizing experimentally feasible measurements. Moreover, we have identified scenarios in which any algorithm to solve this problem would be inevitably inefficient, distinguishing them from situations where efficient algorithms are possible, like ours.

Additionally, we have introduced efficiently estimable lower bounds on the distance of a state from the set of free-fermionic states, which offer valuable insights into quantifying the non-Gaussianity of a quantum system. Furthermore, we have improved the performance guarantees of previously proposed free-fermionic pure states tomography, generalized the algorithm to the more complex mixed-state scenario, and discussed their noise-robustness, relying on our novel bounds on the trace distance between two free-fermionic quantum states. We posit that the findings presented in this work can serve as a valuable resource for designing and conducting quantum simulation experiments. Moreover, while our analysis comprehensively addresses the literature on testing and tomography of free-fermionic states, we expect that the inequalities and tools shown in this work will have other applications beyond the context of quantum learning theory. 
In particular, our proof techniques for fermionic systems have already been instrumental in deriving similar trace-distance bounds in the bosonic setting~\cite{bittel2025optimalestimatestracedistance}.

There are several interesting open questions related to the topics discussed in this work. 

A natural open question regarding the property testing of free-fermionic states is whether there exists a testing algorithm that requires only a constant $O(1)$ number of copies, independent of the system size.

Our results on tomography also open several promising directions for future research. While our algorithm is optimal in its dependence on the accuracy parameter $\varepsilon$ and scales polynomially with the system size $n$, it remains an open question whether an algorithm can be designed that is also optimal in $n$. Fundamentally, any tomography protocol for $n$-mode pure free-fermionic states requires at least $\Omega(n)$ copies, as follows from standard arguments based on the Holevo bound and $\varepsilon$-net constructions~\cite{Zhao_2024}. In contrast, classical shadow techniques~\cite{Zhao_2024}, combined with similar $\varepsilon$-net arguments, yield an upper bound of $O(n^2)$ copies, though such an algorithm would be computationally inefficient. Our algorithm, on the other hand, achieves efficient runtime, albeit with a sample complexity of $O(n^3)$. Closing the gap between the $\Omega(n)$ lower bound and the $O(n^2)$ upper bound (or $O(n^3)$ if computational efficiency is prioritized) remains an important open problem. For context, in the case of stabilizer state learning~\cite{montanaro2017learning}, $\Theta(n)$ copies are both necessary and sufficient. However, for free-fermionic states, the optimal scaling in $n$ remains unresolved, and we believe this is a promising direction for future work.

Furthermore, an intriguing open question is whether our mixed-state tomography algorithm can still be reliably applied if the unknown state is promised to be close in trace distance to the set of free-fermionic states by a constant factor, rather than the $O(1/n)$-closeness we currently assume. It would be valuable to explore whether a better algorithm can be designed that enjoys this more desirable feature.

Moreover, it would be intriguing to explore testing and tomography of states with a small Gaussian extent~\cite{dias2023classical,cudby2023Gaussian,reardonsmith2024improved}, which are states that can be written as a superposition of a few free-fermionic states. %Furthermore, testing of fermionic Gaussian unitaries (or, more generally, Gaussian channels) presents another interesting open problem. 

\textit{Acknowledgments.}
We thank Frederik Von Ende for help in proving theorem \ref{th:mixedtracedistance}, as well as Francesco Anna Mele, Ludovico Lami, Alexander Nietner, Filippo Girardi for useful discussions. This work has been supported by the BMWK (EniQmA), the BMBF (FermiQP, MuniQCAtoms, DAQC), 
the ERC (DebuQC) and 
the DFG (CRC 183).
A.A.M.\ also acknowledges support by Laboratory Directed Research and 
Development (LDRD) program of Los Alamos National 
Laboratory (LANL) under project number 20230049DR and by the 
U.S.\  
Department Of Energy through a quantum computing program sponsored by the 
LANL Information Science \& Technology 
Institute. 

%\vspace{-0.5em}

\smallskip

\bibliography{ref}

\clearpage

\resumetoc
\onecolumngrid
\begin{center}
\vspace*{\baselineskip}
{\textbf{\large Supplemental material
}}\\
\end{center}

%\let\oldaddcontentsline\addcontentsline% Store \addcontentsline
%\renewcommand{\addcontentsline}[3]{}% Make \addcontentsline a no-op
%\medskip

%\bibliographyst{alpha}
  
\renewcommand{\theequation}{S\arabic{equation}}
\renewcommand{\bibnumfmt}[1]{[S#1]}
\newtheorem{thmS}{Theorem S\ignorespaces}

\newtheorem{claimS}{Claim S\ignorespaces}
  
\tableofcontents

\section{Preliminaries}\label{sec:preliminaries}

\subsection{Notation and basic definitions}
We use the following notation throughout our work. We denote with $\mathbb{C}^{d\times d}$ the set of $d\times d$ complex matrices, with $d\in \mathbb{N}$. The notation $[d]$ denotes the set of integers from $1$ to $d$, i.e., $[d] \coloneqq \{1,\dots, d\}$. We denote as $I$ the identity operator, with a subscript specifying the dimension when necessary for clarity.
The \emph{Schatten $p$-norm} of a matrix $A\in \mathbb{C}^{d\times d}$, with $p\in \left[1,\infty\right]$, is given by 
\be
\norm{A}_p\coloneqq \Tr\small(\small(\sqrt{A^\dagger A}\,\small)^p\small)^{1/p},
\ee
which corresponds to the $p$-norm of the vector of singular values of $A$.
The \emph{operator norm} of a matrix $A\in \mathbb{C}^{d\times d}$ is equal to its largest singular value.
We denote the \emph{Hilbert-Schmidt scalar product} as $\hs{A}{B}:=\Tr\left(A^\dagger B\right)$.
Let $A=UDV$ be the singular value decomposition of $A\in \mathbb{C}^{n\times n}$, with $U,V$ being unitary matrices and $D=\mathrm{diag}(D_1,\dots,D_n)$ being a diagonal matrix with $D_1\ge \dots \ge D_n \ge 0$; the Ky Fan norm of order $R$ (with $R\in [n]$) of the matrix $A$ is defined as
\begin{align}
    \|A\|_{\mathrm{KF},R}\coloneqq\sum^R_{i=1} D_i. 
\end{align}
We denote as $\mathrm{O}(2n)$ the group of real orthogonal $2n \times 2n$ matrices. 
We denote the $n$-qubits Pauli operators as the elements of the set $\{I,X,Y,Z\}^{\otimes n}$, where $I,X,Y,Z$ represent the standard single qubits Pauli. Pauli operators are traceless, Hermitian, square to the identity and form an orthogonal basis with respect the Hilbert-Schmidt scalar product for the space of linear operators.
We define the set of quantum states as $\mathcal{S}\!\left(\mathbb{C}^d\right)\coloneqq \!\{\rho \in \mathbb{C}^{d\times d} \,:\,\rho \ge 0,\,\Tr(\rho)=1\}$. A quantum state $\rho$ is pure if and only if its rank is one.

\subsection{Free-fermionic states}
In this subsection, we provide definitions and Lemmas on free-fermionic states, which are useful for deriving our results. While we define these concepts in terms of qubits, it is worth noting that they can also be directly expressed in terms of fermions via the Jordan-Wigner mapping. In the following we consider an $n$-qubits system (or equivalently, $n$ fermionic modes). 
Throughout this discussion, we focus on an $n$-qubit system. To begin, we introduce the definition of Majorana operators in relation to Pauli matrices.

\begin{definition}[Majorana operators]
For each \(k \in \left[n\right]\), Majorana operators can be defined as
\begin{align}
    \gamma_{2 k-1}:=\left(\prod_{j=1}^{k-1} Z_j\right) X_k, \quad \gamma_{2 k}:=\left(\prod_{j=1}^{k-1} Z_k\right) Y_k.
\end{align}
\label{def:majo}
\end{definition}

As can be readily verified, Majorana operators are Hermitian, traceless, and they square to the identity
\begin{align}
    \gamma_\mu=\gamma^\dagger_\mu, \quad \Tr(\gamma_\mu)=0, \quad \gamma^2_\mu=I
\end{align}
for all \(\mu \in [2n]\). Moreover, they anti-commute and are orthogonal with respect to the Hilbert-Schmidt inner product
\begin{align}
\{\gamma_\mu,\gamma_\nu\}=2\delta_{\mu,\nu} I ,\quad
    \langle \gamma_\mu,\gamma_\nu \rangle_{HS}=2^n \delta_{\mu,\nu}
\end{align}
for all $\mu,\nu \in [2n]$. As such, operators defined by Jordan-Wigner transformation conforms with the Majorana anti-commutation relations. A useful and easy to verify identity is $i Z_j=\gamma_{2j-1}\gamma_{2j}$. 

\begin{definition}[Majorana products]
    Let \(S\) be the set \(S := \{\mu_1,\dots,\mu_{|S|}\} \subseteq [2n]\) with \(1\le\mu_1 <\dots < \mu_{|S|}\le 2n \). 
    We define  $$\gamma_S:=\gamma_{\mu_1}\cdots\gamma_{\mu_{|S|}}$$ if $S\neq \emptyset$ and  $\gamma_{\emptyset}=I$ otherwise.
\end{definition}
It is worth noting that the number of different sets \(S\in[2n]\) and hence Majorana products is \(4^n\).
For any set $S,S^\prime \subseteq [2n]$, Majorana products are orthogonal $$\hs{\gamma_S} {\gamma_{S^\prime}}=2^n\delta_{S,S^\prime}.$$
Hence, they form a basis for $\mathbb{C}^{d\times d}$.

\begin{definition}[Free-fermionic unitary]
\label{def:freeuni}
    For any orthogonal matrix $Q\in \mathrm{O}(2n)$, a free-fermionic unitary $U_Q$ (also known as Gaussian unitary) is a unitary which satisfies
\begin{align}
        U^\dagger_Q\gamma_\mu U_Q= \sum^{2n}_{\nu=1} Q_{\mu,\nu} \gamma_\nu
\end{align}
    for any $\mu \in [2n]$.
\end{definition}
From this definition, it follows that $U^{\dagger}_{Q}=U_{Q^T}$.
Since product of Majorana operators $\gamma_{\mu}$ with $\mu \in [2n]$ form a basis for the linear operators $\mathbb{C}^{d \times d}$, it suffices to specify how a unitary acts on the $2n$ Majorana operators  $\gamma_{\mu}$ with $\mu \in [2n]$ to uniquely specify the unitary up to a phase. Specifically, for a given orthogonal matrix, there exists a known exact implementation of the associated free-fermionic unitary, which can be achieved using either $O(n^2)$ local $2$-qubit gates~\cite{Jiang_2018} or $O(n^2)$ local $2$-modes free-fermionic unitary Majorana evolutions~\cite{dias2023classical}.
From the previous definition, the following statement follows.

\begin{lemma}[Adjoint action of Gaussian unitary on a Majorana product]
\label{le:AdjProd}
For any $S \subseteq [2n]$ and $U_Q$ free-fermionic unitary with $Q\in \mathrm{O}(2n)$, we have
    \begin{align}
        U^\dagger_Q \gamma_S U_Q = \sum_{S^\prime \subseteq \binom{[2n]}{|S|}}\det(Q\rvert_{S,S^\prime})\gamma_{S^\prime}
        \label{prop:detQ}
    \end{align}
where $\binom{[2n]}{|S|}$ is defined as the set of subsets of $[2n]$ of cardinality $|S|$, while $Q\rvert_{S,S^\prime}$ is the restriction of the matrix $Q$ to rows and columns indexed by $S$ and $S^\prime$, respectively. 
\end{lemma}
\begin{proof}
See Ref.~\cite{Chapman_2018} for a proof.
\end{proof}

We now give the definition of free-fermionic states (also known as fermionic Gaussian states).
\begin{definition}[Free-fermionic states]
\label{def:freestate}
A free-fermionic state (also known as Gaussian state) in the Jordan-Wigner mapping can be defined as
    \begin{align}
         \rho=U_Q \rho_0 U^{\dagger}_Q, \quad \text{ with }  \quad \rho_0:=\bigotimes^n_{j=1}\left(\frac{I+\lambda_j Z_j }{2}\right) 
\end{align}
where \(\lambda_j \in [-1,1]\) for each $j\in[n]$ and $U_Q$ is the free-fermionic unitary associated to the orthogonal matrix $Q\in \mathrm{O}(2n)$.
\end{definition}
Since $X_1$ and $\{X_j X_{j+1}\}^{n}_{j=1}=\{-i\gamma_{2j} \gamma_{2j+1}\}^{n}_{j=1}$ are fermionic Gaussian unitaries, and the product of Gaussian unitaries is Gaussian, it follows that $\{X_j\}^{n}_{j=1}$ are Gaussian unitaries.
Thus, using that $X_jZ_jX_j=-Z_j$, without loss of generality, $\{\lambda_j\}^n_{j=1}$ can be assumed to be positive.

The definition~\ref{def:freestate} is slightly more general than defining a free-fermionic state like a state of the form $\exp(-\beta H)/\Tr(\exp(-\beta H)))$, where $\beta >0$ and $H$ is a quadratic Hamiltonian in the Majorana operators, i.e., an Hermitian operator of the form $H=\sum^n_{\mu,\nu=1} h_{\mu,\nu}\gamma_\mu \gamma_{\nu}$, with $h\in \mathbb{C}^{2n\times 2n}$.
From the previous definition, it follows that $\rho$ is pure 
if and only if $\lambda_j\in\{-1,1\}$ for each $j\in [n]$. 
In this pure case, the state vector will be $U_Q \ket{x}$, where $\ket{x}:=\bigotimes^{n}_{i=1} \ket{x_i}$ is a computational basis state 
vector with $x_{i}:=(1-\lambda_i)/2$.
 Since Pauli $X$ matrices are Gaussian unitaries,, without loss of generality, any pure fermionic Gaussian state can be written as $U_Q \ket{0^n}$, which is uniquely specified by an orthogonal matrix $Q\in\mathrm{O}(2n)$. 

%, 
We can define now the correlation matrix of any (possibly non-free-fermionic) state.
\begin{definition}[Correlation matrix]
    Given a (general) state $\rho$, we define its correlation matrix $\Gamma(\rho)$ as
    \begin{align}
        [\Gamma(\rho)]_{j,k}:=-\frac{i}{2}\Tr\left(\left[\gamma_j,\gamma_k\right]\rho\right),
    \end{align}
    where $j,k\in[2n].$
\end{definition}
The correlation matrix of any state is real and anti-symmetric, thus it has eigenvalues in pairs of the form $\pm i \lambda_j$ for $j\in[2n]$, where $\lambda_j$ are real numbers such that $|\lambda_j|\le 1$.
Under free-fermionic unitaries, the correlation matrix of any quantum state changes by the adjoint action with
%by conjugation with 
the associated orthogonal matrix, as expressed in the following proposition. 

\begin{lemma}[Transformation of the correlation matrix under free-fermionic unitary]
\label{prop:transfFGU}
    For any state $\rho$, we have
    \begin{align}
        \Gamma(U_Q\rho U_Q^{\dagger})= Q  \Gamma(\rho) Q^{T},
    \end{align}
    for any free-fermionic unitary $U_Q$ associated to an orthogonal matrix $Q\in \mathrm{O}(2n)$.
\end{lemma}
This is easily verified by the definition of correlation matrix and free-fermionic unitary.
We denote as $\Lambda$ the correlation 
matrix of the 
$\ket{0^n}$ state vector, 
namely
\begin{align}
    \Lambda:=\bigoplus_{j = 1}^{n} \begin{pmatrix} 0 &  1 \\ -1 & 0 \end{pmatrix}=\bigoplus_{j = 1}^{n} (i Y).
    \label{eq:zeroCORR}
\end{align}
The correlation matrix of a computational basis state vector $\ket{x}$ with $x\in \{0,1\}^n$ is $\Gamma(\ket{x})=\bigoplus_{j = 1}^{n} (-1)^{x_i}(i Y)$.
It turns out that any real anti-symmetric matrix can be diagonalized with an orthogonal matrix, in particular, we have the following result.

\begin{lemma}[Normal decomposition of real anti-symmetric matrices \cite{Surace_2022}]
\label{le:normal}
Any real anti-symmetric matrix $\Gamma$ can be decomposed in the so-called \emph{normal} form 
\begin{align}
    \Gamma=Q\bigoplus_{j = 1}^{n} \begin{pmatrix} 0 &  \lambda_j \\ -\lambda_j & 0 \end{pmatrix}Q^T,
\end{align}
for an orthogonal matrix $Q\in \mathrm{O}(2n)$ and $\{\lambda_j\}^{n}_{j=1} \in \mathbb{R}$ real numbers ordered in increasing order. Thus, the eigenvalues of $\,\Gamma$ are $\pm i \lambda_j$ for any $j \in[n]$. We denote $\{\lambda_j\}^{n}_{j=1}$ as normal eigenvalues. 
\end{lemma}

The decomposition in Lemma~\ref{le:normal} can be always performed in such a way that the normal eigenvalues are all non-negative with a suitable orthogonal transformation.

Also, in the case that the normal eigenvalues are not imposed to be all non-negative, the decomposition in Lemma~\ref{le:normal} can be always performed in such a way that the orthogonal matrix $Q\in \mathrm{SO}(2n)$ (i.e., its determinant is one). In fact, if that is not the case, then we can write $Q=Q^{\prime} \mathrm{diag}(-1,1,\dots,1)$, for a matrix $Q^{\prime}\in \mathrm{SO}(2n)$ (note that this matrix $Q^{\prime}$ exists because the product of two orthogonal matrix is an orthogonal matrix and because of the Cauchy–Binet formula of the determinant).

Using Lemma~\ref{le:normal}, we can show that the correlation matrix of any state $\rho$ has normal eigenvalues less than one (and thus also the absolute values of its eigenvalues).
In fact, let $ \Gamma(\rho)=Q (\bigoplus_{j = 1}^{n} i\lambda_j Y) Q^T $ be the correlation matrix expressed in its normal form. Then, we have
\begin{align}
     \lambda_j= (Q^{T} \Gamma(\rho) Q)_{2j-1,2j}=  (\Gamma(U^{\dag}_Q\rho U_Q))_{2j-1,2j}  =\Tr(U^{\dag}_Q\rho U_Q Z_j).
\end{align}
Thus, because of H\"older inequality we have $|\lambda_j|\le \| U^{\dag}_Q\rho U_Q\|_1\norm{Z_j}_{\infty }=1$.
Furthermore, by using Lemma~\ref{le:normal}, we establish a bijection between the set of free-fermionic states and the set of real-anti-symmetric matrices with eigenvalues smaller than one in absolute value.

\begin{lemma}[Bijection between free-fermionic states and correlation matrices]
\label{le:bijection}
Given a free-fermionic state of the form
 \begin{align}
         \rho=U_Q \rho_0 U^{\dagger}_Q, \quad \text{ with }  \quad \rho_0:=\bigotimes^n_{j=1}\left(\frac{I+\lambda_j Z_j }{2}\right) 
         \label{eq:rho}
\end{align}
where \(\lambda_j \in [-1,1]\) for each $j\in[n]$ and $U_Q$ is the free-fermionic unitary associated to the orthogonal matrix $Q\in \mathrm{O}(2n)$, then its correlation matrix is
\begin{align}
    \Gamma = Q \bigoplus_{j = 1}^{n} \begin{pmatrix} 0 & \lambda_j \\ -\lambda_j & 0 \end{pmatrix} Q^T,
    \label{eq:cdec}
\end{align}
which is a real and anti-symmetric matrix, with eigenvalues \(\pm i\lambda_j\) in pairs such that \(\lambda_j \in [-1,1]\) for any \(j \in [n]\).
Conversely, given a real, anti-symmetric matrix \(\Gamma\), it can be decomposed as in Eq.~\eqref{eq:cdec}, in particular, its eigenvalues are of the form \(\pm i \lambda_j\) in pairs. If \(\lambda_j \in [-1,1]\), then \(\Gamma\) uniquely defines a state \(\rho\) of the form of Eq.~\eqref{eq:rho}.
\end{lemma}
The previous theorem ensures that by specifying a valid correlation matrix 
(i.e., 
real, anti-symmetric, with eigenvalues smaller than one), we uniquely specify a free-fermionic state. Vice versa, having a free-fermionic state, it uniquely defines a correlation matrix.
In particular, any 
pure free-fermionic state vector $\ket{\psi}=U_Q \ket{0^n}$ is specified by an orthogonal matrix $Q\in \mathrm{O}(2n)$, and its correlation matrix will be $\Gamma(\psi)=Q\Lambda Q^{T}$, where $\Lambda$ is defined in Eq.~\eqref{eq:zeroCORR}. From this, it follows:
\begin{remark}
\label{rem:detORT}
The correlation matrix $\Gamma(\psi)$ of a pure free-fermionic state $\psi$ satisfies $\det(\Gamma(\psi)) = 1$, is an orthogonal matrix and its normal eigenvalues are all equal to one in absolute value.
\end{remark}
Moreover, we also have that the rank of a mixed free-fermionic state is related to the number of normal eigenvalues strictly smaller than one in absolute value, as it follows by Lemma~\ref{le:bijection}. 
\begin{remark}[Relation between rank and normal eigenvalues of a free-fermionic state] 
\label{le:rankEigs}
Let $\rho$ be a free-fermionic state, expressed as $\rho=U_Q \left(\bigotimes^n_{j=1}(I+\lambda_j Z_j)/2\right) U_Q^{\dagger}$, where \(\{\lambda_j\}^n_{j=1} \subseteq [-1,1]\) and $U_Q$ is a free-fermionic unitary. Let $m$ be the number of elements in $\{\lambda_1,\dots,\lambda_n\}$ that are in absolute value smaller than one. We then have $\mathrm{rank}(\rho)=2^m$.
\end{remark}
In our analysis the notion of \emph{Pfaffian} will be useful.
\begin{definition}[Pfaffian of a matrix]
Let $C$ be a $2n\times 2n$ anti-symmetric matrix. Its Pfaffian is defined as
\begin{align}
\operatorname{Pf}(C)=\frac{1}{2^n n !} \sum_{\sigma \in S_{2 n}} \operatorname{sgn}(\sigma) \prod_{i=1}^n C_{\sigma(2 i-1), \sigma(2 i)},
\end{align}
where $S_{2 n}$ is the symmetric group of order $(2 n) !$ and $\operatorname{sgn}(\sigma)$ is the signature of $\sigma$. 
The Pfaffian of an $m\times m$ anti-symmetric matrix with $m$ odd is defined to be zero.
\end{definition}

Well-known properties are the following. 
For any matrix \(B\), we have \(\operatorname{Pf}(BCB^T) = \det(B) \operatorname{Pf}(C)\) and $\operatorname{Pf}(\lambda C) = \lambda^{n} 
 \operatorname{Pf}(C)$, where $C$ is a $2n\times 2n$ anti-symmetric matrix and $\lambda \in \mathbb{C}$.
Moreover, it holds that \(\operatorname{Pf}(C)^2 = \det(C)\) (note that this is consistent with the fact that the Pfaffian of an odd anti-symmetric matrix is defined to be zero, since the determinant of an odd anti-symmetric matrix is zero).
Another useful identity is 
\begin{align}
\operatorname{Pf}\left(\bigoplus_{j = 1}^{n} \begin{pmatrix} 0 &  \lambda_j \\ -\lambda_j & 0 \end{pmatrix}\right)= \prod^n_{j=1} \lambda_j.
\end{align}

Now we recall the well-known Wick's theorem, which states that the any  Majorana product expectation value over a free-fermionic state can be computed efficiently given access to its correlation matrix. 
\begin{lemma}[Wick's Theorem \cite{Bravyi_2017,Surace_2022}]
\label{le:Wick}
    Let $\rho$ be a free-fermionic state with the associated correlation matrix $\Gamma(\rho)$. Then, we have
    \begin{align}
        \Tr(\gamma_S \rho) = i^{|S|/2}\operatorname{Pf}( \Gamma(\rho)\rvert_{S}),
    \end{align}
    where $\gamma_S = \gamma_{\mu_1}\cdots\gamma_{\mu_{|S|}}$, and $S = \{\mu_1, \dots, \mu_{|S|}\} \subseteq [2n]$ with $1 \leq \mu_1 < \dots < \mu_{|S|} \leq 2n$, while $\Gamma(\rho)\rvert_{S}$ is the restriction of the matrix $\Gamma(\rho)$ to the rows and columns corresponding to elements in $S$. 
\end{lemma}
Note that, since for any $S\subseteq [2n]$ the restriction of the correlation matrix $\Gamma\rvert_{S}$ is still anti-symmetric, its Pfaffian is well-defined. The Pfaffian of $\Gamma\rvert_{S}$ can be computed efficiently in time $\Theta(|S|^{3})$.

\begin{definition}[Parity]
\label{def:paritydef}
The parity operator is defined as the operator
\begin{align}
    Z^{\otimes n }:=(-1)^{n} \gamma_{1}\gamma_{2}\cdots\gamma_{2n-1}\gamma_{2n}.
\end{align}
The parity of a quantum state $\rho$ is defined as the expectation values of the parity operator, i.e., $$\mathrm{Parity}(\rho)\coloneqq\Tr(Z^{\otimes n}\rho).$$
\end{definition}

\begin{lemma}[Gaussian states are eigenstates of the parity operator]
\label{le:parity}
    Any Gaussian state vector  $\ket{\psi_Q}\coloneqq U_Q\ket{0^{n}}$ associated with the orthogonal matrix $Q\in \mathrm{O}(2n)$ satisfies:
    \begin{align}
        Z^{\otimes n}\ket{\psi_Q}=\det(Q)\ket{\psi_Q}.
    \end{align}
\end{lemma}
Thus, we also have $\mathrm{Parity}(\psi_{Q})=\det(Q)$.
\begin{proof}
We have
\begin{align}
    Z^{\otimes n}\ket{\psi_Q} &= (-i)^n \gamma_{[2n]} \ket{\psi_Q} = (-i)^n U_Q U_Q^{\dagger} \gamma_{[2n]} U_Q \ket{0^n} = (-i)^n \det(Q) U_Q \gamma_{[2n]} \ket{0^n} \\
    & = \det(Q) U_Q Z^{\otimes n} \ket{0^n}= \det(Q)\ket{\psi_Q},\nonumber
\end{align}
where the first step follows from Jordan-Wigner, and we have defined $\gamma_{[2n]} := \prod^{n}_{j=1}\gamma_j$. The third step follows from Eq.~\eqref{prop:detQ}.
\end{proof}
From this lemma, it follows that if two pure fermionic Gaussian states have different parity, then they have also zero overlap (however, the converse 
is not true: Two pure Gaussian state vectors with the same parity can have zero overlap, e.g., $\ket{0,0}$ and $\ket{1,1}$). 
This result can also be derived using Wick's Theorem (Lemma \ref{le:Wick}). In fact, for a possibly mixed free-fermionic state $\rho$, we have
\begin{align}
    \mathrm{Parity}(\rho)=\Tr(Z^{\otimes n}\rho) = \Pf(\Gamma(\rho)).
\end{align}
Furthermore, according to the properties of the Pfaffian, we obtain
\begin{align}
\Pf(\Gamma(\rho)) = (\prod^n_{j=1}\lambda_j)\det(Q),
\end{align}
where $Q \in \mathrm{O}(2n)$ and $\{\lambda_j\}^n_{j=1}$ are, respectively, the orthogonal matrix and the normal eigenvalues associated with the normal form of $\Gamma(\rho)$. If the Gaussian state $\rho$ is pure, then all the normal eigenvalues are equal to one, hence the parity is equal to $\det(Q)$.
We now mention an important formula that relates the overlap of two pure fermionic Gaussian states to their correlation matrices~\cite{Bravyi_2017}.
\begin{lemma}[Overlap between two pure Gaussian states~\cite{Bravyi_2017}]
\label{le:overlapBr}
The overlap between two Gaussian state vectors $\ket{\psi_1}, \ket{\psi_2}$ is
\begin{align}
    |\braket{\psi_1}{\psi_2}|^2 = \left| \Pf\left(\frac{1}{2}(\Gamma(\psi_1) + \Gamma(\psi_2))\right)\right|.
\end{align}
\end{lemma}
It can be shown that the preceding formula is consistent with the fact that two Gaussian states with opposite parity have zero overlap: in fact, the Pfaffian of the sum of the correlation matrices associated to the two pure Gaussian states with opposite parity must be zero, as follows from Corollary 2.4.(b) of Ref.~\cite{TAM2010412}.

In our discussion, we also leverage the following two well-known lemmas from the literature~\cite{Surace_2022}.
\begin{lemma}[Reduced state of a free-fermionic state is free-fermionic]
\label{le:reduced} 
Let \( \rho \) be an \( n \)-qubit free-fermionic state. Let \( r \leq n \). Then the reduced state \( \rho_r \), obtained by tracing out the last \( n-r \) qubits, is a free-fermionic state.
\end{lemma}
\begin{proof}
The proof follows by expressing the reduced state \( \rho_r \) in the Majorana basis and applying Wick's theorem, which holds for the free-fermionic state \( \rho \). The conclusion follows from the fact that any quantum state whose coefficients, when expressed in the Majorana basis, are consistent with Wick's theorem, is a free-fermionic state.
\end{proof}
The following lemma is available in the literature in various forms (for instance, see Ref.\ \cite{Windt_2021}).
\begin{lemma}[Purification of a mixed free-fermionic state.]
\label{le:purification}
    Any $n$-qubits mixed free-fermionic state $\rho$ can be purified into a pure free-fermionic state vector  $\ket{\psi_\rho}$ of $2n$-qubits. Specifically, the purified state vector $\ket{\psi_{\rho}}$ corresponds to the free-fermionic state associated with the correlation matrix
\begin{align}
    \Gamma(\psi_\rho) = \begin{pmatrix}
        \Gamma(\rho) & \sqrt{I_{2n} + \Gamma^2(\rho)} \\
        -\sqrt{I_{2n} + \Gamma^2(\rho)} & -\Gamma(\rho)
    \end{pmatrix}.
\end{align}
\end{lemma}
\begin{proof}
Since any free-fermionic state is fully specified by its correlation matrix, we only need to demonstrate that $\Gamma(\ketbra{\psi}{\psi})$ is a valid correlation matrix corresponding to a pure state, and that the partial trace of this pure state corresponds to $\rho$. This latter assertion is trivial, as taking the partial trace with respect to qubits indexed by $E := \{n+1, \dots, 2n\}$ of $\ketbra{\psi}{\psi}$ involves considering only the restriction of $\Gamma(\ketbra{\psi}{\psi})$ to the first diagonal block, corresponding to $\Gamma(\rho)$, and thus yielding the state $\rho$~\cite{Surace_2022}.
Therefore, we are left to demonstrate that $\Gamma(\ketbra{\psi}{\psi})$ is a valid correlation matrix corresponding to a pure state. In particular, we need to show that it is real-anti-symmetric with eigenvalues $\{\pm i\}_{j=1}^{2n}$.
The fact that it is anti-symmetric follows from the fact that $\sqrt{I_{2n} + \Gamma^2(\rho)}$ is Hermitian. Since it is anti-symmetric and real, its eigenvalues are purely imaginary.
We are left to show that the eigenvalues are all one in absolute value. This follows from the fact that $\Gamma(\ketbra{\psi}{\psi})$ is unitary, as can be verified by inspection.
\end{proof}

\subsection{Particle-number preserving fermionic states}
In this section, we introduce standard notions related to particle-number preserving fermionic states, which form a common subset of fermionic states often considered in condensed matter physics. This concept will be useful for deriving a particle-number preserving version of the inequality we establish for general free-fermionic states in the subsequent section.  
\begin{definition}[Creation and annihilation operators]
The annihilation operators $\{a_{j}\}^n_{j=1}$ and creation operators $\{a^{\dagger}_{j}\}^n_{j=1}$ are defined as
\begin{align}
    a_{j}  \coloneqq   \frac{\gamma_{2j-1} + i\gamma_{2j}}{2}, \quad a^{\dagger}_{j}  \coloneqq   \frac{\gamma_{2j-1} - i\gamma_{2j}}{2}
\end{align}
for all $j \in [n]$.  
\end{definition}
They satisfy the commutation relations $ \{a_{k}, a_{l}\}=0$, $\{a_{k}, a_{l}^{\dagger}\}=\delta_{k,l}$ for each $k,l\in [n]$.
Moreover, they satisfy $a_k^{\dag} a_k = \frac{1}{2}(I-Z_j) = \ketbra{1}{1}_{k} $ for each $k\in [n]$. Majorana operators can be then written as $\gamma_{2k-1}=a_j + a^{\dagger}_j$ and $\gamma_{2j}=-i(a_j-a^{\dagger}_j)$ for each $k\in [n]$.
\begin{definition}[Particle number operator]
The operator $\hat{N}  \coloneqq   \sum^n_{i=1} a^{\dag}_ia_i$ is denoted as the particle number operator.
\end{definition}
The computational basis forms a set of eigenstates for the particle number operator
\begin{align}
    \hat{N}\ket{x_1,\dots, x_n} = (x_1+\dots +x_n)\ket{x_1,\dots,x_n},
\end{align}
where $x_1,\dots, x_n \in \{0,1\}$. The eigenvalue $|x|\coloneqq x_1+\dots +x_n$ of the particle-number operator is the Hamming weight of the bitstring $x \coloneqq (x_1,\dots, x_n)$.

\begin{definition}[Particle number preserving states]
\label{def:particle}
A state $\rho$ is is said to be particle number preserving if and only if it commutes with the particle number operator, i.e., $[\hat{N}, \rho] = 0$.
\end{definition}
From this definition, it follows that any particle-number preserving pure state is an eigenstate of the particle-number operator $\hat{N}$, and, more generally, any particle-number preserving mixed state can be written as a convex combination of eigenstates of the particle-number operator.
We now define the particle-number preserving correlation matrix of a quantum state $\rho$. 
\begin{definition}[Particle-number preserving correlation matrix]
Given a state $\rho$, we define its particle number-preserving correlation matrix $C(\rho)$ as the $n\times n$ matrix such that, for 
each $j,k\in[n]$, we have
    \begin{align}
        [C(\rho)]_{j,k}:=\Tr\left(a^{\dagger}_j a_k\rho\right).
    \end{align}
\end{definition}

It is easy to see that the particle-number preserving correlation matrix is a Hermitian matrix.
In fact, for each $j,k \in [n]$, we have
    \begin{align}
        [C(\rho)]^*_{j,k}=\Tr\left((a^{\dagger}_j a_k\rho)^{*}\right)=\Tr\left((a^{\dagger}_j a_k\rho)^{\dagger}\right)=\Tr\left(\rho a^{\dagger}_k a_j\right)=[C(\rho)]_{k,j}.
    \end{align}
Thus, it can be unitarily diagonalized.

\begin{lemma}[Relation between the correlation matrix and the particle-number preserving one of a particle-number preserving state]
Let $\rho$ be a particle-number preserving quantum state. The correlation matrix $\Gamma(\rho)$ can be written in terms of the particle-number preserving correlation matrix $C(\rho)$ as follows:
\begin{align}
    \Gamma(\rho)= (I-2\Re(C(\rho))\otimes iY + (2\Im(C(\rho)))\otimes I,
    \label{eq:corrppG}
\end{align}
where $\Re(\cdot)$ and $\Im(\cdot)$ are the entry-wise real and immaginary part.
\begin{proof}
Since $\rho$ is a particle-number preserving quantum state, it commutes with the particle-number operator $[\hat{N},\rho]=0$, which implies that $\rho$ can be diagonalized as $\rho=\sum^{2^n}_{j=1} p_j \ketbra{\psi_j}$, where $\{p_j\}^{2^n}_{j=1}$ are non-negative numbers which add up to one and $\{\ket{\psi_j}\}^{2^n}_{j=1}$ are eigenstates of the particle number operators.
Hence, for all $j,k\in[n]$, we have that
\begin{align}
\label{eq:diffnumber}
    \Tr(\rho a_i^{\dagger} a^{\dagger}_{j})=\sum^{2^n}_{m=1} p_m \bra{\psi_m}a_i^{\dagger} a^{\dagger}_{j}\ket{\psi_m}=0,
\end{align}
where the last equality follows from the fact that $\ket{\psi_m}$ and $a_i^{\dagger} a^{\dagger}_{j}\ket{\psi_m}$ are two vectors which are in the span of two different orthogonal eigenspaces of the particle number operator (i.e., associated with two different Hamming weight). Similarly, we have $\Tr(\rho a_i a_{j})=0$ for each $i,j \in [n]$.
For all $j,k \in [n]$, we have
\begin{align}
    [\Gamma(\rho)]_{2j-1,2k}&=-i \Tr(\gamma_{2j-1} \gamma_{2k}\rho )= \Tr((a_{j} + a^{\dagger}_{j})(a^{\dagger}_{k}-a_{k})\rho )=\Tr(a_{j}a^{\dag}_{k}\rho)-\Tr(a^{\dag}_{j}a_{k}\rho)
    \nonumber
    \\
    \nonumber
    &=\delta_{j,k} -\Tr(a^{\dag}_{k}a_{j}\rho)-\Tr(a^{\dag}_{j}a_{k}\rho)=\delta_{j,k}-[C(\rho)]^{*}_{j,k}-[C(\rho)]_{j,k}\\&=[I-2\Re(C(\rho)]_{j,k},
    \label{eq:firstel}
\end{align}
where in the third step we have used Eq.~\eqref{eq:diffnumber}, in the fourth step we have used the commutation relation of annihilation and creation operators, in the fifth step we have used the definition of $C(\rho)$ and the fact that it is Hermitian.  
For all $j,k \in [n]$, we also have
\begin{align}
    [\Gamma(\rho)]_{2j,2k-1}&=-[\Gamma(\rho)]_{2k-1,2j}=-[I-2\Re(C(\rho)]_{k,j}=-[I-2\Re(C(\rho)]_{j,k},
\end{align}
where in the first step we use the fact that $\Gamma(\rho)$ is anti-symmetric, in the second step we have used Eq.\ \eqref{eq:firstel}, and in the last step the fact that $I-2\Re(C(\rho)$ is a symmetric matrix.
Moreover, for all $j\neq k \in [n] $, it follows that
\begin{align}
    [\Gamma(\rho)]_{2j-1,2k-1}&=-i \Tr(\gamma_{2j-1} \gamma_{2k-1}\rho )= -i\Tr((a_{j} + a^{\dagger}_{j})(a_{k}+a^{\dagger}_{k})\rho )=-i\Tr(a_{j}a^{\dag}_{k}\rho)-i\Tr(a^{\dag}_{j}a_{k}\rho)\\
    \nonumber
    &=i[C(\rho)]^{*}_{j,k}-i[C(\rho)]_{j,k}=[2\Im(C(\rho)]_{j,k}.
\end{align}
Similarly, we also have, for all $j\neq k \in [n]$,
\begin{align}
    [\Gamma(\rho)]_{2j,2k}&=-i \Tr(\gamma_{2j} \gamma_{2k}\rho )= i\Tr((a_{j} - a^{\dagger}_{j})(a_{k}-a^{\dagger}_{k})\rho )=-i\Tr(a_{j}a^{\dag}_{k}\rho)-i\Tr(a^{\dag}_{j}a_{k}\rho)=[2\Im(C(\rho)]_{j,k}.
\end{align}
Thus, Eq.\ \eqref{eq:corrppG} follows.
\end{proof}
\end{lemma}

From the previous Lemma, it follows that for two particle-number preserving state $\rho$ and $\sigma$ and any $p$-norms with $p\in [1,\infty]$, we have
\begin{align}
\nonumber
    \norm{\Gamma(\rho)-\Gamma(\sigma)}_p&=2\norm{ \Re(C(\rho)-C(\sigma))\otimes iY + \Im(C(\rho)-C(\sigma))\otimes I}_{p}\\
    \nonumber
    &\le 2\, 2^{1/p} \norm{ \Re(C(\rho)-C(\sigma))}_{p} +  2\, 2^{1/p}\norm{\Im(C(\rho)-C(\sigma))}_{p},
    \\&\le 4\, 2^{1/p} \norm{ C(\rho)-C(\sigma)}_{p} ,
    \label{eq:partnumbUPP}
\end{align}
where in the second step we have used the triangle inequality and the fact that $\norm{A\otimes B}_p=\norm{A}_p\norm{A}_p$, and in the last step the fact that $\Re(A)=(A+A^{*})/2$ and $\Im(A)=-i(A-A^{*})/2$ .  
\begin{lemma}[Relation between eigenvalues of the correlation matrices]
\label{le:relEIG}
Let $\rho$ be a particle-number preserving quantum state. Let $\{D_j\}^{n}_{j=1}$ be the eigenvalues of the particle-number preserving correlation matrix $C(\rho)$. There exists an orthogonal matrix which puts the correlation matrix $\Gamma(\rho)$ in the normal form (as in Lemma~\ref{le:normal}) with normal eigenvalues 
\begin{align}
    \lambda_j= 1-2D_j,
\end{align} 
for each $j\in [n]$.
\end{lemma}
\begin{proof}
Since $C(\rho)$ is Hermitian, for the spectral theorem, there exists $u\in \mathrm{U}(n)$, such that $C(\rho)= u D u^{\dagger}$, where $D$ is a diagonal (real) matrix.
We now define the matrix $O$
\begin{align}
\label{eq:ortsym}
    O = \Re(u) \otimes I + \Im(u) \otimes iY.
\end{align}
This matrix $O$ is orthogonal (and symplectic), as it can be verified by using that
\begin{align}
\label{eq:relUnOrt}
   \Re(u)\Re(u)^t+\Im(u)\Im(u)^t=I, \quad \quad  \Re(u)\Im(u)^t-\Im(u)\Re(u)^t=0,
\end{align}
which follow from the unitarity of $u$.
Now, using Eq.\ \eqref{eq:ortsym}, Eq.\ \eqref{eq:relUnOrt} and the fact that for particle-preserving states, $\Gamma(\rho)$ can be written in terms of $C(\rho)$ as in Eq.\ \eqref{eq:corrppG}, it can be verified by inspection that
\begin{align}
    O^T\Gamma(\rho)O= \mathrm{diag}(1-2 D_1,\dots, 1-2 D_n)\otimes iY.
\end{align}
Thus, the normal eigenvalues of $\Gamma(\rho)$ are $\{1-2 D_j\}^n_{j=1}$, where $D_j$ are the eigenvalues of $C(\rho)$.
\end{proof}
Using the upper bound in Eq.\ \eqref{eq:partnumbUPP} and the eigenvalues relation in Lemma~\ref{le:relEIG}, we can directly transfer many of the inequalities that we show in the following section to the particle-preserving case.

\section{Norm inequalities for free and non-free fermionic states}\label{Sec:matrixtracedistanceinequalities}
In this section, we derive key relations concerning free-fermionic states, laying the groundwork for our subsequent analysis of property testing and tomography.
Specifically, in Subsection~\ref{App:sub:generallemmas}, we present several results instrumental in proving the main result of this work, some of which were established in previous work. Particularly noteworthy is the exact computation of the derivative of a Gaussian state with respect to an element of its covariance matrix, which proves to be fundamental for deriving the optimal trace-distance bound. In Subsection~\ref{App.subb:purestatesbounds}, we present trace distance bounds between Gaussian states, specifically restricted to the pure-state scenario. Notably, we derive an optimal trace-distance bound between two pure Gaussian states in terms of their covariance matrices. In Subsection~\ref{App.subb:mixedstatesbounds}, we extend these results to the mixed-state scenario by providing a similar optimal trace-distance bound for mixed states. To conclude, in Subsection~\ref{sec:non-Gaussianity}, we present lower bounds for both pure and mixed states on the minimum trace-distance between a generic state and the sets of Gaussian states $\mathcal{G}_{\text{pure}}$, $\mathcal{G}_{\text{mixed}}$, and $\mathcal{G}_{R}$, corresponding to pure, mixed, and bounded-rank-$R$ Gaussian states, respectively. These bounds, in addition to being instrumental for the testing application, provide measurable lower bounds on nongaussianity. 
\subsection{General lemmas}\label{App:sub:generallemmas}

\begin{lemma}[Gentle measurement lemma (or quantum union bound) \cite{Gao_2015}]
\label{le:gentle}
Let $\varepsilon_1, \dots, \varepsilon_M > 0$, where $M \in \mathbb{N}$. Consider the projectors $\{P_i\}^{M}_{i=1}$. Let $\rho$ be a quantum state. If $\,\operatorname{Tr}(P_i\rho) \geq 1-\varepsilon_i$ holds for all $i \in [n]$, then
\begin{align}
    \left\|\rho - \frac{P_n\dots P_1\rho P_1\dots P_M}{\Tr(P_n\dots P_1\rho P_1\dots P_M)}\right\|_{1} \leq 2\sqrt{\sum_{i\in [M]}\varepsilon_i}\,.
\end{align}
\end{lemma}

\begin{lemma}[Perturbation bounds on normal eigenvalues]
\label{le:lbanti-symm}
Let $A$ and $B$ be two $2n\times 2n$ anti-symmetric real matrices with eigenvalues $\{\pm i \lambda_{k}(A)\}^{n}_{k=1}$ and $\{\pm i \lambda_{k}(B)\}^{n}_{k=1}$, 
respectively, where $0\le \lambda_1(A)\le \dots \le\lambda_n(A)$ and $0\le \lambda_1(B)\le \dots \le\lambda_n(B)$. We then have
\begin{align}
    \norm{A-B}_{\infty}\ge |\lambda_{k}(A)- \lambda_{k}(B)|,
\end{align}
for any $k\in [n]$.
\end{lemma}
\begin{proof}
This follows from the fact that $C:=iA$ and $D:=iB$ are Hermitian matrices. Applying \emph{Weyl's perturbation theorem} (see Ref.~\cite{bhatia1996matrix}, Corollary III.2.6), which states that given two $2n\times 2n$ Hermitian matrices $C$ and $D$ with eigenvalues $c_1\le\dots\le c_{2n}$ and $d_1\le \dots\le d_{2n}$, we have
\begin{align}
    \norm{C-D}_{\infty}\ge |c_{j}- d_{j}|,
\end{align}
for any $j\in[2n]$. This implies that
\begin{align}
    \norm{A-B}_{\infty}= \norm{C-D}_{\infty}\ge \max_{j\in [2n]}|c_{j}- d_{j}|=\max_{k\in [n]}|\lambda_{k}(A)- \lambda_{k}(B)|.
\end{align}
\end{proof}

\begin{lemma}[Inequality for anti-symmetric matrices]\label{lem:cov_ineq}
    For any anti-symmetric real matrix $C = -C^T$, the inequality
    \begin{align}
        \|C\|_1^2 + 2\tr(\Lambda C \Lambda C) \geq 2\|C\|_2^2 + \tr(C \Lambda)^2
    \end{align}
     holds,
    where $\Lambda = \mathbb{1} \otimes \begin{pmatrix}
        0 & 1 \\
        -1 & 0
    \end{pmatrix}$, or, more generally, any anti-symmetric orthogonal matrix.
\end{lemma}
\begin{proof}
    We can write $C$ in its normal form $C=O\Lambda_{\lambda}O^{T}$, where 
    \begin{align}
        \Lambda_\lambda\coloneqq \mathrm{diag}(\lambda)\otimes \begin{pmatrix}
        0&1 \\
        -1&0
    \end{pmatrix}\,,
    \end{align}
    where $\lambda$ are the normal eigenvalues of $C$ and $OO^T=\mathbb{1}$ is orthogonal.
    As such, the expression we want to show becomes
    \begin{align}
        \|\Lambda_\lambda\|_1^2+2\tr(\Lambda O \Lambda_\lambda O^T \Lambda O\Lambda_\lambda O^T)-  2\|\Lambda_\lambda\|_2^2- \tr(O \Lambda_\lambda O^T \Lambda)^2\ge 0.
    \end{align}
    We have
     \begin{align}
        &\|\Lambda_\lambda\|_1^2+2\tr(\Lambda O \Lambda_\lambda O^T \Lambda O\Lambda_\lambda O^T)-  2\|\Lambda_\lambda\|_2^2- \tr(O \Lambda_\lambda O^T \Lambda)^2\\
        \nonumber
        &= \Tr(|\Lambda_\lambda|)^2+2\tr(\Lambda O \Lambda_\lambda O^T \Lambda O\Lambda_\lambda O^T)+ 2 \tr(\Lambda_\lambda^2)- \tr(O \Lambda_\lambda O^T \Lambda)^2\\
         \nonumber
        &=  4 (\sum^n_{j=1}\lambda_j)^2+ 2\tr(\Lambda O \Lambda_\lambda O^T \Lambda O\Lambda_\lambda O^T)- 4 \sum^n_{j=1} \lambda^2_j - \tr(O \Lambda_\lambda O^T \Lambda)^2\\
         \nonumber
         &=4(\sum^n_{i=1}\lambda_i)^2+2\tr(\Lambda_\lambda Q \Lambda_\lambda Q )-4\sum^n_{i=1}\lambda_i^2-\tr(\Lambda_\lambda  Q  )^2,
          \nonumber
    \end{align}
    where $Q=O^T\Lambda O$ is also an orthogonal anti-symmetric matrix. This expression can be expanded as
    \begin{align}
    \nonumber
        &4\sum^n_{i,j=1}\lambda_i\lambda_j+4\sum^n_{i,j=1} \lambda_i\lambda_j (-Q_{2i-1,2j-1}Q_{2j,2i}+ Q_{2i-1,2j}Q_{2j-1,2i})-4\sum^n_{i,j=1} \lambda_i\lambda_j Q_{2i-1,2i}Q_{2j-1,2j}-4\sum^n_{i,j=1} \delta_{i,j}\lambda_i\lambda_j\\
        =&4 \sum^n_{i,j=1}\lambda_i\lambda_j (1-Q_{2i-1,2j-1}Q_{2j,2i}+ Q_{2i-1,2j}Q_{2j-1,2i}-Q_{2i-1,2i}Q_{2j-1,2j}-\delta_{i,j})\ge 0\,.
    \end{align}
    For $i=j$, this expression is trivially true. As such, we need that for each $ i\neq j \in [n]$
    \begin{align}
        0&\leq 1-Q_{2i-1,2j-1}Q_{2j,2i}+ Q_{2i-1,2j}Q_{2j-1,2i}-Q_{2i-1,2i}Q_{2j-1,2j}\\
         \nonumber
       & = 1+Q_{2i-1,2j-1}Q_{2i,2j} - Q_{2i-1,2j}Q_{2i,2j-1}-Q_{2i-1,2i}Q_{2j-1,2j}\\
        \nonumber
        &= 1-
        \Pf\begin{pmatrix}
            0&Q_{2i-1,2i}&Q_{2i-1,2j-1}&Q_{2i-1,2j}\\
            -Q_{2i-1,2i}&0&Q_{2i,2j-1}&Q_{2i,2j}\\
            -Q_{2i-1,2j-1}&-Q_{2i,2j-1}&0&Q_{2j-1,2j}\\
            -Q_{2i-1,2j}&-Q_{2i,2j}&-Q_{2j-1,2j}&0\\
    \end{pmatrix}.
     \nonumber
    \end{align}
This corresponds to the Pfaffian of a sub-matrix of \( Q \), denoted \( Q|_{2i-1, 2i, 2j-1, 2j} \). The Pfaffian is equal to the product of all normal eigenvalues of the sub-matrix. 
Since projecting onto a subspace can only reduce the singular values and the largest singular value of the orthogonal matrix $Q$ is one, we obtain the inequality
\begin{equation}
\|PQP^\dagger\|_\infty \leq \|P\|_\infty \|Q\|_\infty \|P^\dagger\|_\infty = 1,
\end{equation}
where \( P \) is the projection matrix associated with the relevant indices. This establishes the desired bound for the Pfaffian of the sub-matrix.
Therefore, $\Pf(Q|_{2i-1,2i,2j-1,2j})\le 1$, which completes the proof.
\end{proof}
\begin{lemma}[Derivative for free-fermionic states]\label{lem:derivativefreefermionic}
Let $\rho(\Gamma)$ be a possibly mixed free-fermionic state associated with the correlation matrix $\Gamma$.
Then, we have 
\begin{align}\label{Eq.derivativeperturbedgaussianstate}
        \partial_{\alpha} \rho(\Gamma+\alpha X)|_{\alpha=0}
        &=-\frac{i}{8}\sum_{a,b} X_{a,b} [\gamma_a,\{\gamma_b,\rho\}] .
    \end{align}
\end{lemma}
\begin{proof}
    Let us now assume the state to be in normal form, i.e. $\Gamma=Q\Lambda_{\Gamma}Q^{T}$. For  the change we only consider the basis element $X=E_{a,b}$, i.e. an anti-symmetric unit matrix. By derivative rules, the general result follows. In normal form, we can write the state as
    \begin{align}
        \rho(\Lambda_{\Gamma})&=\frac{1}{2^n}\prod_{i=1}^n (1-i\lambda_i\gamma_{2i-1}\gamma_{2i}),
    \end{align}
    where $\{\lambda_i\}_{i\in [n]}$ are the normal eigenvalues.
    For the perturbed state, we have
    \begin{align}
        \rho(\Lambda_{\Gamma}+\alpha E_{a,b})&=\begin{cases}
            \frac{1}{2^n}\prod_{i=1}^n (1-i(\lambda_i+\delta_{c,i}\alpha)\gamma_{2i-1}\gamma_{2i}), &\text{if }a=2c-1,b=2c,\\
            \rho -\alpha i\gamma_a\gamma_b  \frac{1}{2^n}\prod_{i=1|i\neq \lceil\frac{a}{2}\rceil,\lceil\frac{b}{2}\rceil}^n (1-i\lambda_i\gamma_{2i-1}\gamma_{2i}) \,.
        \end{cases}
    \end{align}
    The first case follows directly, as the state is still a state in normal form. For the second, we can take all Majorana opertors where the sub matrix corresponds to a non-vanishing Pfaffian. This contains all the previous matrices, but also the additional ones where we select the indices $(a,b)$, without their corresponding paired modes. As such, we have that the derivative is given by
    \begin{align}
        \partial_\alpha\rho(\Lambda_{\Gamma}+\alpha E_{a,b})|_{\alpha=0}&=\begin{cases}
            \frac{-i}{2^n}\gamma_a\gamma_b\prod_{i=1|i\neq c}^n (1-i\lambda_i\gamma_{2i-1}\gamma_{2i}), &\text{if }a=2c-1,b=2c,\\
            \frac{-i}{2^n}\gamma_a\gamma_b  \prod_{i=1|i\neq \lceil\frac{a}{2}\rceil,\lceil\frac{b}{2}\rceil}^n (1-i\lambda_i\gamma_{2i-1}\gamma_{2i}). 
        \end{cases}
    \end{align}
    We can write it in a concise form as
    \be
\partial_\alpha\rho(\Lambda_{\Gamma}+\alpha E_{a,b})|_{\alpha=0}=\frac{-i}{2^n}\gamma_a\gamma_b\prod_{i=1}^n (1-i\lambda_i\gamma_{2i-1}\gamma_{2i}\delta_{i\not\in \{\lceil\frac{a}{2}\rceil,\lceil\frac{b}{2}\rceil\}}).
    \ee
    We proceed to show that Eq.~\eqref{Eq.derivativeperturbedgaussianstate} is indeed the expression for the derivative. First, let us compute $\{\gamma_b,\rho\}$, to get
    \begin{align}
       \{\gamma_b,\rho\}&=\frac{1}{2^n}\left[\gamma_b \prod_{i=1}^n (1-i\lambda_i\gamma_{2i-1}\gamma_{2i})+ \prod_{i=1}^n (1-i\lambda_i\gamma_{2i-1}\gamma_{2i})\gamma_b\right]\\
       \nonumber
       &=\frac{2}{2^n}\gamma_b\prod_{i=1}^n (1-i\lambda_i\gamma_{2i-1}\gamma_{2i}\delta_{i\neq \lceil\frac{b}{2}\rceil}),
    \end{align}
    which follows from the fact that $\gamma_b$ commutes with $\gamma_{2i-1}\gamma_{2i}$ as long as $b\not\in \{2i-1,2i\}$, which can be concisely written as $i\neq \lceil\frac{b}{2}\rceil$.
    Applying the commutator gives (for $a \neq b$)
    \begin{align}
       [\gamma_a,\{\gamma_b,\rho\}]&=\frac{2}{2^n}\left(\gamma_a\gamma_b \prod_{i=1}^n (1-i\lambda_i\gamma_{2i-1}\gamma_{2i}\delta_{i\neq \lceil\frac{b}{2}\rceil})- \gamma_b \prod_{i=1}^n (1-i\lambda_i\gamma_{2i-1}\gamma_{2i}\delta_{i\neq \lceil\frac{b}{2}\rceil}) \gamma_a\right)\\
       \nonumber
       &=\frac{2}{2^n}\left(\gamma_a\gamma_b \prod_{i=1}^n (1-i\lambda_i\gamma_{2i-1}\gamma_{2i}\delta_{i\neq \lceil\frac{b}{2}\rceil})+ \gamma_a\gamma_b \prod_{i=1}^n (1-i\lambda_i\gamma_{2i-1}\gamma_{2i}\delta_{i\neq \lceil\frac{b}{2}\rceil}(-1)^{\delta_{i= \lceil\frac{a}{2}\rceil}}) \right)\\
       \nonumber
       &=\frac{4}{2^n}\gamma_a\gamma_b\prod_{i=1}^n (1-i\lambda_i\gamma_{2i-1}\gamma_{2i}\delta_{i\not\in \{\lceil\frac{a}{2}\rceil,\lceil\frac{b}{2}\rceil\}}),
       \nonumber
    \end{align}
    where we used the fact that $\{\gamma_a,\gamma_b\}=0$ for $a\neq b$. Hence, we can write
    \be
    \partial_\alpha\rho(\Lambda_{\Gamma}+\alpha E_{a,b})|_{\alpha=0}=-\frac{i}{4}[\gamma_a,\{\gamma_b,\rho\}].
    \ee
    As such, we have shown equivalence.
    In general, 
    we can write
    \begin{align}
        \rho(\Gamma+\alpha X)&=\rho(Q\Lambda_{\Gamma}Q^{T}+\alpha X)=U_Q\rho(\Lambda_{\Gamma}+\alpha Q^TXQ)U^\dagger_Q=
        U_Q\rho(\Lambda_{\Gamma}+\alpha \tilde X)U^\dagger_Q.
        \end{align}
    Hence, the derivative, due to the chain rule, is equal to
        \begin{align}
        \partial_\alpha\rho(\Lambda_{\Gamma}+\alpha X)|_{\alpha=0}&=U_Q\left.\partial_\alpha\rho\left(\Lambda_{\Gamma}+\alpha \sum_{a\le b}\tilde X_{a,b}E_{a,b}\right)\right|_{\alpha=0}U^\dagger_Q\\
        \nonumber
        &=U_Q\sum_{a\leq b}\tilde X_{a,b} \partial_\alpha\rho(\Lambda_{\Gamma}+\alpha E_{a,b})|_{\alpha=0}U^\dagger_Q\\
        &=-\frac{i}{4}U_Q\sum_{a\leq b}\tilde X_{a,b} [\gamma_a,\{\gamma_b,\rho(\Lambda_{\Gamma})\}]U^\dagger_Q\\
         \nonumber
        &=-\frac{i}{4}\sum_{a\leq b}\tilde X_{a,b} [U_Q\gamma_aU^\dagger_Q,\{U_Q\gamma_bU^\dagger_Q,U_Q\rho(\Lambda_{\Gamma})U^\dagger_Q\}]\\
         \nonumber
        &=-\frac{i}{4}\sum_{a\leq b}(Q\tilde XQ^T)_{a,b} [\gamma_a,\{\gamma_b,\rho(\Gamma)\}]\\
         \nonumber
        &=-\frac{i}{8}\sum_{a, b}X_{a,b} [\gamma_a,\{\gamma_b,\rho(\Gamma)\}],
         \nonumber
    \end{align}
    which concludes the proof.
\end{proof}
\subsection{Trace-distance bounds for pure free-fermionic states}\label{App.subb:purestatesbounds}
We leverage the \emph{gentle measurement lemma} 
to establish the following result.
\begin{proposition}[Trace distance between a pure free-fermionic state and an arbitrary state]
\label{le:trpure}
    For a free-fermionic state 
    vector $\ket{\psi}$ and an arbitrary (possibly non-free-fermionic) state $\rho$, it holds that
    \begin{align}
        \|\rho-\ketbra{\psi}{\psi}\|_{1} \le \sqrt{\|\Gamma(\rho)-\Gamma(\ketbra{\psi}{\psi})\|_{1}}\,.
    \end{align}
\end{proposition}
\begin{proof}
    Since $\ket{\psi}$ is a free-fermionic state, it can be expressed as $\ket{\psi}=U_Q\ket{0^n}$ for a free-fermionic unitary associated with $Q\in \mathrm{O}(2n)$. Define $\rho^{\prime}:=U^{\dagger}_Q\rho U_Q$. For any $j\in[n]$, we have
    \begin{align}
        \Tr(\ketbra{0}{0}_j \rho^{\prime})=\frac{1}{2}+\frac{1}{2}\left[\Gamma(\rho^{\prime})\right]_{2j-1,2j}\eqqcolon 1-\varepsilon_j,
    \end{align}
    where we have used $\ketbra{0}{0}_j=(I+Z_j)/2=(I-i \gamma_{2j-1}\gamma_{2j})/2$ and $\varepsilon_j:=(1-\left[\Gamma(\rho^{\prime})\right]_{2j-1,2j})/2$.
    Now, we have
    \begin{align}
        \|\rho-\ketbra{\psi}{\psi}\|_{1}&=\|\rho^{\prime}-\ketbra{0^n}{0^n}\|_{1}\le2\sqrt{\sum^{n}_{j=1}\varepsilon_j}=2\sqrt{\sum^{n}_{j=1}\frac{1-\left[\Gamma(\rho^{\prime})\right]_{2j-1,2j}}{2}},
        \label{eq:trp1}
    \end{align}
    where we have used the unitary invariance of the one-norm and 
    the gentle measurement Lemma~\ref{le:gentle}.
    Let $\Lambda$ be the matrix $\Lambda:=\bigoplus_{j = 1}^{n} \begin{pmatrix} 0 &  1 \\ -1 & 0 \end{pmatrix}=\Gamma(\ketbra{0^n}{0^n})$. Observe that $\sum^n_{j=1}\left[\Gamma(\rho^{\prime})\right]_{2j-1,2j}=\frac{1}{2}\Tr(\Lambda^\dagger \Gamma(\rho^{\prime}))$.
    Using this and Eq.\ \eqref{eq:trp1}, we have
\begin{align}
    \|\rho - \ketbra{\psi}{\psi}\|_{1} & \le \sqrt{\Tr\left(I_{2n} - \Lambda^\dagger \Gamma(\rho^{\prime})\right)} \nonumber \\
    & \le \sqrt{\|I_{2n} - \Lambda^\dagger \Gamma(\rho^{\prime})\|_{1}} \nonumber \\
    & = \sqrt{\| \Lambda - \Gamma(U^{\dagger}_Q\rho U_Q)\|_{1}} \nonumber \\
    & = \sqrt{\| \Lambda - Q^{T}\Gamma(\rho)Q\|_{1}} \nonumber \\
    & = \sqrt{\| \Gamma(\ketbra{\psi}{\psi}) - \Gamma(\rho)\|_{1}},
\end{align}
where in the second step we have used H{\"o}lder inequality, in the third step the definition of one-norm, in the third step the unitary invariance of the one-norm and the definition of $\rho^{\prime}$, in the fourth step the fact that $\Gamma(U^{\dagger}_Q\rho U_{Q})= Q^{T}\Gamma(\rho)Q$, and in the last step the unitary invariance of the one-norm and the fact that $\Gamma(\ketbra{\psi}{\psi})=Q\Gamma(\ketbra{0^n}{0^n})Q^{T}$.
\end{proof}

As discussed in the main text, the inequality above can be useful for quantum state certification scopes, as we briefly detail now. 
Consider the scenario where one aims to prepare the free-fermionic state $\psi$, with known correlation matrix $\Gamma(\psi)$. However, in a noisy experimental setup, an unknown state $\rho$ is effectively prepared instead. By estimating the correlation matrix of $\rho$ with precision $\varepsilon$, achievable through efficient experimental methods as discussed in Section~\ref{Sec:propertytesting}, we obtain an approximation $\tilde{\Gamma}$ satisfying $\|\tilde{\Gamma}-\Gamma(\rho)\|_1 \le \varepsilon$.
Next, we evaluate the quantity $a \coloneqq \|\tilde{\Gamma}-\Gamma(\psi)\|_1$ efficiently classically, which can be related to the trace distance error incurred during the preparation of $\psi$ as 
\begin{equation}
    \|\psi - \rho\|_1 \le \sqrt{\|\Gamma(\psi) - \Gamma(\rho)\|_{1}} \le \sqrt{\|\Gamma(\psi) - \tilde{\Gamma}\|_{1} + \|\tilde{\Gamma}-\Gamma(\rho)\|_{1}} = \sqrt{a + \varepsilon } \le \sqrt{a} + \sqrt{\varepsilon}.
\end{equation}
Thus, this inequality provides a bound on the trace distance error between the prepared state $\rho$ and the target state $\psi$, incorporating both the classical evaluated error $a$ and the experimental accuracy $\varepsilon$.

\begin{theorem}[Trace distance upper bound between two pure free-fermionic states]\label{thm:gausspp}
Let $\psi_1, \psi_2$  pure free-fermionic states with correlation matrices $\Gamma(\psi_1), \Gamma(\psi_2)$.
Assuming that $\|\Gamma(\psi_1) - \Gamma(\psi_2)\|_{\infty} < 2$, it holds that
\begin{align}
\|\psi_1 - \psi_2\|_1 \leq \frac{1}{2} \|\Gamma(\psi_1) - \Gamma(\psi_2)\|_2
\end{align}
and
\begin{align}
\mathcal{F}(\psi_1, \psi_2)\geq 1-\frac{1}{16} \|\Gamma(\psi_1) - \Gamma(\psi_2)\|_2^2,
\end{align}
where $\mathcal{F}(\psi_1, \psi_2)=|\braket{\psi_1}{\psi_2}|^2$.
Otherwise, if the quantity $\|\Gamma(\psi_1) - \Gamma(\psi_2)\|_{\infty}$ (which is always $\leq 2$) is equal to $2$, then we simply have $\|\psi_1 - \psi_2\|_1 = 2$ and $\mathcal{F}(\psi_1, \psi_2)=0$.
\end{theorem}
\begin{proof} 
First of all, we have
    \begin{align}
    \label{eq:tracedistfid}
        \norm{\psi_1-\psi_2}_1=2 \sqrt{1-\left|\braket{\psi_1}{\psi_2}\right|^2}.
    \end{align}
This can be seen by considering $Q\coloneqq \psi_1-\psi_2$.
The Hermitian matrix $Q$ has a rank of at most $2$, which means it can have at most two non-zero eigenvalues denoted as $\lambda_1$ and $\lambda_2$. Since the trace of $Q$ is zero, we have $\lambda_2 = -\lambda_1$. Additionally, we know that $\Tr(Q^2) = \lambda_1^2 + \lambda_2^2 = 2\lambda_1^2$, and $\Tr(Q^2) = 2(1 - \left|\braket{u}{v}\right|^2)$. Therefore, we can conclude that $\lambda_1 = \sqrt{1 - \left|\braket{\psi_1}{\psi_2}\right|^2}$.
The $1$-norm of $Q$ is given by $\norm{Q}_1 = |\lambda_1|+|\lambda_2|$, which simplifies to 
\begin{equation}
\norm{\psi_1-\psi_2}_1=2\sqrt{1-|\braket{\psi_1}{\psi_2}|^2}.
\end{equation}
From Lemma~\ref{le:overlapBr}, 
it follows that
\begin{align}
        |\braket{\psi_1}{\psi_2}|^2=\sqrt{\mathrm{det}\!\left(\frac{\Gamma(\psi_1)+\Gamma(\psi_2)}{2}\right)},
    \end{align}
    where we have used that $(\Pf(A))^2= \det(A)$ (note that the determinant of an anti-symmetric matrix is positive). 
    Thus, we have
    \begin{align}
        |\braket{\psi_1}{\psi_2}|^2=\sqrt{\mathrm{det}(\Gamma(\psi_1)^{T})\mathrm{det}\!\left(\frac{\Gamma(\psi_1)+\Gamma(\psi_2)}{2}\right)}=\sqrt{\mathrm{det}\!\left(\frac{\mathbb{1}+\Gamma(\psi_1)^{T}\Gamma(\psi_2)}{2}\right)}=\sqrt{\mathrm{det}\!\left(\frac{\mathbb{1}+X}{2}\right)}
    \end{align}
where the first step follows because $\det(\Gamma(\psi_1)^{T}) = \det(\Gamma(\psi_1)) = 1$ (see Remark~\ref{rem:detORT}), the second step follows from $\det(A)\det(B) = \det(AB)$ and the fact that $\Gamma(\psi_1)$ is an orthogonal matrix (see Remark~\ref{rem:detORT}). In the last step, we have defined $X\coloneqq\Gamma(\psi_1)^{T}\Gamma(\psi_2)$, which is also an orthogonal matrix (being product of two orthogonal matrices).
Since the determinant of a matrix is equal to the product of its eigenvalues, we have
\begin{align}
\label{eq:firstdet}
        |\braket{\psi_1}{\psi_2}|^2=\sqrt{\mathrm{det}\!\left(\frac{\mathbb{1}+X}{2}\right)}=\sqrt{\prod^{2n}_{j=1}\left(\frac{\mathbb{1}+\lambda_j(X)}{2}\right)},
\end{align}
where we have denoted as $\{\lambda_j(X)\}^{2n}_{j=1}$ the eigenvalues of $X$. Since $X$ is orthogonal, we also have that for 
each $j\in[2n]$, there exists $\phi_j \in \mathbb{R}$ such that $\lambda_j(X)=e^{i\phi_j}$.

We are now going to show that all eigenvalues of $X$ have multiplicity two.
Let $\ket{v}$ be an eigenvector of $X$ corresponding to 
the eigenvalue $\lambda$,
so that
    \begin{align}
        X\ket{v}=\lambda \ket{v}.
        \label{eq:eig}
    \end{align}
We also have that
    \begin{align}
        X (\Gamma(\psi_2)\ket{v}^*)= \Gamma(\psi_1)^{T}\Gamma(\psi_2)^2\ket{v}^*=-\Gamma(\psi_1)^{T}\ket{v}^*=- \Gamma(\psi_1)^{T}\frac{1}{\lambda^*} (X\ket{v})^*= \lambda (\Gamma(\psi_2)\ket{v}^*),
    \end{align}
    where the third step follows from the fact that $\Gamma(\psi_2)$ is anti-symmetric and orthogonal and the fifth step from the fact that $\lambda^{-1}=\lambda^*$. Thus, $\Gamma(\psi_2)\ket{v}^{*}$ and $\ket{v}$ are both eigenstates of $X$ with eigenvalue $\lambda$.
Note that they must be different vectors. In fact, if it holds that $\Gamma(\psi_2)\ket{v}^{*}=e^{i\theta}\ket{v}$ for $\theta\in [0,2\pi]$, then multiplying by $\Gamma(\psi_2)$ we also have
\begin{align}\label{eq:contradic}
    -\ket{v}^{*}=e^{i\theta}\Gamma(\psi_2)\ket{v}=e^{i\theta}(\Gamma(\psi_2)\ket{v}^*)^*=e^{i\theta}(e^{i\theta}\ket{v})^*=\ket{v}^* .
\end{align}
Thus, we have $\ket{v}^*=-\ket{v}^*$, which cannot be true. Thus we have shown that all eigenvalues of $X$ have multiplicity two. Thus, without loss of generality, we can assume that $\lambda_{n+k}(X)=\lambda_{k}(X)$ for $k\in[n]$, meaning that we can only look at the first half of the eigenvalues.
Hence, from Eq.~(\ref{eq:firstdet}), we have
\begin{align}
        |\braket{\psi_1}{\psi_2}|^2=\prod^{n}_{j=1}\left(\frac{\mathbb{1}+\lambda_j(X)}{2}\right).
\end{align}
If an eigenvalue $\lambda_j(X)=e^{i\phi_j}$ is not real (which implies that its associated eigenstate is not a real vector), then also $\lambda_j^*(X)=e^{-i\phi_j}$ will be a distinct eigenvalue of $X$ (as follows by simply taking the complex conjugate of Eq.~\eqref{eq:eig}). The remaining eigenvalues will be $+1$ or $-1$.
Let us denote with $2 n_c$ the number of not-real eigenvalues (they are even, as they come in pairs), with $n_+$ the number of $+1$ eigenvalues, and with $n_{-}$ the number of $-1$ eigenvalues. So that we have
\begin{align}
\label{eq:numbers}
    n_+ + n_- + 2n_c = n.
\end{align}
Thus, the $X$'s eigenvalues will be 
\begin{align}
    \{\lambda_1(X),\dots,\lambda_n(X)\} = \{\underbrace{e^{i\phi_1},e^{-i\phi_1},\dots, e^{i\phi_{n_c}},e^{-i\phi_{n_c}}}_{2n_c},\underbrace{+1,\dots, +1}_{n_+},\underbrace{-1,\dots,-1}_{n_-}\},
\end{align}
where $\{e^{i\phi_j}\}^{n_c}_{j=1}$ are not real. 
Hence, we find
\begin{align}
\label{eq:fidlong}
        |\braket{\psi_1}{\psi_2}|^2&=\delta_{n_{-},0}\prod^{n_c}_{j=1}\frac{1}{4}\left(1+e^{i\phi_j}\right)\left(1+e^{-i\phi_j}\right)\\
        \nonumber
        &=\delta_{n_{-},0}\prod^{n_c}_{j=1}\frac{1}{2}\left(1+\cos(\phi_j)\right)\\\nonumber
        &=\delta_{n_{-},0}\prod^{n_c}_{j=1}\left(1-\sin^2\!\left(\phi_j/2\right)\right)\\\nonumber
        &\ge \delta_{n_{-},0}\left(1-\sum^{n_c}_{j=1}\sin^2\!\left(\phi_j/2\right)\right),
\end{align}
where in the last step we have used 
Weierstrass' product inequality.
Next, we can rewrite
    \begin{align}
    \label{eq:Froblong}
        \|\Gamma(\psi_1)-\Gamma(\psi_2)\|^2_2&=\|1-\Gamma(\psi_1)^{T}\Gamma(\psi_2)\|_2^2\\\nonumber
        &=\|1-X\|_2^2  \\\nonumber
        &=\sum_{j=1}^{2n} |1-\lambda_j(X)|^2\\\nonumber
        &=2\sum_{j=1}^{n} |1-\lambda_j(X)|^2\\\nonumber
        &=2\left(\sum_{j=1}^{n_c} \left(|1-e^{i \phi_j}|^2+|1-e^{-i \phi_j}|^2\right) + \sum_{j=1}^{n_+}0 + \sum_{j=1}^{n_-}4 \right)\\\nonumber
        &=16 \sum_{j=1}^{n_c} \sin^2(\phi_j/2) + 8 n_{-},
    \end{align}
    where in the last step we have used that $|1-e^{i \phi_j}|=2|\sin(\phi_j/2)|$.
If $n_{-}>0$, Eq.\ \eqref{eq:fidlong} implies that $|\braket{\psi_1}{\psi_2}|^2=0$. Thus, using Eq.\ \eqref{eq:tracedistfid}, we have
\begin{align}
    \|\psi_1-\psi_2\|_1=2, \quad \text{if $n_{-}>0$}.
\end{align}
If $n_{-}=0$, Eqs.~\eqref{eq:fidlong},\eqref{eq:Froblong} imply that
    \begin{align}
        \mathcal{F}(\psi_1, \psi_2)=|\braket{\psi_1}{\psi_2}|^2&\ge 1-\frac{1}{16}\|\Gamma(\psi_1)-\Gamma(\psi_2)\|^2_2. 
\end{align}
Thus, by Eq.\ \eqref{eq:tracedistfid}, we arrive at
    \begin{align}
        \|\psi_1-\psi_2\|_1&\leq \frac{1}{2}\|\Gamma(\psi_1)-\Gamma(\psi_2)\|_2, \quad \text{if $n_{-}=0$}.
    \end{align}
 We are now only left to show that the condition $n_{-}=0$ is satisfied if and only if it holds that $\norm{\Gamma(\psi_1)-\Gamma(\psi_2)}_{\infty}<2$.
 This is indeed the case, since the quantity (which is always $\le 2$)
\begin{align}
    \norm{\Gamma(\psi_1)-\Gamma(\psi_2)}_{\infty}=\norm{1-X}_{\infty}
\end{align}
can be equal to $2$ if and only if 
$X$ has at least a $-1$ eigenvalues if 
and only if $n_{-}>0$.
\end{proof}

\begin{remark}[Saturation of the inequality]
We note that the above upper bound is saturated for all pure free-fermionic states $\psi_1,\psi_2$ with number of modes/qubits $n\le 3$. Thus, we get
\begin{align}
\|\psi_1 - \psi_2\|_1 =
\begin{cases}
    \frac{1}{2} \|\Gamma(\psi_1) - \Gamma(\psi_2)\|_2, & \text{if } \|\Gamma(\psi_1) - \Gamma(\psi_2)\|_{\infty} < 2, \\
    2, & \text{if } \|\Gamma(\psi_1) - \Gamma(\psi_2)\|_{\infty} = 2.
\end{cases}
\end{align}
\end{remark}
\begin{proof}
    The case $\norm{\Gamma(\psi_1)-\Gamma(\psi_2)}_{\infty}=2$ is analogous to the previous proof, so let us focus on $\norm{\Gamma(\psi_1)-\Gamma(\psi_2)}_{\infty}<2$.
    Using the same notation of the previous proof, we have that $n_{-}=0$, i.e., the number of $-1$ eigenvalues of $X\coloneqq \Gamma(\psi_1)^{T}\Gamma(\psi_2)$ is zero. Hence, because of Eq.~\eqref{eq:fidlong}, we have
    \begin{align}
    \label{eq:redfid}
        |\braket{\psi_1}{\psi_2}|^2= \prod^{n_c}_{j=1}\left(1-\sin^2\!\left(\phi_j\right)\right).
    \end{align}
    From Eq.\ \eqref{eq:numbers} and using that $n\le 3$, we get
    \begin{align}
        2n_c+ n_+\le 3,
    \end{align}
    which implies that $n_c=0$ or $n_c=1$.
    If $n_c=0$, we have $|\braket{\psi_1}{\psi_2}|^2=1$, which implies $\norm{\psi_1-\psi_2}_1=0$. Furthermore, because of Eq.~\eqref{eq:Froblong}, we also have $\frac{1}{2} \|\Gamma(\psi_1) - \Gamma(\psi_2)\|_2=0$.
    
    If $n_c=1$, we then have
    \begin{align}
          |\braket{\psi_1}{\psi_2}|^2= 1-\sin^2\!\left(\phi_1\right)=1-\frac{1}{16}\|\Gamma(\psi_1)-\Gamma(\psi_2)\|^2_2,
    \end{align}
    where in the first step we have used Eq.~\eqref{eq:redfid}, and in the second step we have used Eq.~\eqref{eq:Froblong}. Thus, we reach the conclusion by using that $\norm{\psi_1-\psi_2}_1=2 \sqrt{1-\left|\braket{\psi_1}{\psi_2}\right|^2}$.
\end{proof}

In the next theorem, we also provide a lower bound on the trace distance between two pure Gaussian state in terms of the two norm difference of their covariance matrices.

\begin{theorem}[Lower bound on trace distance in terms of correlation matrices] Let $\psi_1,\psi_2$ be two pure non-orthogonal Gaussian states. Their trace distance is lower bounded by 
\begin{align}
\|\psi_1-\psi_2\|_{1}\ge 2 \sqrt{1 - \exp\left(-\|\Gamma_1 - \Gamma_2\|_2^2 / 16\right)}\,.
\end{align}
In particular, if $\| \Gamma_1 - \Gamma_2\|_2 \leq 4$, it holds that $\|\psi_1-\psi_2\|_{1}\ge\frac{3}{8} \|\Gamma_1 - \Gamma_2\|_2$.
\begin{proof}
    The proof follows the same step of the proof of Theorem~\ref{thm:gausspp}, as such we use the same notation that we defined above. On the one hand, from Eq.~\eqref{eq:fidlong}, we have
    \begin{align}
\|\psi_1-\psi_2\|_1=2\sqrt{1-|\langle\psi_1|\psi_2\rangle|^2}=2\sqrt{1-\delta_{n_-,0}\left(1-\sum_{j=1}^{n_c}\sin^2(\phi_j/2)\right)}\label{dajsidjaisjd}
    \end{align}
On the other hand, from Eq.~\eqref{eq:Froblong}, we have
\begin{align}
\|\Gamma(\psi_1)-\Gamma(\psi_2)\|_{2}^{2}=16\sum_{j=1}^{n_c}\sin^2(\phi_j/2)
\end{align}
As we noticed in the proof of Theorem~\ref{thm:gausspp}, whenever $n_-\neq 0$, then $\|\psi_1-\psi_2\|_1=2$ (i.e., the states are maximally far away). Setting $n_-=0$, we therefore lower bound the trace distance with the norm $2$ of the covariance matrices. Since $(1-x)\le e^{-x}$, we can bound Eq.~\eqref{dajsidjaisjd} from below as
\begin{align}
\|\psi_1-\psi_2\|_1\ge 2\sqrt{1-\exp\left({-\sum_{j=1}^{n_c}\sin^2(\phi_j/2)}\right)}=2\sqrt{1-\exp\left({-\frac{1}{16}\|\Gamma(\psi_1)-\Gamma(\psi_2)\|_{2}^{2}}\right)}
\end{align}
getting the desired result.
\end{proof}
\end{theorem}

\subsection{Trace-distance bounds for mixed free-fermionic states}\label{App.subb:mixedstatesbounds}

\begin{proposition}[Lower bound on trace distance in terms of correlation matrices]\label{le:tracedistancelowerboundcormatrix}
    Given two quantum states $\rho$ and $\sigma$, their trace distance is lower bounded by the operator norm difference 
    \begin{align}
        \norm{\rho - \sigma}_1 \ge \norm{\Gamma(\rho)-\Gamma(\sigma)}_{\infty}
    \end{align}
    of their correlation matrices $\Gamma(\rho)$ and $\Gamma(\sigma)$.
\end{proposition}
\begin{proof}
    Due to H\"older's inequality, we have
    \begin{align}
        \norm{\rho - \sigma}_1 &\ge  \sup_{\norm{W}_{\infty}=1}|\Tr(W(\rho - \sigma))|.
    \end{align}
    Now, let us restrict the operators $W$ to the form $W=U^{\dagger}_{Q}i\gamma_{j}\gamma_{k}U_{Q}$, where $ j,k\in[2n]$ and $U_{Q}$ is a free-fermionic unitary associated with an orthogonal matrix $Q \in \mathrm{O}(2n)$. Indeed, note that $\|U^{\dagger}_{Q}i\gamma_{j}\gamma_{k}U_{Q}\|_{\infty}=\|i\gamma_{j}\gamma_{k}\|_{\infty}=1$. It holds that
    \begin{align}
        \norm{\rho - \sigma}_1 &\ge  \sup_{\norm{W}_{\infty}=1}|\Tr(W(\rho - \sigma))|= \sup_{\substack{j,k\in[2n]\\ Q\in \mathrm{O}(2n)}}|\Tr(U^{\dagger}_Qi\gamma_{j}\gamma_{k}U_Q(\rho - \sigma))| .
    \end{align}
    We then have
    \begin{align}
        \sup_{\substack{j,k\in[2n]\\ Q\in \mathrm{O}(2n)}}|\Tr(U^{\dagger}_Qi\gamma_{j}\gamma_{k}U_Q(\rho - \sigma))| &= \sup_{\substack{j,k\in[2n]\\ Q\in \mathrm{O}(2n)}}|\Tr(i\gamma_{j}\gamma_{k}U_Q(\rho - \sigma)U^{\dagger}_Q)| \\
        &= \sup_{\substack{j,k\in[2n]\\ Q\in \mathrm{O}(2n)}} |(Q(\Gamma(\rho)-\Gamma(\sigma))Q^{T})_{j,k}|,
        \nonumber
    \end{align}
    where in the last step we have used that 
    \begin{align}
    \Gamma(U_Q\rho U^{\dagger}_Q)=Q\Gamma(\rho) Q^T .
    \end{align}
    Since $\Gamma(\rho)-\Gamma(\sigma)$ is real and anti-symmetric, it can be brought into a normal form $\Gamma(\rho)-\Gamma(\sigma)=O'\Lambda' O^{\prime T}$, where $\Lambda'=\bigoplus^n_{i=1} \begin{pmatrix}
        0 & \lambda'_i\\ - \lambda'_i & 0 
    \end{pmatrix}$ and $\{\pm i \lambda'_i\}^{n}_{i=1}$ are the purely imaginary eigenvalues of $\Gamma(\rho)-\Gamma(\sigma)$. By choosing $Q=O^{\prime T}$, we have
    \begin{align}
        \sup_{\substack{j,k\in[2n]\\ Q\in \mathrm{O}(2n)}}|(Q^{\prime}(\Gamma(\rho)-\Gamma(\sigma))Q^{\prime T})_{j,k}| \ge \sup_{j,k\in[2n]}|\Lambda'_{j,k}| = \sup_{i\in[n]}|\lambda'_{i}| = \norm{\Gamma(\rho)-\Gamma(\sigma)}_{\infty}\label{s90}
    \end{align}
    that concludes the proof.
\end{proof}

Next, we present an upper bound on the trace distance between two (possibly mixed) Gaussian states in terms of the norm difference of their respective correlation matrices.
\begin{theorem}[Trace distance upper bound between two mixed free-fermionic states]\label{th:mixedtracedistance}
Let $\rho,\sigma$ be two (possibly mixed) free-fermionic states. Then it holds that
\begin{align}
    \|\rho-\sigma\|_1\leq \frac{1}{2}\|\Gamma(\rho)-\Gamma(\sigma)\|_1\,
\end{align}
as well as (from Fuchs-van de Graaf inequality~\cite{fuchs_1999_inequality})
\begin{align}
    F(\rho,\sigma)\geq (1- \frac{1}{4}\|\Gamma(\rho)-\Gamma(\sigma)\|_1)^2\geq
     1- \frac{1}{2}\|\Gamma(\rho)-\Gamma(\sigma)\|_1,
\end{align}
where $\mathcal{F}(\rho,\sigma)\coloneqq \Tr(\sqrt{\sqrt{\sigma}\rho\sqrt{\sigma}})^{2}$ is the fidelity between $\rho$ and $\sigma$. 
\end{theorem}
\begin{proof}
    Let us start with the result of Lemma~\ref{lem:derivativefreefermionic}, i.e,
    \begin{align}
            \nonumber
       \partial_{\alpha} \rho(\Gamma+\alpha X)%
        &=\frac{-i}{8}\sum_{a,b} X_{a,b} [\gamma_a,\{\gamma_b,\rho\}]\\
        \nonumber
        &=\frac{-i}{8}\sum_{a,b} X_{a,b} (\gamma_a \gamma_b \rho -\rho \gamma_b\gamma_a-\gamma_b\rho\gamma_a+\gamma_a\rho\gamma_b)\\
        \nonumber
        &=\frac{-i}{4}\sum_{a<b} X_{a,b} (\gamma_a \gamma_b \rho -\rho \gamma_b\gamma_a-\gamma_b\rho\gamma_a+\gamma_a\rho\gamma_b)\label{eq:s127}
    \end{align}
    Let us first consider the $n$-mode vacuum state $\rho\equiv\ketbra{0}$ (in the canonical Jordan-Wigner basis).
    We define
    \begin{align}
        \ket{i}&\coloneqq\gamma_{2i-1}\ket{0}=i\gamma_{2i}\ket{0},\\
        \ket{i,j}&\coloneqq\gamma_{2i-1}\gamma_{2j-1}\ket{0}=i\gamma_{2i-1}\gamma_{2j}\ket{0}=i\gamma_{2i}\gamma_{2j-1}\ket{0}=-\gamma_{2i}\gamma_{2j}\ket{0},\quad i\neq j,\\
        \ket{0}&=-i\gamma_{2i-1}\gamma_{2i}\ket{0}.
    \end{align}
    Using these definition, we can rewrite Eq.~\eqref{eq:s127} as 
    \begin{align}
        \partial_{\alpha} \rho(\Lambda+\alpha X)|_{\alpha=0}
        &=\frac{-i}{8}\sum_{i\neq j\in[n];a,b\in \{0,1\}} X_{2i-1+a,2j-1+b} (i^{a+b}\ketbra{i,j}{0} -(-i)^{a+b}\ketbra{0}{i,j}-i^{b-a}\ketbra{j}{i}+i^{a-b}\ketbra{i}{j})
        \nonumber
        \\
        &+\frac{1}{2}\sum_{i\in[n]} X_{2i-1,2i} (\ketbra{0}{0}-\ketbra{i}{i}).\label{eq:133}
    \end{align}
    Therefore, we can define the two matrices
    \begin{align}
    K_{i,j}
    &=\frac{-iX_{2i-1,2j-1}-iX_{2i,2j}-X_{2i-1,2j}+X_{2i,2j-1}}{4},\\
    L_{i,j}&=\frac{-iX_{2i-1,2j-1}+iX_{2i,2j}+X_{2i-1,2j}+X_{2i,2j-1}}{4}.
    \end{align}
    Notice that, in general, it holds that 
    \begin{align}
        X&= 2\begin{pmatrix}
            -\mathrm{Im}(K)-\mathrm{Im}(L) &-\mathrm{Re}(K)+\mathrm{Re}(L)\\
            \mathrm{Re}(K)+\mathrm{Re}(L) & -\mathrm{Im}(K)+\mathrm{Im}(L) 
        \end{pmatrix},\\
    \end{align}
    where the first block describes the odd indices, and the second the even indices.
    This simplifies Eq.~\eqref{eq:133} to
    \begin{align}
    \partial_{\alpha} \rho(\Lambda+\alpha X)|_{\alpha=0}&=-\ketbra{0}{0}\tr(K) +\sum_{i,j} \ketbra{i}{j} K_{i,j}+  \sum_{i<j} \ketbra{0}{i,j} L^{*}_{i,j} + \ketbra{i,j}{0} L_{i,j}\\
    &=\begin{pmatrix}
        -\Tr(K)& 0 & \vec L^\dagger\\
        0& K &0\\
        \vec L& 0 &0
    \end{pmatrix}_{\{\ket{0},\ket{i},\ket{i,j}\}}
    \nonumber
    \end{align}
    where $\vec L$ is a vectorization of the top triangular component of $L$.
    Due to the block diagonal structure of the derivative, we have
    \begin{align}
    \|\partial_{\alpha} \rho(\Lambda+\alpha X)|_{\alpha=0}\|_1=\|K\|_1+\left\|\begin{pmatrix}
        -\Tr(K) &\vec L^\dagger\\
        \vec L &0\\
    \end{pmatrix}\right\|_1 =\|K\|_1+\left\|\begin{pmatrix}
        -\Tr(K) &\frac{\|L\|_2}{\sqrt{2}}\\
        \frac{\|L\|_2}{\sqrt{2}} &0\\
    \end{pmatrix}\right\|_1
    =\|K\|_1+\sqrt{\Tr(K)^2+2\|L\|_2^2}\,,\label{eq:s139}
    \end{align}
    where the last step follows from explicitly computing the eigenvalues of the $2\times2$ matrix.
    Next we bound the terms by the trace norm of $X$. For the first term we get that
    \begin{align}
    \frac{\|X\|_1}{4}\geq \frac{\|X-\Lambda X\Lambda \|_1}{8}=\frac{1}{2}\left\|\begin{pmatrix}
        -\mathrm{Im}(K)& -\mathrm{Re}(K)\\
        \mathrm{Re}(K) &-\mathrm{Im}(K)
    \end{pmatrix}\right\|_1=\frac{1}{4}\left\|
    \begin{pmatrix}
        i& -i\\
        1 &1
    \end{pmatrix}
    \begin{pmatrix}
        iK^\dagger& 0\\
        0 &-iK
    \end{pmatrix}    
    \begin{pmatrix}
        -i& 1\\
        i &1
    \end{pmatrix}\right\|_1=\|K\|_1\,.\label{eq:140}
    \end{align}
    For the second expression, we first observe that
    \begin{align}
        \Tr(K)&=\frac{1}{4}\Tr(\Lambda X),\\
        \|L\|_2^2&=\frac{\|X\|_2^2 -\tr(X\Lambda X\Lambda)}{16}\,.
    \end{align}
    As such, we can rewrite
    \begin{align}
        \Tr(K)^2+2\|L\|_2^2
        &=\frac{\Tr(\Lambda X)^2+2\|X\|_2^2 -2\tr(X\Lambda X\Lambda)}{16}\leq \frac{\|X\|_1^2}{16}\,.\label{eq:143}
    \end{align}
    using lemma~\ref{lem:cov_ineq} in the last step.
    Putting together Eq.~\eqref{eq:140} and Eq.~\eqref{eq:143}, we can bound the derivative in Eq.~\eqref{eq:s139} as
    \begin{align}
        \left\|\partial_{\alpha} \rho(\Lambda+\alpha X)|_{\alpha=0}\right\|_1\leq \frac{\|X\|_1}{4}+\frac{\|X\|_1}{4}=\frac{\|X\|_1}{2}
    \end{align}
    for $\rho\equiv\ketbra{0}{0}$. Instead, for an arbitrary pure Gaussian state, using the unitary invariance of the trace norm, it follows that
    \begin{align}
        \left\|\partial_{\alpha} \rho(\Gamma+\alpha X)|_{\alpha=0}\right\|_1=\left\|\partial_{\alpha} \rho(Q\Lambda Q^T+\alpha X)|_{\alpha=0}\right\|_1=\left\|\partial_{\alpha} \rho(\Lambda +\alpha Q^TXQ)|_{\alpha=0}\right\|_1\leq \frac{\|Q^TXQ\|_1}{2}=\frac{\|X\|_1}{2}.
    \end{align}
    For mixed Gaussian states, we exploit the fact that every Gaussian state can be expressed as a convex mixture of pure Gaussian states, and
    hence we get the bound 
    \begin{align}
        \nonumber\|\partial_\alpha\rho(\Gamma+\alpha X)|_{\alpha=0}\|_1&= \|\sum_{a\in\{0,1\}^n} p_a\partial_\alpha\rho(Q_a\Lambda Q_a^T+\alpha X)|_{\alpha=0}\|_1\\
        &\leq \max_{a\in\{0,1\}^n}\|\partial_\alpha\rho(Q_a\Lambda Q_a^T+\alpha X)|_{\alpha=0}\|_1\leq  \frac{\|X\|_1}{2}\leq \frac{\|\partial_\alpha \Gamma(\alpha)|_{\alpha=0}\|_1}{2}\label{eq:lipschitz}
    \end{align}
    where we used the linearity of the derivative for the first term. To finally prove the result, we consider a transition $\Gamma(\alpha)=\Gamma_\rho+\alpha(\Gamma_\sigma-\Gamma_\rho)$ with $\alpha\in[0,1]$, where $\rho(\alpha)$ is the Gaussian state with $\Gamma(\alpha)$ as correlation matrix. As such, we get
\begin{align}
    \|\sigma-\rho\|_1&\leq \lim_{m\rightarrow\infty }\sum_{i=1}^m\left\|\rho\left(\frac{i}{m}\right)-\rho\left(\frac{i-1}{m}\right)\right\|_1\leq  \lim_{m\rightarrow\infty }\frac{1}{m}\sum_{i=1}^m\left\|m\left(\rho\left(\frac{i}{m}\right)-\rho\left(\frac{i-1}{m}\right)\right)\right\|_1\\
    \nonumber
    &\leq  \lim_{m\rightarrow\infty }\frac{1}{m}\sum_{i=1}^m\left\|\partial_{\alpha}\rho\left(\alpha\right)|_{\alpha=\frac{i}{m}}+O(m^{-1})\right\|_1\\
    \nonumber
    &\leq  \lim_{m\rightarrow\infty }\left(\frac{1}{m}\sum_{i=1}^m\left\|\partial_{\alpha}\rho\left(\alpha\right)|_{\alpha=\frac{i}{m}}\right\|_1+O(m^{-1})\right)\\
    &\leq\lim_{m\rightarrow\infty }\frac{1}{m}\sum_{i=1}^m\frac{\|\partial_{\alpha}\Gamma\left(\alpha\right)|_{\alpha=\frac{i}{m}}\|_1}{2}\nonumber \\
    &=\lim_{m\rightarrow\infty }\frac{1}{m}\sum_{i=1}^m\frac{\|\Gamma_\sigma-\Gamma_\rho\|_1}{2}=\frac{\|\Gamma_\sigma-\Gamma_\rho\|_1}{2}
    \nonumber,
\end{align}
using the telescope sum in the first step, then Taylor expanding in $1/m$ at $\alpha=\frac{i}{m}$, followed by use of Eq.~\eqref{eq:lipschitz}. This concludes the proof.
\end{proof}
\begin{theorem}[Fidelity lower bound between two mixed free-fermionic states]\label{th:mixedtrace}
Let $\rho,\sigma$ be two (possibly mixed) free-fermionic states. Then it holds that
\begin{align}
    \mathcal{F}(\sigma,\rho)\ge 1-\frac{1}{4}\|\Gamma(\rho)-\Gamma(\sigma)\|_1-\frac{1}{8}\|\Gamma(\rho)-\Gamma(\sigma)\|_2^2,
\end{align}
where $\mathcal{F}(\rho,\sigma)\coloneqq \Tr(\sqrt{\sqrt{\sigma}\rho\sqrt{\sigma}})^{2}$ is the fidelity between $\rho$ and $\sigma$.
\end{theorem}
\begin{proof}
By using Lemma~\ref{le:purification}, we can purify $\rho$ and $\sigma$ to two pure free-fermionic states  $\psi_\rho$ and $\psi_\sigma$ 
with correlation matrix
    \begin{align}
         \Gamma(\psi_\rho)=\begin{pmatrix}
            \Gamma(\rho)& \sqrt{I+\Gamma(\rho)^2}\\
            -\sqrt{I+\Gamma(\rho)^2}& -\Gamma(\rho)
        \end{pmatrix}.
    \end{align}
    This purification has an additional symmetry. Namely when we apply the orthogonal transformations
    \begin{align}
    \frac{1}{2} \begin{pmatrix}1&1\\1&-1\end{pmatrix} \Gamma(\psi_\rho) \begin{pmatrix}1&1\\1&-1\end{pmatrix}&=\frac{1}{2}\begin{pmatrix}
           0&\Gamma(\rho)-\sqrt{I+\Gamma(\rho)^2}\\
           \Gamma(\rho)+\sqrt{I+\Gamma(\rho)^2}&0\end{pmatrix}\eqqcolon\begin{pmatrix}
              0&-O_\rho^T\\ O_\rho &0
          \end{pmatrix},
          \\
    \frac{1}{2} \begin{pmatrix}1&1\\1&-1\end{pmatrix} \Gamma(\psi_\rho)\Gamma(\psi_\sigma)
    \begin{pmatrix}1&1\\1&-1\end{pmatrix}&=\begin{pmatrix}
              O_\rho O_\sigma^T  &0\\ 0& O_\rho^TO_\sigma
          \end{pmatrix},
    \end{align}
    we get two copies of the same operator.
 As such, the same structure as before in Theorem~\ref{thm:gausspp} holds, every eigenvalue of $\Gamma(\psi_\rho)\Gamma(\psi_\sigma)$ has an additional multiplicity, meaning that $n_-=1$ cannot occur.
    This means 

    \begin{align}
    \label{eq:chainfrob}
        \|\Gamma(\psi_\rho)-\Gamma(\psi_\sigma)\|_2^2&= 2\|\Gamma(\rho)-\Gamma(\sigma)\|_2^2+2\|\sqrt{I+\Gamma(\rho)^2}-\sqrt{I+\Gamma(\sigma)^2}\|_2^2\\
        \nonumber
        &\leq 2\|\Gamma(\rho)-\Gamma(\sigma)\|_2^2+2\|\Gamma(\rho)^2-\Gamma(\sigma)^2\|_1\\
        \nonumber
        &\leq 2\|\Gamma(\rho)-\Gamma(\sigma)\|_2^2+2\|\Gamma(\rho)(\Gamma(\rho)-\Gamma(\sigma))\|_1+2\|(\Gamma(\rho)-\Gamma(\sigma))\Gamma(\sigma)\|_1\\
        \nonumber
        &\leq 2\|\Gamma(\rho)-\Gamma(\sigma)\|_2^2+2 \|\Gamma(\rho)-\Gamma(\sigma)\|_1(\|\Gamma(\sigma)\|_{\infty}+\|\Gamma(\rho)\|_{\infty})
        \\
        &\leq 2\|\Gamma(\rho)-\Gamma(\sigma)\|_2^2+4\|\Gamma(\rho)-\Gamma(\sigma)\|_1,
        \nonumber
    \end{align}
    where in the first step we have used that the square of the $2$-norm of the entire matrix is given by the sum of the square of the $2$-norm of each sub-block, both being the sum of the square of the individual entries.  
    In the second step we have used the inequality (Ref.~\cite{bhatia1996matrix}, Eq.\ (X.23)) 
    \begin{align}
        \|A^{1/t}-B^{1/t}\|_{p} \leq \|A-B\|_{p/t}^{1/t}
        , \label{eq:bathiainequality}
    \end{align}
    which is valid for each $t \in [1,\infty), p \in [1,\infty] $ and positive matrix $A$ and $B$.
    In particular, for our case we use $p=t=2$. In the third step we have used triangle inequality, and in the fourth step the H\"older inequality for the one-norm 
    and the fact that the infinity norm of any correlation matrix is upper bounded by $1$. 
    From the fact that the fidelity of two mixed state is the maximum over all the possible purifications, we conclude that 
\begin{align}
    \mathcal{F}(\sigma,\rho)& \geq\left|\braket{\psi_\rho}{\psi_\sigma}\right|^2\\
    &\geq 1-\frac{1}{16}\|\Gamma(\psi_\rho) - \Gamma(\psi_\sigma)\|^2_2 
    \nonumber 
    \\
    &\ge 1-\frac{1}{4}\|\Gamma(\rho)-\Gamma(\sigma)\|_1-\frac{1}{8}\|\Gamma(\rho)-\Gamma(\sigma)\|_2^2,
    \nonumber
\end{align}
where we have used Theorem~\ref{thm:gausspp}, and in the last step we have used Eq.\ \eqref{eq:chainfrob}.
\end{proof}
Note that Proposition~\ref{le:trpure}, Theorem~\ref{thm:gausspp} and \ref{th:mixedtrace} directly translate to the particle-number preserving case by using the inequality in Eq.\ \eqref{eq:partnumbUPP}.
 
\subsection{Quantifying non-Gaussianity: Distance from the set of free-fermionic states}\label{sec:non-Gaussianity}
In this section, our objective is to establish lower bounds on the minimum trace distance between a state $\rho$ and the set of free-fermionic states. This distance serves as a metric for the inherent ``magic'' in a given state, providing a precise quantification of its ``non-Gaussianity''.
The presented lower bounds on the proposed measure of non-Gaussianity are carefully designed to enable time and sample-efficient estimation up to a constant precision in an experimental setting. This capability is crucial for quantifying the extent to which a state exhibits non-Gaussianity behavior in experimental scenarios.

We will derive bounds in different scenarios: one without prior assumptions on the state $\rho$ and another assuming that $\rho$ is pure or, more generally, has a rank at most $R$.
Moreover, we consider two distinct sets of free-fermionic states: $\mathcal{G}_{\mathrm{mixed}}$ and $\mathcal{G}_R$. Here, $\mathcal{G}_{\mathrm{mixed}}$ represents the set of all free-fermionic states, while $\mathcal{G}_R$ represents the set of free-fermionic states constrained to those with a rank at most $R$.
We begin with a simple result, addressing the case where $\rho$ is an arbitrary state and we consider its distance from the set $\mathcal{G}_R$.

\begin{theorem}[Lower bound on the distance of an arbitrary state from the set of free-fermionic states with rank at most $R$]\label{th:lowerboundnon-Gaussianity0}
Let $\rho$ be an arbitrary $n$-qubits/modes quantum state. 
Denote with $\mathcal{G}_{R}$ the set of free-fermionic states with rank at most $R=2^r$, with $r\in \{0,\dots,n-1\}$.
Let $\lambda_{r+1}(\Gamma(\rho))$ denote the $r+1$-th smallest normal eigenvalue of the correlation matrix of $\rho$ (with the normal eigenvalues chosen to be all non-negative, i.e. in $[0,1]$). We then have
\begin{align}
\min_{\sigma \in \mathcal{G}_{R}} \norm{\rho - \sigma}_1 &\ge 1-\lambda_{r+1}(\Gamma(\rho)).
\label{non-Gaussianitymeasure}
\end{align}
\end{theorem}

\begin{proof}
By virtue of Proposition~\ref{le:tracedistancelowerboundcormatrix}, we 
further establish
\begin{align}
\norm{\rho-\sigma}_1\ge \norm{\Gamma(\rho)-\Gamma(\sigma)}_{\infty}.
\end{align}
Furthermore, according to Lemma~\ref{le:lbanti-symm}, the infinity norm difference serves as a lower bound for the difference in eigenvalues
\begin{align}
\norm{\Gamma(\rho)-\Gamma(\sigma)}_{\infty} \ge |\lambda_j(\Gamma(\rho)) - \lambda_j(\Gamma(\sigma))|,
\end{align}
where $\{\lambda_j(\Gamma(\rho))\}^n_{j=1}$ and $\{\lambda_j(\Gamma(\sigma))\}^n_{j=1}$ are the normal, regular, eigenvalues of the two correlation matrices ordered in increasing order and are non-negative.
We then have
\begin{align}
\norm{\Gamma(\rho)-\Gamma(\sigma)}_{\infty} \ge |\lambda_{r+1} (\Gamma(\rho)) - 1|=1-\lambda_{r+1}(\Gamma(\rho)),
\end{align}
where we have used that because of Lemma~\ref{le:rankEigs} $\lambda_j(\Gamma(\sigma))=1$ for all $j\in \{r+1,\dots, n\}$ and for all $\sigma\in\mathcal{G}_{R}$. The result thus follows.
\end{proof}
When restricting to the set of pure free-fermionic states $\mathcal{G}_{\mathrm{pure}} \equiv \mathcal{G}_{R=1}$.

\begin{proposition}[Trace distance upper and lower bound from the set of free-fermionic states]
Let $\rho$ be an arbitrary quantum state. 
Denote with $\mathcal{G}_{\mathrm{pure}}$ the set of free-fermionic pure states. Let $0\le \lambda_1\le \dots\le \lambda_n$ be the normal eigenvalues of the correlation matrix $\Gamma(\rho)$ chosen to be all non-negative. We have
\begin{align}
\sqrt{2\sum^n_{j=1}(1-\lambda_j)}\ge \min_{\phi \in \mathcal{G}_{\mathrm{pure}}} \norm{\rho - \phi}_1 &\ge 1-\lambda_{1} .
\label{non-Gaussianitymeasure111}
\end{align}
\end{proposition}
\begin{proof}
    The lower bound is given by Theorem~\ref{th:lowerboundnon-Gaussianity0}. Let us focus on the upper bound.
    Let $U_O$ the Gaussian unitary associated to the orthogonal matrix $O$ that puts the correlation matrix of $\rho$ in its normal form, i.e., $\Gamma(\rho)=O \bigoplus^n_{j=1} (\mathrm{diag}(\lambda_1,\dots,\lambda_n)\otimes iY) O^{\dag}$. Let us consider the state $\rho^{\prime}\coloneqq U^{\dagger}_O \rho U_O$. Note that, because of Lemma~\ref{prop:transfFGU}, we have $\Gamma(\rho^{\prime})=\Lambda$.
For $k\in [n]$, we have
\begin{align}
    \Tr(Z_k \rho^{\prime}) &=\Gamma(\rho^{\prime})_{2k-1,k} = \lambda_k=1-(1-\lambda_k)
\end{align}
where $Z_k=-i\gamma_{2k-1}\gamma_{2k}$ represents the $Z$-Pauli operator acting on the $k$-th qubit. Consequently, we also find $\Tr(\ketbra{0}{0}_k \rho^{\prime})= 1 - (1-\lambda_k)/2$.
Let us consider now a pure free-fermionic state of the form $\hat{\phi}\coloneqq U_O \ketbra{0^{n}}{0^{n}} U^{\dagger}_O$.
Employing the unitary invariance of the one-norm and Lemma~\ref{le:gentle}, we derive
\begin{align}
    \|\rho -  \hat{\phi}\|_1 = \norm{\rho^{\prime} -  \ketbra{0^{n}}{0^{n}}}_{1}&\le 2 \sqrt{\sum^n_{j=1}  (1-\lambda_j)/2}.
\end{align}
Thus, we conclude by using the inequality
\begin{align}
\min_{\phi \in \mathcal{G}_{\mathrm{pure}}} \norm{\rho - \phi}_1 & \le \|\rho - \hat{\phi}\|_1,
\end{align}
where we have used that $\hat{\phi}\in \mathcal{G}_{\mathrm{pure}}$.
\end{proof}

It is crucial to mention that in our previous discussion (Theorem~\ref{th:lowerboundnon-Gaussianity0}), we focused on the set of free-fermionic states with a rank no greater than $R$, while we had no assumption about the state $\rho$. 
Next, we establish a lower bound on the distance between a quantum state $\rho$, that we assume to have a bounded rank, and the set of all free-fermionic states $\mathcal{G}_{\mathrm{mixed}}$. 
For pedagogical reasons, we begin by assuming that $\rho$ is pure. Later, we will extend our analysis to $\rho$ having a rank at most $R$.
\begin{theorem}[Lower bound on the distance of a pure state from the set of all free-fermionic states]\label{th:lowerboundnon-Gaussianitypure}
Let $\psi$ be a pure state. Let $\lambda_{\mathrm{min}}$ denote the smallest normal eigenvalue of the correlation matrix of $\psi$ (where the normal eigenvalues are chosen to be all non-negative). Then, we have
\begin{align}
\min_{\sigma \in \mathcal{G}_{\mathrm{mixed}}} \norm{\psi - \sigma}_1 &\ge \frac{1}{2}\left(1-\lambda_{\mathrm{min}}\right).
\end{align}
Here, $\mathcal{G}_{\mathrm{mixed}}$ is the set of all free-fermionic states.
\end{theorem}

\begin{proof}
Let $\sigma\in\mathcal{G}_{\mathrm{mixed}}$ be a free-fermionic state.
Let $0\le \lambda^{\psi}_1\le \dots \le \lambda^{\psi}_n  $ be the (positive) normal eigenvalues of $\Gamma(\psi)$ and $0\le \lambda^{\sigma}_1\le \dots \le \lambda^{\sigma}_n  $ the (positive) normal eigenvalues of $\Gamma(\sigma)$ for an arbitrary free-fermionic state $\sigma$.
We have
\begin{equation}
\norm{\psi-\sigma}_1\ge \norm{\Gamma(\psi)-\Gamma(\sigma)}_\infty\ge \max_{i\in [n]}|\lambda^{\psi}_i-\lambda^{\sigma}_i|\ge |\lambda^{\psi}_1-\lambda^{\sigma}_1|,
\label{eq:boundnow}
\end{equation}
where in the first step we have used Proposition~\ref{le:tracedistancelowerboundcormatrix}, and in the second step Lemma~\ref{le:lbanti-symm}. 
Let $\{z_i\}^n_{i=1}$ be the eigenvalues (in decreasing order) of the (possibly) mixed free-fermionic state $\sigma$.
We now have 
\begin{equation}
\norm{\psi-\sigma}_{1}\ge \norm{\operatorname{diag}(1,0,\ldots, 0)-\operatorname{diag}(z_1,\ldots,z_{d})}_1=|1-z_1|+\sum_{i=2}^{d}z_i=2(1-z_1) ,\label{5}
\end{equation}
where in the first inequality we have used $\norm{A-B}_p\ge \norm{\mathrm{Eig}(A)-\mathrm{Eig}(B)}_{p}$ for any Hermitian matrix $A$ and $B$ and any p-norms $\norm{\cdot}_p$ (see  Ref.~\cite{bhatia1996matrix}, Eq.\ (IV.62)), where $\mathrm{Eig}(A)$ indicates the diagonal matrix with elements the ordered eigenvalues of $A$. In the last step we have used that the eigenvalues of a state are positive and they add up to one.

We know that the eigenvalues of $\sigma$ are related to the (positive) normal eigenvalues $\{\lambda_i^\sigma\}^n_{i=1}$ of the correlation matrix $\Gamma(\sigma)$ because the diagonal form of $\sigma$ reads $\bigotimes_{i=1}^{n}\frac{I+\lambda_i^\sigma Z_i}{2}$ (because of definition~\ref{def:freestate} of free-fermionic state). Therefore, 
\be z_1=\prod_{i=1}^{n}\frac{1}{2}(1+\lambda_i^\sigma)\le\frac{1+\lambda_1^{\sigma}}{2}.
\ee 
%where we have used that $\lambda_n^{\sigma}=\min_{i\in[n]}\lambda_{i}^{\sigma}$.
Plugging this upper bound on $z_1$ into Eq.~\eqref{5}, we get
\begin{align}
\norm{\psi-\sigma}_1\ge 1-\lambda_1^{\sigma}.
\label{eq:boundnow2}
\end{align}
Therefore, we have two lower bounds~Eqs.\eqref{eq:boundnow},\eqref{eq:boundnow2} to the $1$-norm distance $\norm{\psi-\sigma}_1$ . We can consider the minimum of the one-norm difference over the free-fermionic states $\sigma$ and lower bound it with the maximum between these two lower bounds (valid for any $\sigma$). We thus write
\begin{align}
\min_{\sigma\in\mathcal{G}_{\mathrm{mixed}}}\norm{\psi-\sigma}_1&\ge\min_{\sigma\in\mathcal{G}_{\mathrm{mixed}}}\max\{|\lambda_1^{\psi}-\lambda_1^\sigma|,1-\lambda_1^{\sigma}\}\\ 
\nonumber 
&= \min_{\lambda_1^\sigma\in [0,1]}\max\{|\lambda_1^{\psi}-\lambda_1^\sigma|,(1-\lambda_1^{\sigma})\}\\
\nonumber 
&=  \frac{1-\lambda_1^\psi}{2},
\nonumber 
\end{align}
where the last step follows because the function is maximized for $\lambda_1^\sigma=(1+\lambda_1^\psi)/2$. Therefore, we obtain the lower bound 
\begin{equation}
\min_{\sigma\in\mathcal{G}_{\mathrm{mixed}}}\norm{\psi-\sigma}_{1}\ge\frac{1-\lambda_1^{\psi}}{2}\equiv \frac{1}{2}(1-\lambda_{\mathrm{min}})%\label{eq:LBmixed}
\end{equation}
on the minimal trace norm. 
This proves our claim.
\end{proof}

\begin{theorem}[Lower bound on the distance of a bounded rank state from the set of all free-fermionic states]
\label{le:lbrank}
Let $\rho$ be a quantum state such that $\operatorname{rank}(\rho) \leq R$. Assume that $R \leq 2^{r}$ for some $r \in \{0, \dots, n-1\}$. Let $0 \leq \lambda_{1}^{\rho} \leq \cdots \leq \lambda_{n}^{\rho}$ be the (positive) normal eigenvalues of the correlation matrix $\Gamma(\rho)$. Then, we have
\begin{align}
\min_{\sigma \in \mathcal{G}_{\mathrm{mixed}}} \|\rho - \sigma\|_{1} \geq \frac{(1 - \lambda_{r+1}^{\rho})^{r+1}}{1 + (r+1)(1 - \lambda_{r+1}^{\rho})^r}.
\end{align}
In particular, this implies that for $\lambda_{r+1}^{\rho} \geq \frac{1}{2}$, we have
\begin{align}
    \min_{\sigma \in \mathcal{G}_{\mathrm{mixed}}} \|\rho - \sigma\|_{1} \geq \frac{(1 - \lambda_{r+1}^{\rho})^{r+1}}{2}.
\end{align}
For $r = 1$, this reduces to Theorem~\ref{th:lowerboundnon-Gaussianitypure} valid for $\rho$ being pure.
\end{theorem}

\begin{proof}
Let $\sigma \in \mathcal{G}_{\mathrm{mixed}}$ be a possibly mixed free-fermionic state. 
Let $0 \leq \lambda_{1}^{\sigma} \leq \cdots \leq \lambda_{n}^{\sigma}$ be the (positive) normal eigenvalues of the correlation matrix $\Gamma(\sigma)$.
We have
\begin{align}
\label{eq:lbm1}
    \|\rho - \sigma\|_1 \geq \|\Gamma(\rho) - \Gamma(\sigma)\|_{\infty} \geq |\lambda_{r+1}^{\rho} - \lambda_{r+1}^{\sigma}|,
\end{align}
where in the first step we have used Proposition~\ref{le:tracedistancelowerboundcormatrix} and in the second step we have used Lemma~\ref{le:lbanti-symm}. Let $\{z_i^\sigma\}_{i=1}^n$ be the eigenvalues of $\sigma$ ordered in a non-increasing way. Then, similar to Eq.~\eqref{5}, we can lower bound the norm difference between $\rho$ and $\sigma$ with the norm difference of their diagonal matrices containing their eigenvalues (see Ref.~\cite{bhatia1996matrix}, Eq.\ (IV.62)):
\begin{align}
\label{eq:starttrdist}
    \|\rho - \sigma\|_{1} &\geq \|\operatorname{diag}(z_{1}^{\rho}, \ldots, z_{R}^{\rho}, 0, \ldots, 0) - \operatorname{diag}(z_{1}^{\sigma}, \ldots, z_{R}^{\sigma}, z_{R+1}^{\sigma}, \ldots, z_{d}^{\sigma})\|_1 \\
    \nonumber
    &= \sum_{i=1}^{R} |z_{i}^{\rho} - z_{i}^{\sigma}| + \sum_{i=R+1}^{d} z_{i}^{\sigma} \\
     \nonumber
    &\ge\sum_{i=1}^{R} (|z_{i}^{\rho}| - |z_{i}^{\sigma}|) + \sum_{i=R+1}^{d} z_{i}^{\sigma} \\
     \nonumber
    &= 1 - \sum_{i=1}^{R} z_{i}^{\sigma} + \sum_{i=R+1}^{d} z_{i}^{\sigma}= 2 - 2 \sum_{i=1}^{R} z_{i}^{\sigma}= 2 - 2 \|\sigma\|_{\mathrm{KF},R},
     \nonumber
\end{align}
where
in the last step we have used the definition of the Ky Fan norm of order $R$ of $\sigma$, i.e., $\|\sigma\|_{\mathrm{KF},R} \coloneqq \sum_{i=1}^R z_i$ (see Bhathia~\cite{bhatia1996matrix} for more details).
Because of definition~\ref{def:freestate} of free-fermionic state, $\sigma$ up to (free-fermionic) unitary is equivalent to
\begin{align}
    \bigotimes_{i=1}^{n}\frac{I+\lambda^{\sigma}_i Z_i}{2}= \sigma_{[1,r+1]}\otimes \sigma_{[r+1,n]},
\end{align}
where we have defined 
\begin{align}
\sigma_{[a,b]}\coloneqq \bigotimes_{i=a}^{b}\frac{I+\lambda^{\sigma}_i Z_i}{2}.
\end{align}
Since the Ky Fan norm is unitary invariant, we have 
\begin{align}
    \label{eq:lenbs}
    \|\sigma\|_{\mathrm{KF},R}=\|\sigma_{[1,r+1]}\otimes \sigma_{[r+1,n]}\|_{\mathrm{KF},R}&=\left\|\sigma_{[1,r+1]}\otimes \sum_{k=1}^{2^{n-r-1}}\alpha_k \ketbra{v_k} \right\|_{\mathrm{KF},R}\\
    \nonumber
    &\leq \sum_{k=1}^{2^{n-r-1}}\alpha_k\left\|\sigma_{[1,r+1]}\otimes  \ketbra{v_k} \right\|_{\mathrm{KF},R}\nonumber\\
    \nonumber
    &= \sum_{k=1}^{2^{n-r-1}}\alpha_k\|\sigma_{[1,r+1]}\|_{\mathrm{KF},R}\nonumber
    \\&\leq \|\sigma_{[1,r+1]}\|_{\mathrm{KF},R}\nonumber
\end{align}
where in the second step we have used wrote $\sigma_{[r+1,n]}$ in its eigendecomposition $\sigma_{[r+1,n]}=\sum_{k=1}^{2^{n-r-1}}\alpha_k \ketbra{v_k} $, in the third step we have used the triangle inequality of the Ky Fan norm, in the fourth step we have used that 
$\left\|\rho\otimes  \ketbra{v_k}\right\|_{\mathrm{KF},R}=\left\|\rho\right\|_{\mathrm{KF},R}$ and in the last step that the sum of eigenvalues of a state is smaller than one.
Since we have $\lambda_{j}^{\sigma}\le \lambda_{r+1}^{\sigma}$ for each $j\in [r]$, we also have
\begin{align}
    \|\sigma_{[1,r+1]}\|_{\mathrm{KF},R}=\left\|\bigotimes_{i=1}^{r+1}\frac{I+\lambda^{\sigma}_i Z_i}{2}\right\|_{\mathrm{KF},R}\le \left\|\bigotimes_{i=1}^{r+1}\frac{I+\lambda^{\sigma}_{r+1} Z_i}{2}\right\|_{\mathrm{KF},R}.
    \label{eq:lenb}
\end{align}
This can be seen as follows. If $\lambda^{\sigma}_{r+1}=0$, this is trivially true, so let us assume that it is not. We then have
\begin{align}
 \|\sigma_{[1,r+1]}\|_{\mathrm{KF},R}\nonumber&=\norm{\frac{I+\lambda^{\sigma}_1 Z}{2}\otimes \sigma_{[2,r+1]}}_{\mathrm{KF},R} \\
 & =\norm{\left(\frac{\lambda^{\sigma}_{1}}{\lambda^{\sigma}_{r+1}} \right) \frac{I+\lambda^{\sigma}_{r+1} Z}{2}\otimes \sigma_{[2,r+1]} + \left(1-\frac{\lambda^{\sigma}_{1}}{\lambda^{\sigma}_{r+1}} \right) \frac{I}{2}\otimes \sigma_{[2,r+1]}}_{\mathrm{KF},R}\nonumber\\
  & \le \left(\frac{\lambda^{\sigma}_{1}}{\lambda^{\sigma}_{r+1}} \right) \norm{\frac{I+\lambda^{\sigma}_{r+1} Z}{2}\otimes \sigma_{[2,r+1]}}_{\mathrm{KF},R}  + \left(1-\frac{\lambda^{\sigma}_{1}}{\lambda^{\sigma}_{r+1}} \right) \norm{ \frac{I}{2}\otimes \sigma_{[2,r+1]}}_{\mathrm{KF},R} \nonumber \\
  &\le \norm{\frac{I+\lambda^{\sigma}_{r+1} Z}{2}\otimes \sigma_{[2,r+1]}}_{\mathrm{KF},R},
\end{align}
where in the third step we have used the triangle inequality, and in the last step we have used the fact that $\lambda_{1}^{\sigma}\le \lambda_{r+1}^{\sigma}$ and that $\left\| \frac{I}{2}\otimes \sigma_{[2,r+1]}\right\|_{\mathrm{KF},R}=\left\|\sigma'\right\|_{\mathrm{KF},R/2}$ (since appending the maximally mixed state doubles multiplicity of the singular values, while halving their respective value).
Repeating the same strategy above, we arrive at Eq.\ \eqref{eq:lenb}.
Therefore, using Eq.\ \eqref{eq:lenbs} and Eq.\ \eqref{eq:lenb}, we have
\begin{align}
\label{eq:kf2}
   \|\sigma\|_{\mathrm{KF},R} \leq \|\sigma_{[1,r+1]}\|_{\mathrm{KF},R} &\leq \left\|\bigotimes_{i=1}^{r+1} \frac{I + \lambda^{\sigma}_{r+1} Z_i}{2}\right\|_{\mathrm{KF},R}  \\
   &= \max_{\substack{S \subset \{0,1\}^{r+1} \\ |S| = R}} \sum_{s \in S} \prod_{i=1}^{r+1} \left(\frac{1 + (-1)^{s_i}\lambda^{\sigma}_{r+1}}{2}\right) \nonumber \\
   &= 1 - \min_{\substack{S_c \subset \{0,1\}^{r+1} \\ |S_c| = 2^{r+1} - R}} \sum_{s \in S_c} \prod_{i=1}^{r+1} \left(\frac{1 + (-1)^{s_i}\lambda^{\sigma}_{r+1}}{2}\right) \nonumber \\
   &\leq 1 - (2^{r+1} - R) \left(\frac{1 -\lambda^{\sigma}_{r+1}}{2}\right)^{r+1}\nonumber \\
   &\leq 1 - \frac{(1 -\lambda^{\sigma}_{r+1})^{r+1}}{2}\nonumber,
\end{align}
where in the second to last step we have bounded the sum by its smallest term, and in the last step we have used that $ R \le 2^r$.
Now, employing Eq.\ \eqref{eq:starttrdist} and Eq.\ \eqref{eq:kf2}, we get
\begin{align}
\label{eq:lbm2}
    \norm{\rho -\sigma}_1\ge 2 - 2 \|\sigma\|_{\mathrm{KF},R}= (1 -\lambda^{\sigma}_{r+1})^{r+1}.
\end{align}
Thus, by minimizing $\|\rho - \sigma\|_1$ over all possible mixed free-fermionic states $\sigma \in \mathcal{G}_{\mathrm{mixed}}$, and using Eq.~\eqref{eq:lbm1} and Eq.~\eqref{eq:lbm2}, we obtain

\begin{align}
\label{eq:norm_diff_cases}
\min_{\sigma \in \mathcal{G}_{\mathrm{mixed}}} \|\rho - \sigma\|_1 &\geq \min_{\sigma \in \mathcal{G}_{\mathrm{mixed}}} \max\left((1 - \lambda^{\sigma}_{r+1})^{r+1}, |\lambda^{\sigma}_{r+1} - \lambda^{\rho}_{r+1}|\right) \\
&= \min_{\lambda^{\sigma}_{r+1} \in [0, 1]} \max\left((1 - \lambda^{\sigma}_{r+1})^{r+1}, |\lambda^{\sigma}_{r+1} - \lambda^{\rho}_{r+1}|\right) \nonumber \\
&= \min_{x \in [0, 1]} \max\left((1 - x)^{r+1}, |x - \lambda^{\rho}_{r+1}|\right). \label{eq:maxmin}
\end{align}
\begin{figure}[h]
\centering
\includegraphics[width=0.6\textwidth]{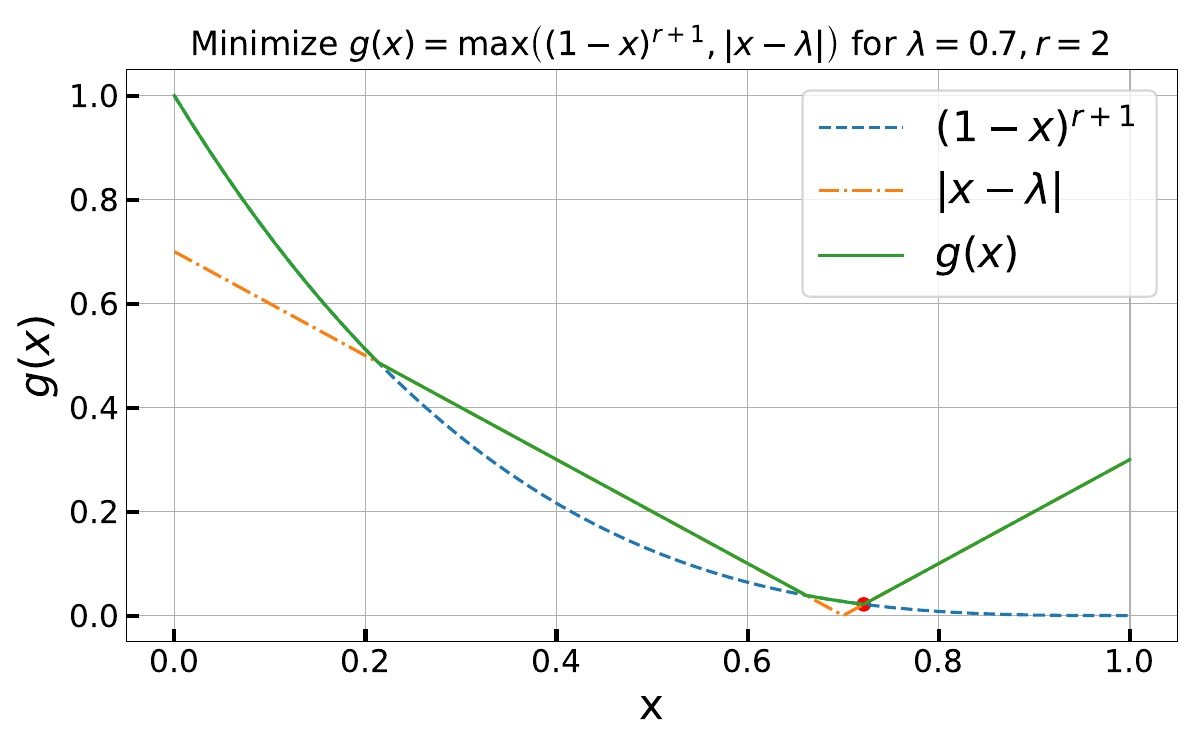}
\caption{Plot illustrating the behavior of the objective function $g(x) = \max\{(1-x)^{r+1}, |x-\lambda|\}$, which arises in Eq.\ \eqref{eq:maxmin}, where $\lambda = 0.7$ and $r = 2$. The red point indicates the location of the minimum value of the function.
}
\label{Fig:plot}
\end{figure}
The minimum point of the function $\max\left((1 - x)^{r+1}, |x - \lambda^{\rho}_{r+1}|\right)$ will occur where the functions $(1 - x)^{r+1}$ and $|x - \lambda^{\rho}_{r+1}|$ intersect. Notably, the relevant intersection will be in the interval $x \in \left[\lambda^{\rho}_{r+1}, 1\right]$ (see Fig.~\ref{Fig:plot} for a graphical representation).

To determine where these two functions intersect, we consider the difference function
\begin{align}
f(x) \coloneqq (1 - x)^{r+1} - (x - \lambda^{\rho}_{r+1}),
\end{align}
for the interval $x \in [\lambda^{\rho}_{r+1}, 1]$. Applying Newton's method starting from $x_0 = \lambda^{\rho}_{r+1}$, we get

\begin{align}
x_1 = x_0 - \frac{f(x_0)}{f'(x_0)} = \lambda^{\rho}_{r+1} + \frac{(1 - \lambda^{\rho}_{r+1})^{r+1}}{1 + (r+1)(1 - \lambda^{\rho}_{r+1})^r}.
\end{align}

Since $f(x)$ decreases monotonically and is convex on this interval, we can be confident that $x_1 \leq x^*$, the root of $f$. Consequently, $(1 - x_1)^{r+1} \geq |x_1 - \lambda^{\rho}_{r+1}|$, which allows us to conclude
\begin{align}
\|\rho - \sigma\|_1 &\geq  \min_{x \in [0, 1]} \max\left((1 - x)^{r+1}, |x - \lambda^{\rho}_{r+1}|\right)\\
\nonumber
&\ge \max 
\left((1 - x_1)^{r+1}, |x_1 - \lambda^{\rho}_{r+1}|\right) \\
\nonumber
&=  |x_1 - \lambda^{\rho}_{r+1}|\\
\nonumber
&=\frac{(1 - \lambda^{\rho}_{r+1})^{r+1}}{1 + (r+1)(1 - \lambda^{\rho}_{r+1})^r}.
\nonumber
\end{align}
In particular, for $\lambda^{\rho}_{r+1} \geq \frac{1}{2}$, 
we have
\begin{align}
1 + (r+1)(1 - \lambda^{\rho}_{r+1})^r \leq 1 + \frac{r+1}{2^r} \leq 2,
\end{align}
which concludes the proof.

\end{proof}

\section{Property testing of free-fermionic states}\label{Sec:propertytesting}
In this section, we address the problem of property testing free-fermionic states. Our goal is to determine whether a given quantum state $\rho$ is close
to or far from the set of all free-fermionic states $\mathcal{G}_{\mathrm{mixed}}$. To enable an efficient testing algorithm, we must make certain assumptions about the rank of the state $\rho$, as we show.
Assuming that $\rho$ is pure (i.e., rank one), we present an efficient learning algorithm that relies solely on single-qubit measurements (as detailed in Subsection~\ref{subsec:purtesting}). 
However, when no assumptions are made about the rank of the state $\rho$, we establish the general hardness of the problem. Information-theoretically, it requires $\Omega(\mathrm{rank}(\rho))$ copies to be solved (as discussed in Subsection~\ref{subsec:hardtesting}). Moreover, we provide a matching upper bound by introducing a an algorithm which utilizes only single-copy free-fermionic measurements and it is efficient as long as $\mathrm{rank}(\rho)=O(\mathrm{poly}(n))$ (explored in Subsection~\ref{subsec:mixedtesting}).

Similarly, we show that there is an efficient property-testing algorithm in the case when $\rho$ is an arbitrary quantum state and the goal is to determine whether a given quantum state $\rho$ is close to or far from the set of pure free-fermionic states $\mathcal{G}_{\mathrm{pure}}$. We also show an information theoretic lower bound to solve the property testing problem when $\rho$ is arbitrary and consider $\mathcal{G}_{R}$, that is, the set of all free-fermionic states with rank at most $R$. This section concludes by providing a matching upper bound for the above scenario.
We formalize the property testing problem as follows.

\begin{problem}[Property testing of free-fermionic states]
\label{prob:testing}
Given $N$ copies of an unknown quantum state $\rho$, with the promise that it falls into one of two distinct scenarios $\varepsilon_B > \varepsilon_A \ge 0$.
\begin{itemize}
    \item \textbf{Case A:} There exists a free-fermionic state $\sigma\in\mathcal{G}$ such that $\| \rho - \sigma \|_{1} \le \varepsilon_A$.
    \item \textbf{Case B:} The state $\rho$ is $\varepsilon_B$-far from all free-fermionic states $\sigma$, i.e., $\min_{\sigma\in\mathcal{G}}\| \rho - \sigma \|_{1} > \varepsilon_B$.
\end{itemize}
Determine whether we are in Case A or Case B by performing arbitrary measurements on the queried copies of the state $\rho$.
\end{problem}
Further restrictions on the state $\rho$ and precise specifications regarding the set of free-fermionic states $\mathcal{G}$ considered (e.g., $\mathcal{G}=\mathcal{G}_{\mathrm{pure}}$, $\mathcal{G}=\mathcal{G}_{\mathrm{mixed}}$, $\mathcal{G}=\mathcal{G}_{R}$) must be provided as input to the problem.
Before delving into more details, we establish 
a preliminary standard lemma that will be instrumental later.

\begin{lemma}[Sample complexity for infinity norm approximation of the correlation matrix via Pauli measurements]
\label{le:corrMatrixEst}
Let $N \ge 16(n^4/\varepsilon_{\mathrm{stat}}^{2}) \log(n^2/\delta)$ be the number of copies of an $n$-qubit state $\rho$. Through single-qubit Pauli-basis measurements, we can find a real and anti-symmetric matrix $\hat{\Gamma}$ such that, with probability at least $1-\delta$, it holds that
\be 
\norm{\hat{\Gamma}-\Gamma(\rho)}_\infty \le \epsstat,\quad \norm{\hat{\Gamma}-\Gamma(\rho)}_2 \le \epsstat.
\ee
\end{lemma}
\begin{proof}
For $j<k \in [2n]$, we estimate $ [\Gamma(\rho)]_{j,k} = -\mathrm{Tr}(i\gamma_j\gamma_k \rho ) $. For each $i\gamma_j\gamma_k$, which are $M:=n(2n-1)$ in total, we measure $N' \ge \frac{2}{\varepsilon^2}\log\left(\frac{2M}{\delta}\right) $
many copies of $\rho$ in the Pauli basis, obtaining outcomes $\{\hat{X}_m\}^{N'}_{m=1}$, where $\hat{X}_m \in \{-1,+1\}$. Define $\hat{\Gamma}_{j,k} := \frac{1}{N'}\sum^{N'}_{m=1} \hat{X}_m$. By Hoeffding's inequality, we have
    \be
        \operatorname{Pr}\left(\lvert \hat{\Gamma}_{j,k} - [\Gamma(\rho)]_{j,k} \rvert \ge \varepsilon \right) \le 2 \exp\left(-\frac{2N' \varepsilon^2}{(b-a)^2}\right)
    \ee
    where $a:=-1$ and $b:=1$. By the union bound, we have
    \begin{align}
        \operatorname{Pr}\left(\forall\, j<k: \lvert \hat{\Gamma}_{j,k} - [\Gamma(\rho)]_{j,k} \rvert < \varepsilon \right) = 1 - \sum_{i<j}\operatorname{Pr}\left( \lvert \hat{\Gamma}_{j,k} - [\Gamma(\rho)]_{j,k} \rvert \ge \varepsilon \right) \ge 1 - 2 M \exp\left(-\frac{N' \varepsilon^2}{2}\right).
    \end{align}
    By using that $N' \ge \frac{2}{\varepsilon^2}\log\left(\frac{2 M}{\delta}\right)$, we have that this probability is greater than $1-\delta$.
    The total number of measurements is $N=N^{\prime}M$. Finally, employing the Gershgorin circle theorem, we have 
    \be
    \label{eq:eqscm}
    \norm{\hat{\Gamma}-\Gamma(\rho)}_\infty\le 2n\varepsilon. 
    \ee By choosing $\varepsilon=\varepsilon_{\mathrm{stat}}/2n$, we can draw our conclusions. Moreover, note that by construction $\hat{\Gamma}$ is real and anti-symmetric. The proof for the two norm bound concludes by using, instead of \eqref{eq:eqscm}, that 
    \begin{align}
    \norm{\hat{\Gamma}-\Gamma(\rho)}_2\le 2n \max_{i,j \in [2n]} |[\hat{\Gamma}]_{i,j}-[\Gamma(\rho)]_{i,j}|\le 2n \varepsilon.
    \end{align}
\end{proof}

While sequentially estimating correlation matrix entries in the Pauli basis may not be the most sample-efficient, it proves convenient for experiments due to its easy implementation. Alternatively, measuring commuting observables simultaneously~\cite{PartitionBabbush} reduces sample complexity by \(n\) but requires a slightly more intricate setup.
For completeness, we present a Lemma establishing a sample complexity upper bound for estimating the correlation matrix using this refined measurement scheme. The idea is to partition observables \(O^{(j,k)} \coloneqq -i \gamma_j\gamma_k\) into \(2n-1\) sets of commuting observables. Commuting Pauli observables can be measured simultaneously via a Clifford Gaussian measurements~\cite{PartitionBabbush}. Notably, different Pauli observables \(-i \gamma_j\gamma_k\) commute only if associated with different Majorana operators. This allows us to partition \(M=(2n-1)n\) observables into \(2n-1\) sets, each containing \(n\) commuting Pauli observables. We refer to Ref.~\cite{PartitionBabbush}, Appendix C, 
for partition details, omitted here for brevity.

\begin{lemma}[Sample complexity for infinity norm approximation of the correlation matrix via commuting observables]
\label{le:samplecompAPPcommuting}
Let $\epsstat, \delta > 0$. Assume access to $N\ge   \left\lceil (8n^3/\epsstat^{2})\log\!\left(4n^2/\delta\right)\right\rceil$ copies of an $n$-qubit state $\rho$. Using $N$ single-copy measurements, with probability $\ge 1-\delta$, we can construct an anti-symmetric real matrix $\hat{\Gamma}$ such that
\be 
\norm{\hat{\Gamma}-\Gamma(\rho)}_\infty \le \epsstat,\quad \norm{\hat{\Gamma}-\Gamma(\rho)}_2 \le \epsstat.
\ee
\end{lemma}

\begin{proof}
For each of the $2n-1$ sets of commuting Pauli, find the Clifford $C$ allowing simultaneous measurement of such observables, and measure $N'$ copies of $C\rho C^{\dagger}$ in the computational basis. Note that we can choose such Clifford to be also Gaussian. This because such Clifford diagonalizes the free fermionic Hamiltonian given by the sum of the commuting Pauli (which corresponds to the quadratic Majorana operators), hence it can be chosen to be Gaussian. The same can be seen also directly by using the Gaussian unitary definition~\ref{def:freeuni}, see for example Ref.~\cite{mele2024efficient}. The rest of the proof goes as in the previous Lemma, by using Hoeffding's inequality and union bound. 
The total number of measurements needed this time is $N=N^{\prime} (2n-1)$ (whereas in the previous Lemma it was $N^{\prime} (2n-1)n$). We can conclude now also like in the previous lemma. 
\end{proof}

Analogously, we have also the following Lemma, with recovery guarantees expressed in terms of the Frobenious norm differences of the correlation matrices.

\subsection{Efficient testing of pure free-fermionic states}
\label{subsec:purtesting}
We will now show an efficient quantum learning algorithm to solve Problem~\ref{prob:testing} when $\rho$ is a assumed to be pure, having as assumption that $\varepsilon_B, \varepsilon_A \in (0,1)$ are such that $\varepsilon_B>2\sqrt{ n\varepsilon_A}$. The proposed algorithm uses only single copies of the state $\rho$.
The set of free-fermionic state $\mathcal{G}$ that will be considered can be either the set of all free-fermionic states $\mathcal{G}_{\mathrm{mixed}}$ or the set of free-fermionic states restricted to the pure ones $\mathcal{G}_{\mathrm{pure}}$.
The high-level idea of the algorithm is to output $A$ (the state is close to the free-fermionic set) if the eigenvalues of the estimated correlation matrix are all close to $1$, while outputting $B$ otherwise. 
We present now the following Theorem which show the correctness of the algorithm presented in Table~\ref{alg:algoTEST}.

\begin{algorithm}[H]
\caption{Property testing algorithm for pure free-fermionic states}
\label{alg:algoTEST}
\KwIn{ 
    Error thresholds $\varepsilon_A, \varepsilon_B$, failure probability $\delta$. $N:=\lceil 8(n^3/\epsstat^{2}) \log(4n^2/\delta)\rceil$ copies of the pure state $\rho$, where $\epsstat<\frac{1}{2}(\frac{\varepsilon^2_B}{2n}-2\varepsilon_A)$. Let $\epsT:=\frac{1}{2}\left(\frac{\varepsilon^2_B}{2n}+2\varepsilon_A\right)$.
}
\KwOut{
    Output either Case A or Case B.
}

\textbf{Step 1:} Estimate the entries of the correlation matrix using $N$ single-copy measurements, resulting in the estimated $2n \times 2n$ matrix $\hat{\Gamma}$\;

\textbf{Step 2:} Find $\hat{\lambda}_{\mathrm{min}}$, which corresponds to the smallest singular value of $\hat{\Gamma}$\;

\textbf{Step 3:} 
\If{$\hat{\lambda}_{\mathrm{min}} \ge 1 - \epsT$}{
    \textbf{Output:} Case A
}
\Else{
    \textbf{Output:} Case B
}
\end{algorithm}

\begin{theorem}[Efficient pure free-fermionic testing]\label{th:proofofefficientpurestatetesting}
    Let $\rho$ be an $n$-qubit pure state. Assume $\varepsilon_B, \varepsilon_A \in (0,1)$ such that $\varepsilon_B > 2\sqrt{n\varepsilon_A}$, $\delta \in (0,1]$, and $\epsstat < \frac{1}{2}(\frac{\varepsilon^2_B}{2n} - 2\varepsilon_A)$. Assume that $\rho$ is one of the two cases detailed in Problem~\ref{prob:testing}, i.e., there exists a free-fermionic state $\sigma \in \mathcal{G}$ such that $\| \rho - \sigma \|_{1} \le \varepsilon_A$ or $\min_{\sigma\in\mathcal{G}}\| \rho - \sigma \|_{1} > \varepsilon_B$. The set $\mathcal{G}$ considered here can be either the set of all free-fermionic states $\mathcal{G}_{\mathrm{mixed}}$ or the set of pure free-fermionic states $\mathcal{G}_{\mathrm{pure}}$.
    Then there exists a quantum learning algorithm (\ref{alg:algoTEST}) which can solve Problem~\ref{prob:testing} using $N=8(n^3/\epsstat^{2}) \log(4n^2/\delta)$ single-copies measurements of the state $\rho$ with a probability of success at least $1-\delta$. 
\end{theorem}

\begin{proof}
    Let $\epsstat>0$ be an accuracy parameter to be fixed later. By Lemma~\ref{le:samplecompAPPcommuting}, with $N \ge 8(n^3/\varepsilon_{\mathrm{stat}}^{2}) \log(4n^2/\delta)$, single-copies measurements we can find a matrix $\hat{\Gamma}$ such that, with probability at least $1-\delta$, it holds that $\norm{\hat{\Gamma}-\Gamma(\rho)}_\infty < \varepsilon_{\mathrm{stat}}$. This implies that for all $k\in[n]$ (Lemma~\ref{le:lbanti-symm})
    \begin{align}
        |\hat{\lambda}_k - \lambda_k| < \varepsilon_{\mathrm{stat}},
        \label{ineq:eigdiff}
    \end{align}
    where $\{ \hat{\lambda}_k \}^{n}_{k=1}, \{\lambda_k\}^{n}_{k=1}$ are the normal eigenvalues of $\hat{\Gamma}$ and $\Gamma(\rho)$, respectively. 
    We can put $\hat{\Gamma}$ in its normal form \be\hat{\Gamma}=\hat{O}\hat{\Lambda}\hat{O}^{T},\ee where \be\hat{\Lambda}=\bigoplus^n_{i=1} \begin{pmatrix}
        0 & \hat{\lambda}_i\\ - \hat{\lambda}_i & 0 
    \end{pmatrix}
    \ee
    and find its eigenvalues. Our algorithm now works as follows. Let $\epsT$ be a parameter to fix later. If for all $k\in[n]$, we have $ \hat{\lambda}_k\ge 1-\epsT$, then we output $A$, otherwise $B$. 
    In case we output $A$, 
    we need to prove that there exists a free-fermionic state $\sigma$ such that $\| \rho - \sigma \|_{1}\le \varepsilon_B$. From Eq.\ \eqref{ineq:eigdiff}, we have that for all $k\in[2n]$
    \begin{align}
        \lambda_k \ge \hat{\lambda}_k -\epsstat \ge 1-\epsT- \epsstat.
    \end{align}
    The anti-symmetry of the correlation matrix $\Gamma(\rho)$ implies the existence of an orthogonal matrix $O$ such that $\Gamma(\rho) = O\Lambda O^{T}$, where $\Lambda=\bigoplus^n_{i=1} \begin{pmatrix}
        0 & \lambda_i\\ - \lambda_i & 0 
    \end{pmatrix}$. Moreover, because of Lemma~\ref{prop:transfFGU}, we have $\Gamma(U_O^{\dagger}\rho U_O)=\Lambda$, where $U_O$ is the free-fermionic unitary associated with $O$. This leads to \be \Tr(Z_k U_O^{\dagger}\rho U_O)=\Gamma(U_O^{\dagger}\rho U_O)_{2k-1,2k}=\Lambda_{2k-1,2k}=\lambda_k,\ee 
    where we have used that $Z_k=-i\gamma_{2k-1}\gamma_{2k}$. Consequently, we obtain
    \begin{align}
        \Tr(\ketbra{0}{0}_k U_O^{\dagger}\rho U_O) = \frac{1+\lambda_k}{2}  \ge  1-\frac{\epsT+\epsstat}{2},
    \end{align}
    where we have used that $\ketbra{0}{0}_k=(I + Z_k)/2$.
    Applying the quantum union bound (Lemma~\ref{le:gentle}), we deduce
    \begin{align}
        \norm{U_O^{\dagger}\rho U_O - \ketbra{0^n}{0^n}}_1 \le 2 \sqrt{n\left(\frac{\epsT+\epsstat}{2}\right)}.
    \end{align}
    Thus, taking $\sigma:=U_O\ketbra{0^n}{0^n}U_O^{\dagger}$ and using the unitary invariance of the one-norm, we have successfully demonstrated the existence of a free-fermionic state $\sigma$ that closely approximates $\rho$ in terms of the one-norm distance, i.e.,  $\norm{\rho  - \sigma}_1 \le \sqrt{2n\left(\epsT+\epsstat\right)}$. To ensure the validity of this approximation, we must impose the condition
    \begin{align}
        \sqrt{2n(\epsT+ \epsstat)} \le \varepsilon_B,
        \label{eq:cond1eps}
    \end{align}
    which constitutes the initial requirement for determining the values of $\epsT$ and $\epsstat$, which we will fix at a later stage.
    Let us analyze the case in which we output case $B$, i.e.,  we observe there exists $k\in [2n]$ such that $ \hat{\lambda}_k< 1-\epsT$. In this case, from Eq.\ \eqref{ineq:eigdiff}, we have
    \begin{align}
        \lambda_k \le \hat{\lambda}_k +\epsstat < 1-\epsT+ \epsstat.
        \label{eq:boundup}
    \end{align}
Using Theorem~\ref{th:lowerboundnon-Gaussianity0} and Theorem~\ref{th:lowerboundnon-Gaussianitypure}, we establish the following inequality
\begin{equation}
    \min_{\sigma\in\mathcal{G}}\norm{\rho-\sigma}_{1} \ge \frac{1}{2}(1 - \underset{k\in [n]}{\min}\lambda_k ),
    \label{eq:LBmixed}
\end{equation}
where $\mathcal{G}$ can represent either the set of all free-fermionic states, denoted as $\mathcal{G}_{\mathrm{mixed}}$, or the set of all pure free-fermionic states, denoted as $\mathcal{G}_{\mathrm{pure}}$.
If we had considered only the set of pure free-fermionic states, the lower bound would be $1 - \underset{k\in [n]}{\min} \lambda_k$, without the factor of one-half, as indicated by Theorem~\ref{th:lowerboundnon-Gaussianity0}. As such, Eq.~\eqref{eq:LBmixed} holds true for both $\mathcal{G}_{\mathrm{mixed}}$ and $\mathcal{G}_{\mathrm{pure}}$.
From this and Eq.\ \eqref{eq:boundup} we have that
    \begin{align}
        \min_{\sigma \in \mathcal{G}}\norm{\rho - \sigma}_1 \ge \frac{1-(1-\epsT+ \epsstat)}{2} = \frac{\epsT - \epsstat}{2}.
    \end{align}
    Therefore, we impose that
    \begin{align}
        \frac{\epsT - \epsstat}{2} > \varepsilon_{A}.
        \label{eq:cond2eps}
    \end{align}
    Putting together the two inequalities in Eq.\ \eqref{eq:cond1eps} and Eq.\ \eqref{eq:cond2eps}, we have
    \begin{align}
        2\varepsilon_A + \epsstat < \epsT \le \frac{\varepsilon^2_B}{2n} - \epsstat.
    \end{align}
    Therefore, assuming $\varepsilon_B>\sqrt{4n \varepsilon_A}$, we can choose 
    \be\epsT=\frac{1}{2}(\frac{\varepsilon^2_B}{2n}+2\varepsilon_A)
    \ee and \be \epsstat<\frac{1}{2}(\frac{\varepsilon^2_B}{2n}-2\varepsilon_A).\ee 
\end{proof}
The preceding theorem has been presented under the assumption that $\rho$ is pure, and the set of free-fermionic states considered can either be the set of all (possibly mixed) free-fermionic states $\mathcal{G}_{\mathrm{mixed}}$ or the more restricted set of all pure free-fermionic states $\mathcal{G}_{\mathrm{pure}}$. However, if we focus solely on the set of all pure free-fermionic states $\mathcal{G}_{\mathrm{pure}}$, we can establish an analogous result without assuming that $\rho$ is pure; i.e., it can be an arbitrary quantum state. The theorem is detailed as follows, and the algorithm is the same as Algorithm \ref{alg:algoTEST} with slightly different accuracy parameters, as detailed below.

\begin{theorem}[Efficient Pure free-fermionic testing with arbitrary input states]
    Let $\rho$ be an arbitrary $n$-qubit state. Assume $\varepsilon_B, \varepsilon_A \in (0,1)$ such that $\varepsilon_B > \sqrt{2n\varepsilon_A}$, $\delta \in (0,1]$, and $\epsstat := \frac{1}{4}(\frac{\varepsilon^2_B}{2n} - \varepsilon_A)$. Assume that $\rho$ satisfies one of the two cases detailed in Problem~\ref{prob:testing}, i.e., there exists a free-fermionic state $\sigma \in \mathcal{G}_{\mathrm{pure}}$ such that $\| \rho - \sigma \|_{1} \le \varepsilon_A$ or $\min_{\sigma\in\mathcal{G}_{\mathrm{pure}}}\| \rho - \sigma \|_{1} > \varepsilon_B$. 
    Then, there exists a quantum learning algorithm which, utilizing only single-copies measurements, can solve Problem~\ref{prob:testing} using $N=8(n^3/\epsstat^{2}) \log(4n^2/\delta)$ copies of the state $\rho$ with a probability of success at least $1-\delta$. 
\end{theorem}
\begin{proof}
    The proof is analogous to the one of the previous theorem, but instead of Eq.~\eqref{eq:LBmixed}, it utilizes Lemma~\ref{th:lowerboundnon-Gaussianitypure}, which has no assumptions on the state $\rho$ and provides the inequality
    \begin{align}
        \min_{\sigma\in\mathcal{G}_{\mathrm{pure}}}\norm{\rho-\sigma}_{1}\ge 1-\underset{k\in [n]}{\min}( \lambda_k ).
    \end{align}
    Following the same steps as before, we conclude that, assuming $\varepsilon_B>\sqrt{2n \varepsilon_A}$, we can choose $\epsT=\frac{1}{2}(\frac{\varepsilon^2_B}{2n}+\varepsilon_A)$ and $\epsstat=\frac{1}{4}(\frac{\varepsilon^2_B}{2n}-\varepsilon_A)$. Note that, throughout the proof  of Theorem~\ref{th:proofofefficientpurestatetesting}, there has been no need to assume that $\rho$ is a pure state. 
\end{proof}

\subsection{Hardness of testing general mixed free-fermionic states}
\label{subsec:hardtesting}
In this section, we establish the general hardness of the free-fermionic property testing problem (Problem~\ref{prob:testing}), demonstrating the necessity for \( \Omega(2^n) \) copies of the state when no prior assumptions on the state and no restrictions on the set of all free-fermionic states are provided. The core of this complexity arises from recognizing that the maximally mixed state is free-fermionic. This insight allows us to leverage the hardness of a problem closely related to identity testing~\cite{odonnell2015quantum}. The following theorem is essential to our reduction:
\begin{theorem}[Hardness of Identity Testing, corollary of Theorem 1.12 in \cite{odonnell2015quantum}]  
\label{th:id_hard}  
Let \( \rho \) be an \( n \)-qubit quantum state with \( n > 1 \), promised to belong to one of the following two cases:  
\begin{itemize}  
    \item \( \rho_A = \frac{I}{2^n} \), the maximally mixed state; or  
    \item \( \rho_B \in \left\{ \frac{1}{K} \sum_{i=1}^K U \ketbra{i}{i} U^\dagger \mid U \in U(2^n) \right\} \),  
\end{itemize}  
where \( K \coloneqq \frac{3}{4} 2^n \) is a positive integer. Any algorithm that distinguishes between these two cases with success probability at least \( 2/3 \) requires \( \Omega(2^n) \) copies of the state \( \rho \).
\end{theorem}

\begin{proof}
The proof of this theorem is a direct corollary of Theorem 1.12 proven in Ref.~\cite{odonnell2015quantum}. 
\end{proof}
In the following, we establish the hardness of testing free-fermionic states in the non-tolerant setting, i.e., when \(\varepsilon_A = 0\). Note that this result also implies hardness in the more general and difficult scenario of `tolerant property testing,' where \(\varepsilon_A > 0\).
\begin{theorem}[Hardness of testing free-fermionic states (Problem~\ref{prob:testing})]
\label{th:hardnesstesting}
At least $N=\Omega(2^n)$ copies of the $n$-qubits state $\rho$ are necessary to solve the free-fermionic property testing problem (Problem~\ref{prob:testing}) with parameters $\varepsilon_A=0$ and $\varepsilon_B=O(1)$ with a probability of success at least $\frac{2}{3}$. This holds when no prior assumptions about $\rho$ are provided, and the set $\mathcal{G}$ considered is that of all free-fermionic states. 
\end{theorem}
\begin{proof}
We reduce the free-fermionic property testing problem to the problem in Theorem~\ref{th:id_hard}. To establish equivalence, we need to show that \(\rho_A\) is free-fermionic, and that any state \(\rho_B\) (as defined in Theorem~\ref{th:id_hard}) is far from being free-fermionic.
The first condition is trivial since \(\rho_A\), the maximally mixed state, is free-fermionic. For the second condition, we need to show that any free-fermionic state \(\sigma\) is far in trace norm from all \(\rho_B\). We denote by \(\{\lambda_i^\downarrow(\rho)\}\) the eigenvalues of \(\rho\) in decreasing order and by \(\{\lambda_i^\uparrow(\rho)\}\) the eigenvalues in increasing order. 
We have:
\begin{align}
    \|\rho_B - \sigma\|_1 &\geqt{(i)} \|\mathrm{Eig}(\rho_B) - \mathrm{Eig}(\sigma)\|_{1} \\
    &=  \sum_{i=1}^{2^n} |\lambda_i^\downarrow(\rho_B) - \lambda_i^\downarrow(\sigma)| \\
    &\geq \sum_{i=1}^{2^{n}/4} |\lambda_i^\downarrow(\rho_B) - \lambda_i^\downarrow(\sigma)| + \sum_{i=\frac{3}{4}2^{n} + 1}^{2^n} |\lambda_i^\downarrow(\rho_B) - \lambda_i^\downarrow(\sigma)| \\
    &= \sum_{i=1}^{2^{n}/4} |\lambda_i^\downarrow(\rho_B) - \lambda_i^\downarrow(\sigma)| + \sum_{i=1}^{2^{n}/4} |\lambda_i^\uparrow(\rho_B) - \lambda_i^\uparrow(\sigma)| \\
    &\geqt{(ii)}  \sum_{i=1}^{2^{n}/4} (\lambda_i^\downarrow(\sigma) + \lambda_i^\uparrow(\sigma)) - \sum_{i=1}^{2^{n}/4} (\lambda_i^\downarrow(\rho_B) + \lambda_i^\uparrow(\rho_B)) \\
    &\geqt{(iii)} \sum_{i=1}^{2^{n}/4} (\lambda_i^\downarrow(\sigma) + \lambda_i^\uparrow(\sigma)) - \frac{1}{3} \\
    &\geqt{(iv)} \frac{1}{2} - \frac{1}{3} = \frac{1}{6},
\end{align}
where in step \text{(i)} we used that $\|\mathrm{Eig}(A) - \mathrm{Eig}(B)\|_{p} \le \|A - B\|_p$ for any Hermitian matrices $A$ and $B$ and any p-norms $\|\cdot\|_p$ (see Ref.~\cite{bhatia1996matrix}, Eq.~(IV.62)), where $\mathrm{Eig}(A)$ indicates the diagonal matrix with elements being the ordered eigenvalues of $A$.
In step \text{(ii)}, we used the fact that $|x|\ge \max{(-x,x)}$.
In step \text{(iii)}, we used the fact that the spectrum of \(\rho_B\) has \(\frac{3}{4}2^{n}\) non-zero eigenvalues equal to \((\frac{3}{4}2^{n})^{-1}\), with only \(2^{n}/4\) eigenvalues contained in the largest and smallest quartiles. In step \text{(iv)}, we used Lemma~\ref{lem:spec_of_prod} and noted that the spectrum of free-fermionic states matches that of product states.
\end{proof}
We now prove a Lemma of independent interest concerning the spectrum of product states which we used in the above proof. 
\begin{lemma}[Spectrum of product states]\label{lem:spec_of_prod}
    For an \(n\)-qubit product state \(\sigma=\sigma_1\otimes\cdots\otimes \sigma_n\), we have the inequality
    \begin{align}
        \sum_{i=1}^{k} \left(\lambda_i^\downarrow(\sigma) + \lambda_i^\uparrow(\sigma) \right)\geq \frac{k}{2^{n-1}},
    \end{align}
    where \(k \leq 2^{n-1}\). Here, \(\{\lambda_i^\downarrow(\sigma)\}^{2^n}_{i=1}\) and \(\{\lambda_i^\uparrow(\sigma)\}^{2^n}_{i=1}\) denote the eigenvalues of \(\sigma\) in decreasing and increasing order, respectively.
\end{lemma}
\begin{proof}
We prove the result by induction on \(n\), the number of qubits.
For the base case \(n = 1\), the state \(\sigma = \sigma^{(1)}\) is a single-qubit density matrix with only two eigenvalues which sum to one. Since $k$ can only be either $0$ or $1$, the result trivially follows. For the induction step, we assume the result holds for \(n\)-qubit product states $\sigma^{(n)}$, i.e.,
\[
\sum_{i=1}^{k} \left(\lambda_i^\downarrow(\sigma^{(n)}) + \lambda_i^\uparrow(\sigma^{(n)}) \right) \geq \frac{k}{2^{n-1}}.
\]
We now prove the statement for \((n+1)\)-qubit product states.
Let \(\sigma^{(n+1)} = \sigma^{(n)} \otimes \sigma_{n+1}\), where \(\sigma_{n+1}\) is a single-qubit state with eigenvalues \(p_{n+1}\) and \(1 - p_{n+1}\). The eigenvalues of \(\sigma^{(n+1)}\) are given by the products of the eigenvalues of \(\sigma^{(n)}\) and \(\sigma_{n+1}\).  
For a suitable choice of positive integers \( k' \) and \( k'' \) such that \( k = k' + k'' \), we can express the sum of the eigenvalues of \( \sigma^{(n+1)} \) in the increasing order as:
\begin{align}
    \sum_{i=1}^{k} \lambda_i^\uparrow(\sigma^{(n+1)}) = \sum_{i=1}^{k'} \lambda_i^\uparrow(\sigma^{(n)}) p_{n+1} + \sum_{i=1}^{k''} \lambda_i^\uparrow(\sigma^{(n)}) (1 - p_{n+1}).
\end{align}
Without loss of generality, assume \( p_{n+1} \leq \frac{1}{2} \), which implies \( k' \geq k'' \). We can also write
    \begin{align}
        \sum_{i=1}^{k} \lambda_i^\downarrow(\sigma^{(n+1)})\geq\sum_{i=1}^{k'} \lambda_i^\downarrow(\sigma^{(n)}) (1-p_{n+1})+\sum_{i=1}^{k''} \lambda_i^\downarrow(\sigma^{(n)}) p_{n+1},
\end{align}
where the inequality comes from us simply selecting these eigenvalues.
Combining the two inequalities, we obtain:
\begin{align}
   & \sum_{i=1}^{k} \left( \lambda_i^\downarrow(\sigma^{(n+1)}) + \lambda_i^\uparrow(\sigma^{(n+1)}) \right) 
    \\&\geq \sum_{i=1}^{k'} \lambda_i^\downarrow(\sigma^{(n)}) (1-p_{n+1}) + \sum_{i=1}^{k''} \lambda_i^\downarrow(\sigma^{(n)}) p_{n+1} 
    + \sum_{i=1}^{k'} \lambda_i^\uparrow(\sigma^{(n)}) p_{n+1} + \sum_{i=1}^{k''} \lambda_i^\uparrow(\sigma^{(n)}) (1-p_{n+1}) \\
    &= \sum_{i=1}^{k''} \left( \lambda_i^\downarrow(\sigma^{(n)}) + \lambda_i^\uparrow(\sigma^{(n)}) \right) 
    + \sum_{i=k''+1}^{k'} \left( \lambda_i^\uparrow(\sigma^{(n)}) p_{n+1} + \lambda_i^\downarrow(\sigma^{(n)}) (1-p_{n+1}) \right) \\
    &\geq \sum_{i=1}^{k''} \left( \lambda_i^\downarrow(\sigma^{(n)}) + \lambda_i^\uparrow(\sigma^{(n)}) \right) 
    + \sum_{i=k''+1}^{k'} \frac{\lambda_i^\uparrow(\sigma^{(n)}) + \lambda_i^\downarrow(\sigma^{(n)})}{2} \\
    &\geq \sum_{i=1}^{k''} \frac{\lambda_i^\uparrow(\sigma^{(n)}) + \lambda_i^\downarrow(\sigma^{(n)})}{2} 
    + \sum_{i=1}^{k'} \frac{\lambda_i^\uparrow(\sigma^{(n)}) + \lambda_i^\downarrow(\sigma^{(n)})}{2} \\
    &\geqt{(i)} \frac{k'}{2^{n-1} \times 2} + \frac{k''}{2^{n-1} \times 2} \\
    &= \frac{k}{2^{(n+1)-1}},
\end{align}
where in (i) we used the inductive assumption.
This concludes the proof.
\end{proof}
%As the reader can appreciate from the proof of Theorem~\ref{th:hardnesstesting}, the hardness of Problem~\ref{prob:testing} arises from the unknown quantum state being arbitrarily close to the maximally mixed state, which is known to be challenging to test. 
We have shown that a number of copies exponential in the system size is necessary to test free-fermionic states, without any additional prior assumption. 
A natural assumption to facilitate Problem~\ref{prob:testing} is to assume the state $\rho$ to have at most a fixed rank. In this context, we establish a fundamental lower bound on the number of copies required to solve the free-fermionic testing problem, which depends on the rank of the quantum state. 
\begin{theorem}[Lower bound for free-fermionic testing of states with bounded rank]  
Let \( \rho \) be a quantum state such that \( \mathrm{rank}(\rho) \le R \). To solve the free-fermionic property testing problem (Problem~\ref{prob:testing}) with parameters \( \varepsilon_A = 0 \) and \( \varepsilon_B = O(1) \) with at least \( 2/3 \) probability of success, \( N = \Omega(R) \) copies are necessary. This holds when considering the set \( \mathcal{G} \) in Problem~\ref{prob:testing} corresponding to the set of all free-fermionic states \( \mathcal{G}_\mathrm{mixed} \).
\end{theorem}
\begin{proof}
Let \( r \coloneqq \lfloor \log_2(R) \rfloor \). We choose \( \rho \) of the form \( \rho = \rho_r \otimes \ketbra{0^{n-r}}{0^{n-r}} \), where \( \ket{0^{n-r}} \) is the zero computational basis state on the last \( (n - r) \)-qubits, which is free-fermionic. We further impose that the \( r \)-qubit state \( \rho_r \) can be either \( \rho_A \) (the maximally mixed state), as denoted in Theorem~\ref{th:id_hard}, or of the form \( \rho_B \), as defined in Theorem~\ref{th:id_hard}. 
Clearly, we have \( \mathrm{rank}(\rho) \le R \). The proof follows the same lines as the proof of Theorem~\ref{th:hardnesstesting}, but within the effective space of \( r \) qubits. 
To establish equivalence between the restricted rank free-fermionic testing and the problem in Theorem~\ref{th:id_hard}, we need to show that:
\begin{itemize}
    \item \( \rho_A \otimes \ketbra{0^{n-r}}{0^{n-r}} \) is free-fermionic, and
    \item Any state \( \rho_B \otimes \ketbra{0^{n-r}}{0^{n-r}} \), as defined in Theorem~\ref{th:id_hard}, is far from any free-fermionic state \( \sigma \) on \( n \)-qubits.
\end{itemize}
The first condition is trivial. For the second condition, we proceed analogously to the proof of Theorem~\ref{th:hardnesstesting}. The only change is that, for any free-fermionic state \( \sigma \) on \( n \)-qubits, we have:
\be
    \|\rho_B \otimes \ketbra{0^{n-r}}{0^{n-r}} - \sigma\|_1 \geq \|\rho_B - \sigma_r\|_1,
\ee
where \( \sigma_r \) is the reduced state of \( \sigma \), and we used the standard fact that tracing out qubits reduces the trace distance (e.g., see Ref.\cite{MARK}). In particular, \( \sigma_r \) is still a free-fermionic state (since the reduced state of a free-fermionic state is also free-fermionic, see Lemma~\ref{le:reduced}).
\end{proof}

The same lower bound applies when considering the state $\rho$ to be an arbitrary state, while restricting the set $\mathcal{G}$ in Problem~\ref{prob:testing} to be the set of all free-fermionic states with rank at most $R$, denoted by $\mathcal{G}_R$.
\begin{theorem}[Lower bound for free-fermionic testing with respect to the set of bounded rank free-fermionic states $\mathcal{G}_R$] 
For an arbitrary quantum state $\rho$, if we consider $\mathcal{G}$ in Problem~\ref{prob:testing} to correspond to the set $\mathcal{G}_R$ of all free-fermionic states with rank at most $R:=2^r$, where $r \in [n]$, then to solve the free-fermionic testing property testing (Problem~\ref{prob:testing}) with at least a $2/3$ probability of success, $N=\Omega(R)$ copies are necessary. 
\end{theorem}
\begin{proof}
The proof follows the same lines as the previous theorem, and everything holds analogously, even with the assumption that the set to be considered is $\mathcal{G}_R$ instead of $\mathcal{G}_{\mathrm{mixed}}$. 
\end{proof}

\subsection{Efficient testing of low-rank free-fermionic states}
\label{subsec:mixedtesting}

In the preceding subsection, we established information-theoretic lower bounds for solving the free-fermionic property testing problem. Specifically, we demonstrated that when the input $n$-qubit state $\rho$ has a rank less than or equal to $2^r$, where $r \in [0,n]$ is an integer, a minimum of $\Omega(2^r)$ copies is required. Similarly, when no assumptions are made about $\rho$, but $\mathcal{G}$ in Problem~\ref{prob:testing} corresponds to the set $\mathcal{G}_R$ of all free-fermionic states with rank at most $R:=2^r$, then $N=\Omega(2^r)$ copies are necessary.

Now, we address the question of whether an algorithm can match these information-theoretic lower bounds. We notice that if $r=O(\log(n))$, i.e., $R=O(\mathrm{poly}(n))$, then such an algorithm would be efficient. The following theorem provides an affirmative answer to this question. Recall that the rank of a free-fermionic state is always $2^{r}$, for some integer $0 \le r \le n $.
\begin{theorem}[Upper bound for free-fermionic testing with respect to the set of bounded rank free-fermionic states $\mathcal{G}_R$]\label{th:mix1}
    Let $\rho$ be any $n$-qubit state. Assume error thresholds $\varepsilon_B, \varepsilon_A \in (0,1)$ such that 
    \be
    \varepsilon_B > \max ( \sqrt{2^5 (n-r)  \varepsilon_{A}}, 2(n+1)\varepsilon_A),
    \ee
    and consider a failure probability $\delta \in (0,1]$. Suppose $\rho$ falls into one of two cases in Problem~\ref{prob:testing}: either there exists a free-fermionic state $\sigma \in \mathcal{G}_R$ with $\| \rho - \sigma \|_{1} \le \varepsilon_A$ (Case A), or $\min_{\sigma\in\mathcal{G}_R}\| \rho - \sigma \|_{1} > \varepsilon_B$ (Case B), where $R:=2^r$ with $r\in [0,n]$ being an integer.

    Then, a quantum learning algorithm (Algorithm~\ref{alg:algoTESTbound}) can solve Problem~\ref{prob:testing} using 
   \be
    N:=\lceil 8(n^3/\epsstat^2) \log(8n^2/\delta) + N_{\mathrm{tom}}(\varepsilon_{\mathrm{tom}},\delta/2, r) \rceil
    \ee
    copies of the state $\rho$ with a success probability at least $1-\delta$. Here, $ N_{\mathrm{tom}}(\varepsilon_{\mathrm{tom}},\delta/2, r)$ is the number of copies sufficient for a full state tomography algorithm (see, e.g., Lemma~\ref{besttomalg} for the performance guarantees of a well-known tomography algorithm) of an $r$-qubit state with accuracy $\varepsilon_{\mathrm{tom}}$ and failure probability at most $\delta/2$. Here, $    \varepsilon_{\mathrm{stat}}< \frac{1}{2}\left(\frac{\varepsilon^2_B}{2^5(n-r)} - \varepsilon_A\right) $ and $\varepsilon_{\mathrm{tom}}  < \frac{1}{n+2}\left( \frac{\varepsilon_B}{2} - (n+1)\varepsilon_A\right)$.
\end{theorem}

\begin{algorithm}
\caption{Property testing algorithm for bounded rank free-fermionic states}
\label{alg:algoTESTbound}
\KwIn{Let $R=2^r$, with $r \in [0,n]$ being an integer. Error thresholds $\varepsilon_A, \varepsilon_B$ such that $  \varepsilon_B > \max ( \sqrt{2^5 (n-r)  \varepsilon_{A}}, 2(n+1)\varepsilon_A)$, failure probability $\delta$. $N:=\lceil 8(n^3/\varepsilon_{\mathrm{stat}}^{2}) \log(8n^2/\delta) + N_{\mathrm{tom}}(\varepsilon_{\mathrm{tom}},\delta/2, r) \rceil $ copies of $\rho$. Here, $ N_{\mathrm{tom}}(\varepsilon_{\mathrm{tom}},\delta/2, r)$ is the number of copies sufficient for full state tomography of an $r$-qubit state with accuracy $\varepsilon_{\mathrm{tom}}$ and failure probability at most $\delta/2$. Let $\varepsilon_{\mathrm{stat}}< \frac{1}{2}\left(\frac{\varepsilon^2_B}{2^5(n-r)} - \varepsilon_A\right) $ and $\varepsilon_{\mathrm{tom}} < \frac{1}{n+2}\left( \frac{\varepsilon_B}{2} - (n+1)\varepsilon_A\right)$.}

\KwOut{Output either Case A or Case B.}

\textbf{Step 1:} Estimate the entries of the correlation matrix of $\rho$ using $\lceil 8(n^3/\varepsilon_{\mathrm{stat}}^{2}) \log(8n^2/\delta)\rceil$ single-copies measurements, resulting in the matrix $\hat{\Gamma}$ \;

\textbf{Step 2:} Find $\hat{\lambda}_{r+1}$, the $(r+1)$-th smallest normal eigenvalue of $\hat{\Gamma}$ \;

\textbf{Step 3:} 
\If{$\hat{\lambda}_{r+1} \le 1 - \epsT$, where $\epsT\coloneqq \frac{\varepsilon^2_B}{2^6(n-r)}  +\frac{1}{2}\varepsilon_A$}{
    \textbf{Output:} Case B
}
\Else{
    \textbf{Step 4:} Evolve $\rho$ with the free-fermionic unitary $U_{\hat{O}}$, where $\hat{O}$ is the orthogonal matrix that puts $\hat{\Gamma}$ in its normal form\;
    
    \textbf{Step 5:} Full state tomography on the first $r$ qubits of $U_{\hat{O}} \rho U^{\dagger}_{\hat{O}}$, which returns the state $\hat{\rho}_{r}^{\prime}$ with correlation matrix $\Gamma_r$\;
    
    \textbf{Output:} Case B if $\|\hat{\rho}_{r}^{\prime} - \sigma(\hat{\Gamma}_r)\|_1 > \varepsilon_{\mathrm{T},2}$, where $ \varepsilon_{\mathrm{T},2}\coloneqq \frac{n+1}{n+2}(\frac{\varepsilon_B}{2} + \varepsilon_A )$, else Case A. Here, $\sigma(\hat{\Gamma}_r)$ is the free-fermionic state associated with the correlation matrix $\hat{\Gamma}_r$ of $\hat{\rho}_{r}^{\prime}$.  
}
\end{algorithm}

\begin{proof}
Let $\varepsilon_{\mathrm{stat}}$ be an accuracy parameter, yet to be determined. According to Lemma~\ref{le:corrMatrixEst}, with $N \ge 8(n^3/\varepsilon_{\mathrm{stat}}^{2})\log(8n^2/\delta)$ single-copies measurements, we can construct a matrix $\hat{\Gamma}$ such that, with a probability of at least $1-\delta/2$, it satisfies $\|\hat{\Gamma} - \Gamma(\rho)\|_{\infty} < \varepsilon_{\mathrm{stat}}$. Consequently, for each sorted normal eigenvalue $\lambda_{k}$ of the correlation matrix, it holds that $|\hat{\lambda}_k - \lambda_k| < \varepsilon_{\mathrm{stat}}$, where $\hat{\lambda}_k$ is the $k$-th eigenvalue of $\hat{\Gamma}$.
Now, consider that if $\rho$ were a free-fermionic state with rank bounded by $R$, the first $r$ normal eigenvalues could potentially be less than one, while the remaining $n-r$ should be one (as per Remark~\ref{le:rankEigs}). We proceed to perform our first check. If $\hat{{\lambda}}_{r+1} \ge 1- \varepsilon_{\text{T}}$, we continue; otherwise, we output case $B$, where $\varepsilon_{\text{T}}>0$ is an error threshold to be fixed later.
In case we output case $B$, let us demonstrate that we cannot be in case $A$. We need to show that $\min_{\sigma \in \mathcal{G}_R} \|\rho - \sigma\|_{1} > \varepsilon_{A}$. We have
\begin{align}
\label{eq:1stproofRANK}
    \min_{\sigma \in \mathcal{G}_R}\norm{\rho - \sigma}_{1} \ge  1 - \lambda_{r+1} \ge 1 - \hat{\lambda}_{r+1} - \varepsilon_{\mathrm{stat}} > \varepsilon_{\text{T}} - \varepsilon_{\mathrm{stat}},
\end{align}
where the first inequality follows from Theorem~\ref{th:lowerboundnon-Gaussianity0}.
Therefore, by choosing 
\begin{align}
    \varepsilon_{\text{T}} - \varepsilon_{\mathrm{stat}} > \varepsilon_{A},
    \label{eq:constr1}
\end{align}
we successfully ensure that $\rho$ cannot be in case $A$. This condition forms the first criterion for determining the accuracy parameter $\varepsilon_{\mathrm{stat}}$ and the threshold $\varepsilon_{\text{T}}$.
To proceed, we define $\hat{O}$ as the orthogonal matrix such that $\hat{O}\hat{\Gamma}\hat{O}^{T}=\hat{\Lambda}$, where \be \hat{\Lambda}:= \bigoplus_{k=1}^{n}\begin{pmatrix}
    0&\hat{\lambda}_k\\
    -\hat{\lambda}_k&0
\end{pmatrix}\ee  and define $U_{\hat{O}}$ as the associated free-fermionic unitary. Consider the state $\rho^{\prime}:=U_{\hat{O}}\rho U_{\hat{O}}^{\dagger}$. We observe that
\begin{align}
    \lvert \Gamma(\rho^{\prime})_{j,k} - (\hat{\Lambda} )_{j,k} \rvert \le \|\Gamma(\rho^{\prime})-\hat{\Lambda}\|_{\infty} \le \|\Gamma(\rho)-\hat{\Gamma}\|_{\infty}\le\varepsilon_{\mathrm{stat}},
\end{align}
leveraging the relationships $\Gamma(\rho^{\prime})=\hat{O} \Gamma(\rho)\hat{O}^T $ and $\hat{\Lambda}=\hat{O} \hat{\Gamma}\hat{O}^T$, Cauchy-Schwartz, and the definition of the infinity norm. Consequently, we establish $\Gamma(\rho^{\prime})_{j,k} \ge (\hat{\Lambda} )_{j,k} - \varepsilon_{\mathrm{stat}}$. Specifically, for $k\ge r +1$, we find
\begin{align}
    \Tr(Z_k \rho^{\prime}) &=\Gamma(\rho^{\prime})_{2k-1,k} \ge (\hat{\Lambda} )_{2k-1,2k} - \varepsilon_{\mathrm{stat}} = \hat{\lambda}_{k} - \varepsilon_{\mathrm{stat}}  \ge 1 - \varepsilon_{\text{T}} - \varepsilon_{\mathrm{stat}},
\end{align}
where $Z_k=-i\gamma_{2k-1}\gamma_{2k}$ represents the $Z$-Pauli operator acting on the $k$-th qubit. Consequently, we also find $\Tr(\ketbra{0}{0}_k \rho^{\prime})\ge 1 - (\varepsilon_{\text{T}} + \varepsilon_{\mathrm{stat}})/2$. Employing Lemma~\ref{le:gentle}, we derive
\begin{align}
    \norm{\rho^{\prime} - \phi\otimes \ketbra{0^{n-r}}{0^{n- r}}}_{1}&\le 2 \sqrt{(n-r)  (\varepsilon_{\text{T}} + \varepsilon_{\mathrm{stat}})/2},
    \label{eq:Qunbound}
\end{align}
where $\phi\otimes \ketbra{0^{n-r}}{0^{n- r}}$ represents the post-measurement state obtained after measuring the outcomes corresponding to $\ket{0^{n-r}}$ in the last $n-r$ qubits.
Define the subsystem $E$ with sites $E=\{r+1,\ldots, n\}$ and define
\begin{equation}    \rho^{\prime}_{r}:=\tr_E[\rho^{\prime}]=\tr_E[U_{\hat{O}}\rho U^{\dagger}_{\hat{O}}].
\end{equation}
We also have
\begin{align}
    \norm{\rho^{\prime} - \rho^{\prime}_{r}\otimes \ketbra{0^{n-r}}{0^{n- r}}}_1&\le \norm{\rho^{\prime} -\phi\otimes \ketbra{0^{n-r}}{0^{n- r}}}_{1}+\norm{\phi - \rho^{\prime}_{r} }_{1}\\
    \nonumber 
    &\le \norm{\rho^{\prime} -\phi\otimes \ketbra{0^{n-r}}{0^{n- r}}}_{1}+\norm{\phi\otimes \ketbra{0^{n-r}}{0^{n- r}} - \rho^{\prime}}_{1}    
    \\&\le 4\sqrt{(n-r)  (\varepsilon_{\text{T}} + \varepsilon_{\mathrm{stat}})/2},
    \nonumber 
\end{align}
where in the first step, we have used the triangle inequality, in the second step, the data-processing inequality ($\|\tr_{E}(\rho-\sigma)\|_1\le \|\rho-\sigma\|_1$ for any quantum states $\rho,\sigma$), and in the last step, we have used Eq.\ \eqref{eq:Qunbound}.
We now perform full-state tomography on the first $r$ qubits of $\rho^{\prime}$. More precisely, using copies of $\rho_r^{\prime}$, we can output a state $\hat{\rho}^{\prime}_r$ such that, with a probability of at least $1-\delta/2$, we find that
\begin{align}
    \norm{\hat{\rho}_{r}^{\prime}-\rho^{\prime}_{r}}_{1}\le \varepsilon_{\mathrm{tom}}.
\end{align}
There are various algorithms for full-state tomography that utilize single-copy measurements (see, e.g., Ref.~\cite{FastFranca}), all having sample complexity that scales exponentially with the number of qubits constituting the quantum state, in our case, $r$.
Furthermore, through the computation of the correlation matrix of $\hat{\rho}_r^{\prime}$, we can compute its correlation matrix $\hat{\Gamma}_r$, which satisfies
\begin{align}
\label{eq:normtom}
    \norm{\hat{\Gamma}_r-\Gamma(\rho_r^{\prime})}_{\infty}\le \norm{\hat{\rho}_r^{\prime}-\rho_r^{\prime}}_1\le  \varepsilon_{\mathrm{tom}},
\end{align}
where we have invoked Proposition~\ref{le:tracedistancelowerboundcormatrix}.
Now, let us consider the free-fermionic state $\sigma(\hat{\Gamma}_r)$ associated with the correlation matrix $\hat{\Gamma}_r$.
Our second discrimination test hinges on the quantity $\norm{\hat{\rho}_r^{\prime} - \sigma(\hat{\Gamma}_r)}_1$, which can be computed with a time complexity scaling as $O(\exp(r))$, which is efficient as long as $r=O(\log(n))$.
If $\norm{\hat{\rho}_r^{\prime} - \sigma(\hat{\Gamma}_r)}_1 \le \varepsilon_{\mathrm{T},2}$, we output A; otherwise, we output B.
In the case of outputting A, our goal is to demonstrate that we cannot be in case B. Specifically, we show that there exist a free-fermionic state closer, in trace distance, than $\varepsilon_B$ to $\rho$.
Consider the free-fermionic state $U_{\hat{O}}^{\dagger}\left(\sigma(\hat{\Gamma}_r)\otimes \ketbra{0^{n-r}}{0^{n- r}}\right)U_{\hat{O}}$, which is readily free-fermionic. We have
\begin{align}
    &\norm{\rho - U_{\hat{O}}^{\dagger}\left(\sigma(\hat{\Gamma}_r)\otimes \ketbra{0^{n-r}}{0^{n- r}} \right)U_{\hat{O}}}_{1}  \\
    \nonumber 
    &= \norm{\rho^{\prime} - \sigma(\hat{\Gamma}_r)\otimes \ketbra{0^{n-r}}{0^{n- r}}}_{1} \nonumber \\
    \nonumber 
    &\le \norm{\rho^{\prime} - \rho_r^{\prime}\otimes \ketbra{0^{n-r}}{0^{n- r}}}_1 + \norm{\rho_r^{\prime} - \sigma(\hat{\Gamma}_r)}_{1} \nonumber \\
    \nonumber 
    &\le \norm{\rho^{\prime} - \rho_r^{\prime}\otimes \ketbra{0^{n-r}}{0^{n- r}}}_1 + \norm{\rho_r^{\prime} - \hat{\rho}_r^{\prime}}_1 + \norm{\hat{\rho}_r^{\prime} - \sigma(\hat{\Gamma}_r)}_{1} \nonumber \\
    \nonumber 
    &\le 4\sqrt{(n-r)  (\varepsilon_{\text{T}} + \varepsilon_{\mathrm{stat}})/2} + \varepsilon_{\mathrm{tom}} + \varepsilon_{\mathrm{T},2},
    \nonumber 
\end{align}
where in the first inequality, we have used the unitary invariance of the trace norm; in the second and third steps, we applied the triangle inequality; and in the last step, we have used the previously derived bound.
Now, we must ensure
\begin{align}
    4\sqrt{(n-r)  (\varepsilon_{\text{T}} + \varepsilon_{\mathrm{stat}})/2} + \varepsilon_{\mathrm{tom}} + \varepsilon_{\mathrm{T},2} \le \varepsilon_B.
    \label{eq:constr2}
\end{align}
Now, let us explore the scenario where we output case B. We need to show that, if $\|\hat{\rho}_r^{\prime}-\sigma(\hat{\Gamma}_r)\|_1>\varepsilon_{\mathrm{T},2}$, then we cannot be in case A, i.e., that $\min_{\sigma\in\mathcal{G}_R}\|\rho-\sigma\|_1 > \varepsilon_A$. 
Considering that the application of a free-fermionic unitary $U_{\hat{O}}$ and the partial trace map a free-fermionic state into another free-fermionic state (Lemma~\ref{le:reduced}), we employ the data processing inequality, leading to the inequality
\begin{equation}
\label{eq:ineq1}
\min_{\sigma\in\mathcal{G}_R}\|\rho-\sigma\|_1=\min_{\sigma\in\mathcal{G}_R}\|\rho^{\prime}-\sigma\|\ge \min_{\sigma_r\in\mathcal{G}_R}\|\rho^{\prime}_r-\sigma_r\|.
\end{equation}
Let us express the lower bounds
\begin{equation}
\label{eq:ineq2}
\|\rho^{\prime}_r-\sigma_r\|_1\ge \begin{cases}
\norm{\hat{\Gamma}_r-\Gamma(\sigma_r)}_{\infty}-\varepsilon_{\mathrm{tom}}\\
\norm{\hat{\rho}^{\prime}_r-\sigma(\hat{\Gamma}_r)}_{1}-\varepsilon_{\mathrm{tom}}-n\norm{\hat{\Gamma}_r-\Gamma(\sigma_r)}_{\infty}
\end{cases},
\end{equation}
where, for the first bound, we utilize Proposition~\ref{le:tracedistancelowerboundcormatrix} and the
triangle inequality together with Eq.\ \eqref{eq:normtom}, namely
\begin{equation}
\|\rho_r^{\prime}-\sigma_r\|_1\ge \norm{\Gamma(\rho^{\prime}_r)-\Gamma(\sigma_r)}_{\infty}\ge \norm{\hat{\Gamma}_r-\Gamma(\sigma_r)}_{\infty}-\varepsilon_{\mathrm{tom}},
\end{equation}
 
and for the second bound, by utilizing Theorem~\ref{th:mixedtracedistance} and triangle inequality, we have
\begin{equation}
\begin{aligned}
\|\rho^{\prime}_r-\sigma_r\|_{1}&\ge \|\rho_r^{\prime}-\sigma(\hat{\Gamma}_r)\|_{1}-\|\sigma(\hat{\Gamma}_r)-\sigma_r\|_{1}\\
&\ge\|\rho^{\prime}_r-\sigma(\hat{\Gamma}_r)\|_{1}-\frac{1}{2}\|\hat{\Gamma}_r-\Gamma(\sigma_r)\|_{1}\\
&\ge\|\hat{\rho}^{\prime}_r-\sigma(\hat{\Gamma}_r)\|_{1}-\varepsilon_{\mathrm{tom}}-\frac{1}{2}\|\hat{\Gamma}_r-\Gamma(\sigma_r)\|_{1}\\
&\ge\|\hat{\rho}^{\prime}_r-\sigma(\hat{\Gamma}_r)\|_{1}-\varepsilon_{\mathrm{tom}}-n\|\hat{\Gamma}_r-\Gamma(\sigma_r)\|_{\infty}
\end{aligned}
\end{equation}
Using standard norm inequalities, the result follows. Since we do not know the optimal $\sigma_r$, we use the both lower bounds to find a universal lower bound. To achieve this, we solve the equation
\begin{equation}
\begin{aligned}
\|\hat{\rho}^{\prime}_r-\sigma(\hat{\Gamma}_r)\|_{1}-\varepsilon_{\mathrm{tom}}-n\|\hat{\Gamma}_r-\Gamma(\sigma_r)\|_{\infty}&=\norm{\hat{\Gamma}_r-\Gamma(\sigma_r)}_{\infty}-\varepsilon_{\mathrm{tom}},\\
\norm{\hat{\Gamma}_r-\Gamma(\sigma_r)}_{\infty}&=\frac{1}{n+1}\norm{\hat{\rho}^{\prime}_r-\sigma(\hat{\Gamma}_r)}_{1}
\end{aligned}
\end{equation}
Substituting the solution in the previous inequality (Eq.\ \eqref{eq:ineq1},\eqref{eq:ineq2}), we obtain
\begin{equation}
\begin{aligned}
\min_{\sigma_r\in\mathcal{G}_R}\|\rho^{\prime}_r-\sigma_r\|_{1}&\ge\frac{1}{n+1}\norm{\hat{\rho}^{\prime}_r-\sigma(\hat{\Gamma}_r)}_{1}-\varepsilon_{\mathrm{tom}}\,.
\end{aligned}\end{equation}
Now, we only need to impose that, given $\norm{\hat{\rho}^{\prime}_r-\sigma(\hat{\Gamma}_r)}_{1}>\varepsilon_{\mathrm{T},2}$, we cannot be in case $A$ and, therefore, impose
\begin{equation}
\frac{1}{n+1}\varepsilon_{T,2}-\varepsilon_{\mathrm{tom}} 
>\varepsilon_A.
\label{eq:constr3}
\end{equation}
To satisfy the constraints in Eq.~\eqref{eq:constr1}, \eqref{eq:constr2}, and \eqref{eq:constr3}, we need to choose the constants $\epsstat$, $\epsT$, $\varepsilon_{\mathrm{tom}}$, and $\varepsilon_{\mathrm{T},2}$.
We start by imposing the two inequalities (which implies Eq.~\eqref{eq:constr2})
\begin{align}
\label{eq:eqconst2A}
    4\sqrt{(n-r)  (\varepsilon_{\text{T}} + \varepsilon_{\mathrm{stat}})/2} &\le \frac{\varepsilon_B}{2},\\
    \varepsilon_{\mathrm{tom}} + \varepsilon_{\mathrm{T},2} &\le \frac{\varepsilon_B}{2}.
\label{eq:eqconst2B}
\end{align}
Therefore, we have ``disentangled'' the three inequalities in Eq.~\eqref{eq:constr1}, \eqref{eq:constr2}, and \eqref{eq:constr3} into two systems of two inequalities, the first containing $\varepsilon_\mathrm{stat}$ and $\varepsilon_{\text{T}}$ involving Eq.~\eqref{eq:constr1} and \eqref{eq:eqconst2A}, and the other one containing $\varepsilon_{\mathrm{tom}}$ and $\varepsilon_{\mathrm{T},2}$ involving Eq.~\eqref{eq:constr3} and \eqref{eq:eqconst2B}.
By considering the one involving Eq.~\eqref{eq:constr1} and \eqref{eq:eqconst2A}, we get that it suffices to choose
\begin{align}
    \epsT&= \frac{\varepsilon^2_B}{2^6(n-r)}  +\frac{1}{2}\varepsilon_A,\\
    \varepsilon_{\mathrm{stat}}&< \frac{1}{2}\left(\frac{\varepsilon^2_B}{2^5(n-r)} - \varepsilon_A\right)  .
\end{align}
Moreover, by combining the one involving Eq.~\eqref{eq:constr3} and \eqref{eq:eqconst2B}, we get that it is sufficient to choose:
\begin{align}
\varepsilon_{\mathrm{T},2} &= \frac{n+1}{n+2}\left(\frac{\varepsilon_B}{2} + \varepsilon_A \right),\\
\varepsilon_{\mathrm{tom}}  &< \frac{1}{n+2}\left( \frac{\varepsilon_B}{2} - (n+1)\varepsilon_A\right)
\end{align}
By union bound, the total failure probability of the protocol is at most $1-\delta$.
\end{proof}

The previous theorem has been presented under the assumption that $\rho$ is an arbitrary $n$-qubit state, and the set of free-fermionic states considered is $\mathcal{G}_{R}$, i.e., the set of free-fermionic states with rank at most $R=2^r$, where $r\in [n]$. However, if we assume that $\rho$ has at most rank $R$, then we can consider the largest set $\mathcal{G}_{\mathrm{mixed}}$, and we can establish an analogous result. The theorem is detailed as follows, and the algorithm is the same as Algorithm \ref{alg:algoTESTbound} with slightly different accuracy parameters, as detailed below.

\begin{theorem}[Upper bound for free-fermionic testing for a bounded rank quantum state]\label{th:mix2}
    Let $\rho$ be any $n$-qubit state with rank at most $2^r$, where $r\in[n]$. Assume error thresholds $\varepsilon_B, \varepsilon_A \in (0,1)$ such that
\begin{align}
   \varepsilon_B > \max(\sqrt{2^5(n-r)  (2 \varepsilon_A)^{1/(r+1)}}, 2(n+1)\varepsilon_A)    
\end{align} and consider a failure probability $\delta \in (0,1]$. Suppose $\rho$ falls into one of two cases in Problem~\ref{prob:testing}: either there exists a free-fermionic state $\sigma \in \mathcal{G}_{\mathrm{mixed}}$ with $\| \rho - \sigma \|_{1} \le \varepsilon_A$ (Case A), or $\min_{\sigma\in\mathcal{G}_{\mathrm{mixed}}}\| \rho - \sigma \|_{1} > \varepsilon_B$ (Case B).
    Then, a quantum learning algorithm (Algorithm~\ref{alg:algoTESTbound}) can solve Problem~\ref{prob:testing} using 
    \[
    N:=\lceil 8(n^3/\epsstat^2) \log(8n^2/\delta) + N_{\mathrm{tom}}(\varepsilon_{\mathrm{tom}},\delta/2, r) \rceil
    \]
    copies of the state $\rho$ with a success probability at least $1-\delta$. Here, $N_{\mathrm{tom}}(\varepsilon_{\mathrm{tom}},\delta/2, r)$ is the number of copies sufficient for a full state tomography algorithm (see, e.g., Lemma~\ref{besttomalg}) of an $r$-qubit state with accuracy $\varepsilon_{\mathrm{tom}}$ and failure probability at most $\delta/2$. Here, we impose
    \begin{align}
    \epsT&>( \epsstat +(2\varepsilon_A)^{\frac{1}{r+1}}),\\
    \varepsilon_{\mathrm{stat}}&< \frac{1}{2}\left(\frac{\varepsilon^2_B}{2^5(n-r)} - (2\varepsilon_A)^{\frac{1}{r+1}} \right),\\
    \varepsilon_{\mathrm{T},2} &= \frac{n+1}{n+2}\left(\frac{\varepsilon_B}{2} + \varepsilon_A \right),\\
\varepsilon_{\mathrm{tom}}  &< \frac{1}{n+2}\left( \frac{\varepsilon_B}{2} - (n+1)\varepsilon_A\right)
\end{align}
\end{theorem}
\begin{proof}
    The proof is the same as the one of the previous theorem, but this time we
    have utilized instead of Eq.\ \eqref{eq:1stproofRANK}
    the expression
    \begin{align}
    \min_{\sigma\in\mathcal{G}_{\mathrm{mixed}}}\norm{\rho-\sigma}_{1}\ge \frac{1}{2}(1-\lambda_{r+1})^{r+1},
    \end{align}
    which follows from Lemma~\ref{le:lbrank}. Note that Lemma~\ref{le:lbrank} can be applied if $\lambda_{r+1} \geq \frac{1}{2}$; however, if it was $\lambda_{r+1} < \frac{1}{2}$, then we would have output that the state is not free-fermionic because $\hat{\lambda}_{r+1}\le \lambda_{r+1} + \varepsilon_{\mathrm{stat}}< \frac{1}{2} + \varepsilon_{\mathrm{stat}}\le 1-\varepsilon_\mathrm{T}$ (this is satisfied if $\epsstat+\epsT < 0.5$, which is the regime of interest).
    From this, it follows that (using the same notation as in the previous proof)
    \begin{align}
        \min_{\sigma\in\mathcal{G}_{\mathrm{mixed}}}\norm{\rho-\sigma}_{1} > \frac{1}{2}(\epsstat-\epsT)^{r+1}.
    \end{align}
    Hence, we have the condition
    \begin{align}
       (\epsT-\epsstat)^{r+1}> 2\varepsilon_A.
    \end{align}
    This is the only condition that is different from the ones in the previous Theorem.
We impose
 \begin{align}
    (2\varepsilon_A)^{\frac{1}{r+1}},\\
    \varepsilon_{\mathrm{stat}}&< \frac{1}{2}\left(\frac{\varepsilon^2_B}{2^5(n-r)} - (2\varepsilon_A)^{\frac{1}{r+1}} \right),\\
     \varepsilon_{\mathrm{tom}} &< \frac{1}{n+2}(\frac{\varepsilon_B}{2}- (n+1)\varepsilon_A)
\end{align}
and this suffices to satisfy all the constraints.
\end{proof}

We now present a lemma that provides, as an example, the performance guarantee for a known tomography algorithm~\cite{wrightHowLearnQuantum}. Specifically, we provide the precise performance guarantees of an optimal tomography algorithm presented in Refs.~\cite{wrightHowLearnQuantum, odonnell2015efficient} (see also Lemma S34 in Ref.\  \cite{Mele2024bosonic}).

\begin{lemma}[Sample complexity of an optimal tomography algorithm~\cite{wrightHowLearnQuantum}]\label{besttomalg}
    Let $\varepsilon, \delta \in (0,1)$, and let $\rho$ be an unknown state of dimension $d$ (where $d = 2^n$ for $n$-qubit systems), such that $\rho$ is $\frac{\varepsilon}{3}$-close in trace distance to a state with rank $R$. Then, there exists a tomography algorithm such that, given 
    \begin{align}
        N \ge 2^{18} \frac{Rd}{\varepsilon^2} \log\left(\frac{2}{\delta}\right)
    \end{align}
    copies of $\rho$, it can construct (a classical description of) an $R$-rank state estimator $\tilde{\rho}$ satisfying
    \begin{align}
        \mathrm{Prob}\left[\frac{1}{2} \|\rho - \tilde{\rho}\|_1 \le \varepsilon \right] \geq 1 - \delta\,.
    \end{align}
\end{lemma}

While this algorithm achieves optimal performance using entangled (and possibly highly complex) measurements across multiple copies of the unknown state, it might be impractical for experimental implementations. However, more experimentally feasible algorithms exist that use \emph{unentangled} or \emph{single-copy} measurements on each queried copy of the unknown state~\cite{guta2018fast, FastFranca}. The best known algorithm of this kind requires $\tilde{O}(R^2d/\varepsilon^2)$ copies~\cite{guta2018fast}.

\section{Optimal tomography of pure and mixed free-fermionic states}\label{Sec:tomography}

In this section, we present our results on tomography of free-fermionic states. First, we provide an improved analysis for learning pure free-fermionic states (Subsection~\ref{sec:puretom}). Next, we introduce the first time-efficient algorithm for learning mixed free-fermionic states with respect to the trace distance (Subsection~\ref{sec:mixedtom}). Finally, we analyze the tomography of states that are promised to be close to the set of free-fermionic states, exploring the applicability of our free-fermionic learning algorithm to noisy scenarios where the prepared states are not exactly free-fermionic but \emph{almost}.

\subsection{Pure state tomography}
\label{sec:puretom}
In this subsection, we present an improved analysis for learning pure free-fermionic states, which directly follows from Lemma~\ref{le:samplecompAPPcommuting} and our novel inequality for pure free-fermionic states (Theorem~\ref{thm:gausspp}). 
Our analysis provides an improvement over the sample complexity bounds for learning pure free-fermionic states found in previous works~\cite{Gluza_2018, aaronson2023efficient, ogorman2022fermionic}.
\begin{proposition}[Tomography of pure free-fermionic states]
Let $\psi$ be a free-fermionic quantum state. For $\varepsilon\in(0,1)$ and $\delta\in(0,1]$ there exist a learning algorithm (outlined in table~\ref{alg:tomPURE}) that utilizes $N=8(n^3 / \varepsilon_{\mathrm{stat}}^{2}) \log(4n^2 / \delta)$ copies of the state and only single-copies measurements to learn an efficient classical representation of the state $\hat{\psi}$ obeying $\|\psi-\hat{\psi}\|_1\le \varepsilon$ with at least $1 - \delta$ probability.
\end{proposition}
\begin{algorithm}
\caption{Learning pure free-fermionic states}
\label{alg:tomPURE}

\KwIn{ 
    Error threshold $\varepsilon>0$, failure probability $\delta>0$. $N=\lceil 8 (n^3 /\varepsilon^{2}) \log(4n^2/\delta)\rceil$ copies of the pure free-fermionic state $\psi$.
}

\KwOut{
    A classical description of a state $\hat{\psi}$, such that $\norm{\hat{\psi}-\psi}_1\le \varepsilon$ with at least $1-\delta$ success probability.
}

\textbf{Step 1:} Estimate the entries of the correlation matrix of $\psi$ using $N$ single-copy measurements, resulting in the estimated $2n \times 2n$ matrix $\hat{\Gamma}$\;

\textbf{Step 2:} Put $\hat{\Gamma}$ in its normal form $\hat{\Gamma}=\hat{O}\hat{\Lambda}\hat{O}^{T}$, where $\hat{O} \in \mathrm{O}(2n)$, and $\hat{\Lambda}$ is the matrix determined by the normal eigenvalues $\{\hat{\lambda}_j\}^n_{j=1}$ \;

\Return $\hat{O}$, so that
$\hat{\psi}:= G_{\hat{O}}\ketbra{0^n}{0^n}G_{\hat{O}}^{\dagger}$, where $G_{\hat{O}}$ is the free-fermionic unitary associated with $\hat{O}$.
\end{algorithm}
\begin{proof}
Let $\varepsilon_{\mathrm{stat}} > 0$ be an accuracy parameter. By Lemma~\ref{le:samplecompAPPcommuting}, with $N \ge 8(n^3 / \varepsilon_{\mathrm{stat}}^{2}) \log(4n^2 / \delta)$ copies of the state, we can find a matrix $\hat{\Gamma}$ such that, with probability at least $1 - \delta$, it holds that
\begin{align}
\label{eq:epsnumb}
    \|\hat{\Gamma} - \Gamma(\psi)\|_2 \le \varepsilon_{\mathrm{stat}}.
\end{align}
By putting $\hat{\Gamma}$ in its normal form, we have $\hat{\Gamma} = \hat{O} \hat{\Lambda} \hat{O}^{\dagger}$, where $\hat{O}$ is an orthogonal matrix and $\hat{\Lambda} = \bigoplus_{i=1}^n \begin{pmatrix}
    0 & \hat{\lambda}_i \\ 
    -\hat{\lambda}_i & 0 
\end{pmatrix}$.
We now consider the correlation matrix $\hat{\Gamma}^{\prime} \coloneqq \hat{O} \Lambda \hat{O}^{\dagger}$, where $\Lambda = \bigoplus_{i=1}^n \begin{pmatrix}
    0 & 1 \\ 
    -1 & 0 
\end{pmatrix}$. The correlation matrix $\hat{\Gamma}^{\prime}$ is associated with the pure free-fermionic state $\hat{\psi} = U_{\hat{O}} \ketbra{0^n} U^{\dagger}_{\hat{O}}$, where $U_{\hat{O}}$ is the free-fermionic unitary associated with $\hat{O}$ (because of Lemma~\ref{le:bijection}). We now have
\begin{align}
    \|\hat{\psi}^{\prime} - \psi\|_{1} &\le \frac{1}{2} \|\hat{\Gamma}^{\prime} - \Gamma(\psi)\|_2 \\
    \nonumber
    &\le \frac{1}{2} \|\hat{\Gamma}^{\prime} - \hat{\Gamma}\|_2 + \frac{1}{2} \|\hat{\Gamma} - \Gamma(\psi)\|_2 \\
    \nonumber
    &\le \frac{1}{2} \|\Lambda - \hat{\Lambda}\|_2 + \frac{1}{2} \varepsilon_{\mathrm{stat}} \\
    \nonumber
    &\le \frac{1}{2} \|\Gamma(\psi) - \hat{\Gamma}\|_2 + \frac{1}{2} \varepsilon_{\mathrm{stat}} \\
    \nonumber
    &\le \varepsilon_{\mathrm{stat}},
    \nonumber
\end{align}
where in the first step we have used Theorem~\ref{thm:gausspp}, in the second step we have used the triangle inequality, in the third step we have used Eq.~\eqref{eq:epsnumb}, and in the fourth step we have used that
\begin{align}
  \|\Lambda - \hat{\Lambda}\|_2 = \left\|\bigoplus_{j=1}^n \begin{pmatrix}
         \lambda_j & 0 \\ 
         0 & -\lambda_j 
    \end{pmatrix} - \bigoplus_{i=1}^n \begin{pmatrix}
        1 & 0 \\ 
        0 & -1 
    \end{pmatrix}\right\|_2  \le \|\hat{\Gamma} - \Gamma(\psi)\|_2,
\end{align}
where we have used the fact that $i\Gamma$ and $i\hat{\Gamma}$ are Hermitian and that $\|\mathrm{Eig}(A) - \mathrm{Eig}(B)\|_{p} \le \|A - B\|_p$ for any Hermitian matrices $A$ and $B$ and any p-norms $\|\cdot\|_p$ (see Ref.~\cite{bhatia1996matrix}, Eq.~(IV.62)), where $\mathrm{Eig}(A)$ indicates the diagonal matrix with elements being the ordered eigenvalues of $A$. 
\end{proof}
\subsection{Mixed state tomography}
In this subsection, we present an efficient algorithm for learning mixed free-fermionic states. This result follows directly from Theorem~\ref{th:mixedtracedistance}, which establishes a relationship between the trace distance of mixed free-fermionic states and the one-norm difference of their correlation matrices.

While previous works have provided sample complexity bounds for learning pure free-fermionic states~\cite{Gluza_2018, aaronson2023efficient, ogorman2022fermionic}, to the best of our knowledge, this is the first work to rigorously demonstrate an efficient method for learning mixed free-fermionic states with respect to the trace distance.

\label{sec:mixedtom}
\begin{algorithm}
\caption{Learning mixed free-fermionic states}
\label{alg:tomMIXED}

\KwIn{ 
    Error threshold $\varepsilon>0$, failure probability $\delta>0$. $N=\lceil 16 (n^{4} /\varepsilon^{2}) \log(4n^2/\delta)\rceil$ copies of the mixed free-fermionic state $\rho$.
}

\KwOut{
    A classical description of a state $\hat{\rho}$, such that $\norm{\hat{\rho}-\rho}_1\le \varepsilon$ with at least $1-\delta$ success probability.
}

\textbf{Step 1:} Estimate the entries of the correlation matrix of $\rho$ using $N$ single-copy measurements, resulting in the estimated $2n \times 2n$ matrix $\hat{\Gamma}$\;

\textbf{Step 2:} Put $\hat{\Gamma}$ in its normal form $\hat{\Gamma}=\hat{O}\hat{\Lambda}\hat{O}^{T}$, where $\hat{O} \in \mathrm{O}(2n)$, and $\hat{\Lambda}$ is the matrix determined by the normal eigenvalues $\{\hat{\lambda}_j\}^n_{j=1}$ \;

\textbf{Step 3:}  For each $j\in[n]$, set $
        \tilde{\lambda}_j \coloneqq 
        \begin{cases}
            1, & \text{if } \hat{\lambda}_j > 1, \\
            \hat{\lambda}_j, & \text{otherwise}.
        \end{cases}
$\;

\Return $\hat{O}$ and $\{\tilde{\lambda}_i\}^n_{i=1}$, so that
$\hat{\rho}:= G_{\hat{O}}\left( \bigotimes^n_{j=1}\frac{I+\tilde{\lambda}_j Z_j }{2}\right) G_{\hat{O}}^{\dagger}$, where $G_{\hat{O}}$ is the free-fermionic unitary associated with $\hat{O}$.
\end{algorithm}

\begin{theorem}[Tomography of free-fermionic mixed states]\label{thsm:tomography}
Let $\rho$ be a free-fermionic state. For $\varepsilon \in (0,1)$ and $\delta \in (0,1]$, there exists a quantum learning algorithm (outlined in table \ref{alg:tomMIXED}) that, utilizing $N=\lceil 16 (n^{4} /\varepsilon^{2}) \log(4n^2/\delta)\rceil$ single-copies of the state $\rho$ learns an efficient representation of a state $\hat{\rho}$ such that
\begin{align}
    \norm{\hat{\rho}-\rho}_1\le \varepsilon,
\end{align}
with a probability of success at least $1-\delta$.
\end{theorem}

\begin{proof}
    Let $\varepsilon_{\mathrm{stat}} > 0$ be an accuracy parameter to be fixed later. By Lemma~\ref{le:samplecompAPPcommuting}, with $N \geq 8(n^3 / \varepsilon_{\mathrm{stat}}^2) \log(4n^2 / \delta)$ copies of the state, we can find a matrix $\hat{\Gamma}$ such that, with probability at least $1 - \delta$, it holds that $\norm{\hat{\Gamma} - \Gamma(\rho)}_2 < \varepsilon_{\mathrm{stat}}$.
In particular, $\hat{\Gamma}$ can be expressed in its normal form as $\hat{\Gamma} = \hat{O} \hat{\Lambda} \hat{O}^{T}$,
    where $\hat{O} \in \mathrm{O}(2n)$, and $\hat{\Lambda} = \bigoplus_{k=1}^{n} \begin{pmatrix}
            0 & \hat{\lambda}_k \\
            -\hat{\lambda}_k & 0
        \end{pmatrix}$.
    Now, as in Step 3 of the algorithm, for each $j \in [n]$, we set 
    \begin{align}
        \tilde{\lambda}_j = 
        \begin{cases}
            1, & \text{if } \hat{\lambda}_j > 1, \\
            \hat{\lambda}_j, & \text{otherwise}.
        \end{cases}
    \end{align}
    Let us now denote
    \begin{align}
        \tilde{\Lambda} \coloneqq \bigoplus_{i=1}^{n} \begin{pmatrix}
            0 & \tilde{\lambda}_i \\
            -\tilde{\lambda}_i & 0
        \end{pmatrix}
        \quad \text{and} \quad
        \Lambda \coloneqq \bigoplus_{i=1}^{n} \begin{pmatrix}
            0 & \lambda_i \\
            -\lambda_i & 0,
        \end{pmatrix}
    \end{align}
    where $\Lambda$ is the matrix containing the normal eigenvalues of the correlation matrix $\Gamma(\rho)$ (which are $\le 1$). We have
    \begin{align}
        \|\hat{\Lambda} - \tilde{\Lambda}\|_2 
        &\leq \|\hat{\Lambda} - \Lambda\|_2 \\
        \nonumber
        &\leq \|\hat{\Gamma} - \Gamma(\rho)\|_2 \\
        \nonumber
        &\leq \varepsilon_{\mathrm{stat}},
        \nonumber
    \end{align}
    where the first inequality uses the fact that $|\hat{\lambda}_j - \tilde{\lambda}_j| \leq |\hat{\lambda}_j - \lambda_j|$ (which holds because $\lambda_j \le 1$), the second inequality follows from the fact that $i\Gamma$ and $i\hat{\Gamma}$ are Hermitian, combined with the result $\|\mathrm{Eig}(A) - \mathrm{Eig}(B)\|_p \leq \|A - B\|_p$ for any Hermitian matrices $A$ and $B$ and any $p$-norms $\|\cdot\|_p$ (see Ref.~\cite{bhatia1996matrix}, Eq.~(IV.62)), where $\mathrm{Eig}(A)$ denotes the diagonal matrix of ordered eigenvalues of $A$.
    Thus, we can define the free-fermionic quantum state
    \begin{align}
        \hat{\rho} := G_{\hat{O}} \left( \bigotimes_{j=1}^{n} \frac{I + \tilde{\lambda}_j Z_j}{2} \right) G_{\hat{O}}^{\dagger},
    \end{align}
    where $G_{\hat{O}}$ is the free-fermionic unitary associated with $\hat{O}$. The correlation matrix associated with $\hat{\rho}$ is then $\Gamma(\hat{\rho}) = \hat{O} \tilde{\Lambda} \hat{O}^{\dagger}.$
    By Theorem~\ref{th:mixedtracedistance}, we have
    \begin{align}
        \|\hat{\rho} - \rho \|_1 
        &\leq \frac{1}{2} \|\Gamma(\hat{\rho}) - \Gamma(\rho)\|_1 \\
        \nonumber
        &\leq \sqrt{\frac{n}{2}} \|\Gamma(\hat{\rho}) - \Gamma(\rho)\|_2 \\
        \nonumber
        &\leq \sqrt{\frac{n}{2}} \left( \|\Gamma(\hat{\rho}) - \hat{\Gamma}\|_2 + \|\hat{\Gamma} - \Gamma(\rho)\|_2 \right) \\
        \nonumber
        &= \sqrt{\frac{n}{2}} \left( \|\tilde{\Lambda} - \hat{\Lambda}\|_2 + \|\hat{\Gamma} - \Gamma(\rho)\|_2 \right) \\
        \nonumber
        &\leq \sqrt{2n} \varepsilon_{\mathrm{stat}}.
        \nonumber
    \end{align}
    By choosing $\varepsilon_{\mathrm{stat}} = \varepsilon / \sqrt{2n}$, we conclude the proof.
\end{proof}
The previous algorithm for learning possibly mixed free-fermionic states uses only Gaussian operations and achieves a sample complexity of $\tilde{O}(n^4/\varepsilon^2)$. However, in the NISQ era~\cite{Preskill_2018}, it may be more practical to rely solely on single Pauli measurements. In this case, the sample complexity for learning mixed free-fermionic states would be $   N = O\!\left(\frac{n^5}{\varepsilon^2} \log\left(\frac{4n^2}{\delta}\right)\right)$ (as follows from Lemma~\ref{le:corrMatrixEst} and a similar analysis to that in the previous proof).

\subsection{Noise robustness of the learning algorithm}
\label{sec:noisetom}
In the previous section, we showed how to efficiently learn an unknown, possibly mixed, free-fermionic state. However, it is practically significant to consider whether the proposed learning algorithm is robust to slight deviations from the assumption that the state is free-fermionic. This is particularly relevant in experimental scenarios where one aims to prepare a free-fermionic state and learn it to verify that the quantum device is functioning correctly. In such cases, the experimental apparatus may be noisy, and it is reasonable to assume that the prepared state is not exactly free-fermionic but close to, if the noise is sufficiently low.
It is worth stressing that free-fermionic states are not so trivial to prepare experimentally, especially on a digital quantum computer, where \(O(n^2)\) two-qubit gates are required~\cite{dias2023classical} where $n$ is the number of qubits (or fermionic modes). However, their classical tractability may enable experimentalists to benchmark their devices by preparing specific free-fermionic quantum states, efficiently learning them, and verifying the results in a feedback loop.

Thus, the core question of this section is: Given the assumption that our state is \emph{close} to the set of free-fermionic states, can we still achieve efficient recovery guarantees for our learning algorithm? We analyze two notions of closeness by considering the \emph{relative entropy of non-Gaussianity}~\cite{Genoni2008} (in the first subsection) and the trace distance (in the second subsection).
We note that in the context of bosonic Gaussian state tomography, the tomography of a state with small relative entropy of non-Gaussianity has  recently been analyzed in Ref.~\cite{Mele2024bosonic}, also by our team.

\subsubsection*{Perturbations with respect to the relative entropy of non-Gaussianity}
\label{subsec:relative-entropy}
Here, in order to quantify the non-Gaussian (or non-free-fermionic) character of a quantum state, we employ the \emph{relative entropy of non-Gaussianity}~\cite{Genoni2008,Marian2013}. For any state $\rho$, the relative entropy of non-Gaussianity $d_{\mathcal{G}}(\rho)$ is defined as the minimum relative entropy 
\begin{align}
    d_{\mathcal{G}}(\rho) \coloneqq \min_{\sigma \in \mathcal{G}} S(\rho \| \sigma)\,
\end{align}
between $\rho$ and any Gaussian (i.e., free-fermionic) state, where $\mathcal{G}$ denotes the set of Gaussian states and
\begin{align}
    S(\rho \| \sigma) \coloneqq \Tr[\rho \log_2 \rho] - \Tr[\rho \log_2 \sigma]
\end{align}
represents the quantum relative entropy between $\rho$ and $\sigma$. 
The relative entropy of non-Gaussianity possesses several desirable properties, making it a meaningful measure of non-Gaussianity. Specifically, $d_{G}(\rho)$ is faithful: $d_{G}(\rho) \ge 0$ and it equals zero if and only if $\rho$ is Gaussian. Importantly, the minimum in the definition of $d_{\mathcal{G}}(\rho)$ is attained by the \emph{Gaussianification} of $\rho$. The Gaussianification $G(\rho)$ of a state $\rho$ is the free-fermionic state with the same correlation matrix as $\rho$. The proofs of these facts can be found in Refs.~\cite{GaussianChannels,Genoni2008,Marian2013} for bosonic Gaussian states, and they generalize seamlessly to fermions, utilizing the representation of fermionic Gaussian states as Gibbs states of quadratic Hamiltonians in the Majorana operators. 
For instance, see Refs.~\cite{PhysRevLett.119.020501,Gluza_2016} for discussions on Gaussianification in the fermionic context, or Appendix B of Ref.\cite{Lumia:2023ofv} for a list of properties of the relative entropy of fermionic non-Gaussianity.
Thus, the relative entropy of non-Gaussianity of a state $\rho$ is given by~\cite{Marian2013}
\begin{align}\label{charact_rel_entropy_non_gauss}
    d_{\mathcal{G}}(\rho) = S\!\left(\rho \| G(\rho)\right)\,,
\end{align}
where $G(\rho)$ is the Gaussianification of $\rho$.

We now proceed to analyze the robustness of our tomography algorithm for free-fermionic states. First, we highlight the following observation:
\begin{remark}\label{remark_gauss}
    Let $\rho$ be a quantum state  (possibly non-free fermionic). The algorithm designed for learning free-fermionic states in Table~\ref{alg:tomMIXED} (see Theorem~\ref{thsm:tomography}) effectively learns the Gaussianification $G(\rho)$. Thus, $\lceil 16 (n^{4} /\varepsilon^{2}) \log(4n^2/\delta)\rceil$ copies of $\rho$ suffice to build a classical description of a free-fermionic state $\hat{\rho}$ such that $\|\hat{\rho} - G(\rho)\|_1 \le \varepsilon $ with probability at least $1 - \delta$.
\end{remark}

Now we show that if the relative entropy of non-Gaussianity is sufficiently small, then we can still apply our learning algorithm.

\begin{proposition}[Robustness to small deviations in relative entropy of non-Gaussianity]\label{thm_robustness_gaussian_states}
    Let $\varepsilon, \delta \in (0,1)$. Let $\rho$ be an unknown state such that its relative entropy of non-Gaussianity satisfies $d_{\mathcal{G}}(\rho) \le \varepsilon^2$. Then, $O\!\left( (n^4 /\varepsilon^{4}) \log(4n^2/\delta) \right)$ copies of $\rho$ suffice to build a classical description of a free-fermionic state $\hat{\rho}$ such that $\|\hat{\rho} - \rho\|_1 \le \varepsilon $ with probability at least $1 - \delta$.
\end{proposition}

\begin{proof}
    By using Remark~\ref{remark_gauss}, $O\!\left(  (n^{4} /\varepsilon^{2}) \log(4n^2/\delta)  \right)$ copies of $\rho$ suffice to build a classical description of a free-fermionic state $\hat{\rho}$ such that $\|\hat{\rho} - G(\rho)\|_1 \le \left(1 - \sqrt{2\ln 2}\right)\varepsilon$ with probability at least $1 - \delta$, where $G(\rho)$ denotes the Gaussianification of $\rho$. If this event happens, then 
    \begin{align}
        \|\rho - \hat{\rho}\|_1 &\le \|G(\rho) - \hat{\rho}\|_1 + \|G(\rho) - \rho\|_1 \\
        \nonumber
        &\le \left(1 - \sqrt{2\ln 2}\right)\varepsilon + \|G(\rho) - \rho\|_1 \\
         \nonumber
        &\le \left(1 - \sqrt{2\ln 2}\right)\varepsilon + \sqrt{2\ln 2} \sqrt{S(\rho \| G(\rho))} \\
         \nonumber
        &= \left(1 - \sqrt{2\ln 2}\right)\varepsilon + \sqrt{2\ln 2} \sqrt{d_{\mathcal{G}}(\rho)} \\
         \nonumber
        &\le \varepsilon\,.
         \nonumber
    \end{align}
    Here, in the third step we employed the quantum Pinsker inequality~\cite[Theorem 11.9.1]{MARK}, which states that for any $\tau$ and $\sigma$, the trace distance can be upper bounded in terms of the relative entropy as $\frac{1}{2}\|\tau - \sigma\|_1 \le \sqrt{\frac{\ln 2}{2} S(\tau \| \sigma)}$. Finally, in the fourth step, we have used the characterization of the relative entropy of non-Gaussianity in~\eqref{charact_rel_entropy_non_gauss}.
\end{proof}

\subsubsection*{Perturbations with respect to the trace distance}
\label{subsec:trace-distance}
We now discuss the case in which the unknown state is promised to be close in trace-distance to the set of free-fermionic states.
\begin{proposition}[Robustness to small deviations in trace distance from the set of free-fermionic states]\label{thm_robustness_gaussian_states2}
   Let $\varepsilon, \delta \in (0,1)$. Let $\rho$ be a quantum state such that $\min_{\sigma \in \mathcal{G}}\|\rho-\sigma\|_1 \le \frac{\varepsilon}{3 n}$. Then, $O\!\left( (n^{4} /\varepsilon^{2}) \log(4n^2/\delta) \right)$ copies of $\rho$ suffice to build a classical description of a free-fermionic state $\hat{\rho}$ such that $\|\hat{\rho} - \rho\|_1 \le \varepsilon $ with probability at least $1 - \delta$.
\end{proposition}
\begin{proof}
Let $\varepsilon_1=\varepsilon/3n$.
By assumption, there exists a free-fermionic state $\rho_{\mathrm{free}}$ such that $\|\rho - \rho_{\mathrm{free}}\|_1 \le \varepsilon_1$. Utilizing Proposition~\ref{le:tracedistancelowerboundcormatrix}, we obtain
\begin{align}
    \|\Gamma(\rho) - \Gamma(\rho_{\mathrm{free}})\|_{\infty} \le \varepsilon_1.
    \label{eq:crucnoi}
\end{align}
By using Remark~\ref{remark_gauss}, $O\!\left(  (n^{4} /\varepsilon^{4}) \log(4n^2/\delta)  \right)$ copies of $\rho$ suffice to build a classical description of a free-fermionic state $\hat{\rho}$ such that $\|\hat{\rho} - G(\rho)\|_1 \le \varepsilon/3$ with probability at least $1 - \delta$, where $G(\rho)$ denotes the Gaussianification of $\rho$.
Subsequently, we have
\begin{align}
    \|\hat{\rho} - \rho\|_1 &\le \|\hat{\rho} - G(\rho)\|_1 + \|G(\rho) - \rho_{\mathrm{free}}\|_1 + \|\rho_{\mathrm{free}} - \rho\|_1 \\
     \nonumber
    &\le \frac{\varepsilon}{3} + \frac{1}{2}\|\Gamma(G(\rho)) - \Gamma(\rho_{\mathrm{free}})\|_{1} + \varepsilon_1 \\
     \nonumber
    &\le \frac{\varepsilon}{3} + n\|\Gamma(\rho) - \Gamma(\rho_{\mathrm{free}})\|_{\infty} + \varepsilon_1 \\
     \nonumber
    &\le \frac{\varepsilon}{3} + n\varepsilon_1 + \varepsilon_1\\
     \nonumber
    &\le \varepsilon,
\end{align}
where in the first inequality we employed the triangle inequality, in the second step we utilized Theorem~\ref{th:mixedtracedistance} and the assumption that the state $\rho$ is $\varepsilon_1$-close to a Gaussian state $\rho_{\mathrm{free}}$, in the third step we used that $\Gamma(\rho)=\Gamma(G(\rho))$ together with the relation between trace-norm and operator norm, and in the fourth step we have used Eq.~\eqref{eq:crucnoi}.
\end{proof}
The previous Proposition implies that if a given state is $O(\frac{\varepsilon}{n})$ close to the set of free-fermionic states, then our learning algorithm can still be reliably applied.

%\bibliographystyle{naturemag}

%\bibliography{ref}

\end{document}